\documentclass[oneside,english]{amsart}
\usepackage[latin9]{inputenc}
\usepackage{geometry}
\usepackage{color}
\definecolor{note_fontcolor}{rgb}{1, 0, 0}
\usepackage{mathtools}
\usepackage{enumitem}
\usepackage{amstext}
\usepackage{amsthm}
\usepackage{amssymb}
\usepackage[foot]{amsaddr}
\usepackage[all]{xy}

\makeatletter

\numberwithin{equation}{section}
\numberwithin{figure}{section}
\theoremstyle{plain}
\newtheorem{thm}{\protect\theoremname}
\theoremstyle{definition}
\newtheorem{defn}{\protect\definitionname}[section]
\theoremstyle{remark}
\newtheorem{rem}[defn]{\protect\remarkname}
\theoremstyle{plain}
\newtheorem{prop}[defn]{\protect\propositionname}
\theoremstyle{remark}
\newtheorem{claim}[defn]{\protect\claimname}
\theoremstyle{plain}
\newtheorem{fact}[defn]{\protect\factname}
\theoremstyle{definition}
\newtheorem{problem}[defn]{\protect\problemname}
\theoremstyle{plain}
\newtheorem{lem}[defn]{\protect\lemmaname}
\theoremstyle{plain}
\newtheorem{cor}[defn]{\protect\corollaryname}
\theoremstyle{plain}
\newtheorem{exm}[defn]{\protect\examplename}

\newtheorem{innercustomthm}{Theorem}
\newenvironment{customthm}[1]
{\renewcommand\theinnercustomthm{#1}\innercustomthm}
{\endinnercustomthm}

\usepackage[backref=page,colorlinks,citecolor=blue,bookmarks=true]{hyperref}\hypersetup{linkcolor=blue,  citecolor=blue, urlcolor=blue}
\usepackage{amsrefs}

\usepackage{fullpage}

\usepackage{algorithm}
\usepackage[noend]{algpseudocode}
\usepackage{caption}

\let\originalleft\left
\let\originalright\right
\renewcommand{\left}{\mathopen{}\mathclose\bgroup\originalleft}
\renewcommand{\right}{\aftergroup\egroup\originalright}

\date{}

\makeatother

\usepackage{babel}
\providecommand{\claimname}{Claim}
\providecommand{\corollaryname}{Corollary}
\providecommand{\definitionname}{Definition}
\providecommand{\factname}{Fact}
\providecommand{\lemmaname}{Lemma}
\providecommand{\problemname}{Problem}
\providecommand{\propositionname}{Proposition}
\providecommand{\remarkname}{Remark}
\providecommand{\theoremname}{Theorem}
\providecommand{\examplename}{Example}

\DeclareMathOperator{\SL}{SL}

\DeclareMathOperator{\Sym}{Sym}

\let\U\relax
\DeclareMathOperator{\U}{U}

\DeclareMathOperator{\id}{id}

\DeclareMathOperator{\T}{(T)}
\DeclareMathOperator{\ptau}{(\tau)}

\DeclareMathOperator{\SAS}{SAS}%
\DeclareMathOperator{\LSM}{LSM}%
\DeclareMathOperator{\TV}{TV}%
\DeclareMathOperator{\SOL}{Sol}%
\DeclareMathOperator{\FGSOL}{FGSol}%
\DeclareMathOperator{\GSOL}{GSol}%
\DeclareMathOperator{\HSOL}{HSol}%

\DeclareMathOperator{\res}{res}%
\DeclareMathOperator{\stab}{Stab}%
\DeclareMathOperator{\pstab}{PStab}%
\DeclareMathOperator{\last}{last}%
\DeclareMathOperator{\suff}{suff}%

\DeclareMathOperator{\paths}{Paths}
\DeclareMathOperator{\spaths}{Simple-Paths}
\DeclareMathOperator{\inv}{inv}
\DeclareMathOperator{\bp}{Bipaths}
\DeclareMathOperator{\Binomial}{Binomial}
\DeclareMathOperator{\Var}{Var}
\DeclareMathOperator{\flex}{flex}
\DeclareMathOperator{\BS}{BS}
\DeclareMathOperator{\poly}{poly}
\DeclareMathOperator{\comm}{comm}
\DeclareMathOperator{\FinSubsets}{FinSubsets}
\DeclareMathOperator{\Subsets}{Subsets}
\DeclareMathOperator{\TotalSize}{TotalSize}

\DeclareMathOperator{\IRS}{IRS}

\DeclareMathOperator{\ham}{H}
\DeclareMathOperator{\GL}{GL}
\DeclareMathOperator{\Bad}{Bad}
\DeclareMathOperator{\Prob}{Prob}
\DeclareMathOperator{\Emp}{Empirical}

\newcommand{\eq}[1]{\mathsf{#1}}%
\newcommand{\ol}[1]{\overline{#1}}%

\newcommand{\lla}{\langle\!\langle}%
\newcommand{\rra}{\rangle\!\rangle}%
\newcommand{\PG}{\mathcal{PG}}%

\newcommand{\eps}{\varepsilon}

\newcommand{\EE}{\mathbb{E}}

\newcommand{\NN}{\mathbb{N}}

\newcommand{\RR}{\mathbb{R}}

\newcommand{\ZZ}{\mathbb{Z}}
\newcommand{\calA}{\mathcal{A}}
\newcommand{\calB}{\mathcal{B}}

\newcommand{\calD}{\mathcal{D}}

\newcommand{\calG}{\mathcal{G}}
\newcommand{\calH}{\mathcal{H}}

\newcommand{\calM}{\mathcal{M}}

\newcommand{\calR}{\mathcal{R}}

\newcommand{\cM}{\mathcal{M}}

\newcommand{\rightedge}{{\hspace{-0.05em}\longrightarrow\hspace{-0.05em}}}%

\usepackage[dvipsnames]{xcolor}

\begin{document}

\title{Testability of relations between permutations}

\author[O.\ Becker]{Oren Becker}
\address{Oren Becker,
Centre for Mathematical Sciences,
Wilberforce Road, Cambridge CB3 0WA, United Kingdom}
\email{oren.becker@gmail.com}

\author[A.\ Lubotzky]{Alexander Lubotzky}
\address{Alexander Lubotzky,
	Einstein Institute of Mathematics,
	The Hebrew University, Jerusalem 91904, Israel}
\email{alex.lubotzky@mail.huji.ac.il}

\author[J.\ Mosheiff]{Jonathan Mosheiff}
\address{Jonathan Mosheiff,
	Carnegie Mellon University,
	5000 Forbes Avenue, Pittsburgh, PA, USA
	}
\email{jmosheif@cs.cmu.edu}

\begin{abstract}
We initiate the study of property testing problems concerning relations between permutations.
In such problems, the input is a tuple $\left(\sigma_{1},\dotsc,\sigma_{d}\right)$ of permutations on $\{1,\dotsc,n\}$,
and one wishes to determine whether this tuple satisfies a certain system
of relations $E$, or is far from every tuple that satisfies $E$.
If this computational problem can be
solved by querying only a small number of entries of the given permutations,
we say that $E$ is \emph{testable}. For example, when $d=2$ and
$E$ consists of the single relation $\mathsf{XY=YX}$, this corresponds
to testing whether $\sigma_{1}\sigma_{2}=\sigma_{2}\sigma_{1}$, where $\sigma_1\sigma_2$ and $\sigma_2\sigma_1$ denote composition of permutations.

We define a collection of graphs, naturally associated with the system $E$, that encodes all
the information relevant to the testability of $E$.
We then prove two theorems that provide criteria for testability and non-testability in terms
of expansion properties of these graphs.
By virtue of a deep connection with group theory, both theorems are applicable to wide classes of systems of relations.

In addition, we formulate the well-studied group-theoretic notion of stability
in permutations as a special case of the testability notion above, interpret all previous works
on stability as testability results, survey previous results on stability from a computational
perspective, and describe many directions for future research on stability and testability.
\end{abstract}

\maketitle
\textbf{}

\section{\label{sec:intro}Introduction}

In this paper we study the testability of relations between permutations.
We consider problems where several permutations are given in a black box form (e.g., as circuits or oracles), and one wishes to determine whether they satisfy a fixed system of relations $E$ or are far from doing so. For example, suppose that $E$ consists of the single formal relation $\eq{XY=YX}$. The computational problem corresponding to $E$ is to test whether two given permutations $A$ and $B$, over the same finite set, commute. More precisely, the problem is to distinguish between the following two cases (i) $AB=BA$, and (ii) $A$ and $B$ are $\eps$-far in the normalized Hamming metric (see \eqref{eq:intro-hamming-metric} and Definition \ref{def:solutions}) from every pair of permutations $A'$, $B'$ that satisfies $A'B'=B'A'$. 

Testing whether $AB=BA$ was shown in \cite{BeckerMosheiff} to be achievable by an algorithm whose query and time complexities are both polynomial in $\frac 1\eps$. In particular, we say that the system $E=\{\eq{XY=YX}\}$ is \emph{testable} since it has a testing algorithm whose query complexity depends only on $\eps$, and not\footnote{More generally, testing algorithms with query complexity sublinear in $n$ are a fascinating topic for future research, but the focus of the present paper is on testing with a constant number of queries. See also Section \ref{subsec:QueriesDependOnN}.} on the size $n$ of the domain $\left\{1,\dotsc,n\right\}$ of $A$ and $B$. Here, \emph{query complexity} counts queries of the form ``what is $A(x)$'' or ``what is $B(x)$'' for $x$ in the domain of the permutations. 

On the other hand, consider the system of relations $E = \{\eq{XZ=ZX}, \eq{YZ=ZY}\}$.
Testing $E$ amounts to testing, for three given permutations $A$, $B$ and $C$, whether $C$ commutes with both $A$ and $B$. Theorem \ref{thm:LSMUniversal} in the present paper, together with \cite{Ioana}, imply that this system is not testable (see Example \ref{example:BS-stability}). In other words, the query complexity of every testing algorithm for $E$ must depend on $n$ and not just on $\eps$. Similarly, the system $E=\{\eq{XY^2=Y^2X}\}$ is also not testable (this is also discussed in Example \ref{example:BS-stability}), even though it is superficially similar to  the system $\{\eq{XY=YX}\}$.

In this work we go far beyond the systems of relations in the examples above.
We establish a framework and initiate a systematic study of the testability of systems of relations between permutations (also known as \emph{equations in permutations}). Given such a system $E$, we naturally associate with it a certain infinite family of graphs $\GSOL_E$ (see Section \ref{sec:GraphView}), which contains both finite and infinite graphs. Our main results give criteria for the testability and non-testability of $E$ in terms of expansion properties of $\GSOL_E$.

\begin{defn}[Measures of expansion]
	Let $G$ be a graph of bounded degree, with vertex set $V$. The \emph{isoperimetric constant} of $G$ is 
	$$\iota(G) = \inf\left\{\frac{|E(X,V\setminus X)|}{|X|}\mid X\subseteq V\text{ is finite and nonempty}\right\}.$$
	If $G$ is finite, its \emph{Cheeger constant} is
	$$h(G) = \inf\left\{\frac{|E(X,V\setminus X)|}{|X|}\mid X\subseteq V\text{ and }1\le |X|\le \frac{|V|}2\right\}.$$
	Here, $E(X,Y)$ denotes the set of edges between the sets $X$ and $Y$.	
\end{defn}
Note that if $G$ is finite then $\iota(G)$ is trivially $0$ since the set $X=V$ has no expansion. However, some infinite graphs, such as the $d$-regular tree ($d\ge 3$), have a positive isoperimetric constant. 

Our main positive theorem states that $E$ is testable whenever $\GSOL_E$ is nonexpanding in the isoperimetric sense.
\begin{thm}[Main positive theorem]\label{thm:MainPositive}
	If $\iota(G) = 0$ for every $G\in \GSOL_E$ then $E$ is testable.
\end{thm}
The notion of a \emph{testable system of relations} is defined formally in Definition \ref{def:testable-equations} in Section \ref{subsec:framework}. The family of graphs $\GSOL_E$ is defined in Section \ref{sec:GraphView}.

The aforementioned testability of $\{\eq{XY=YX}\}$ is a narrow special case\footnote{
For $E=\{\eq{XY=YX}\}$,
it is not hard to verify directly that $\iota(G)=0$ for all $G\in\GSOL_E$,
but in fact it suffices to verify that $\iota(C)=0$,
where $C$ is the infinite grid in the plane.
This suffices because $C$ is the Cayley graph of the group $\ZZ^2$ (see Section \ref{sec:GroupTheory}).} of Theorem \ref{thm:MainPositive}.

Our main negative theorem states that nonexpansion in the Cheeger sense is a necessary condition for testability.
\begin{thm}[Main negative theorem]\label{thm:MainNegative}
	Let $\FGSOL_E$ denote the set of graphs in $\GSOL_E$ that are finite and connected. If $\FGSOL_E$ is infinite and $\inf\{h(G)\mid G\in \FGSOL_E\} > 0$ then $E$ is non-testable.
\end{thm}

In Section \ref{sec:GroupTheory} we will show how to associate a group\footnote{
The family $\GSOL_E$ is in fact the family of graphs whose connected
components are Schreier graphs of the group $\Gamma(E)$.
}
$\Gamma(E)$ with the system $E$.
The expansion and isoperimetric constants appearing in
Theorems \ref{thm:MainPositive} and \ref{thm:MainNegative} have been studied
extensively in the framework of group theory, leading to numerous examples where
the theorems are applicable, some of which are discussed in Section \ref{sec:GroupTheory} and Appendix \ref{app:Equations}.

We note that, while Theorems \ref{thm:MainPositive} and \ref{thm:MainNegative} apply to many systems of relations, they do not provide a complete classification. For example, the system consisting of the single relation $\eq{XY^2 = Y^3X}$ is known to satisfy neither the hypothesis of Theorem \ref{thm:MainPositive} nor that of Theorem \ref{thm:MainNegative}, and the question of its testability remains open (see Problem \ref{prob:BS}). In the spirit of well-known classification theorems, such as those concerning constraint-satisfaction problems \cites{Chen09,Zhuk20} and efficient testability of the $H$-freeness property of a graph \cite{Alon02}, a prominent objective of the present line of research is obtaining a complete characterization of testable systems of relations. We elaborate on this goal in Section \ref{subsec:characterization}. We now turn to providing the necessary framework for a precise statement of our results. 

\subsection{A framework for systems of relations between permutations}\label{subsec:framework}
Fix a finite alphabet $S=\left\{ s_{1},\dotsc,s_{d}\right\} $
throughout the introduction. Let $S^{-}=\left\{ s_{1}^{-1},\dotsc,s_{d}^{-1}\right\} $
be the set of formal inverses of the letters in $S$, and write $S^{\pm}=S\cup S^{-}$.
\begin{defn}
	A \emph{relation} is a formal equation of the form $w_{i,1}=w_{i,2}$, where $w_{i,j}$ is a word over $S^{\pm}$.
	A \emph{system of relations} is a finite set of relations.
\end{defn}

Let $\Sym(n)$ denote the group of all permutations on $[n]\coloneqq\{1\,\dotsc,n\}$. Fix a system of relations $E=\left\{ w_{i,1}=w_{i,2}\right\} _{i=1}^{r}$.
We think of the letters in $S$ as variables, and study the space of assignments
${s_1 \leftarrow \sigma_1,\dotsc, s_d \leftarrow \sigma_d}$ that satisfy all of the relations in $E$, where each $\sigma_i$
is a permutation in $\Sym(n)$.
In other words, we study the space of simultaneous solutions for $E$ inside $\left(\Sym\left(n\right)\right)^{d}$.
More precisely, we are interested in testability problems related
to this space of solutions. As an example, to fit the relation $\eq{XY=YX}$
into this framework, we set $d=2$, $S=\left\{ s_{1},s_{2}\right\} $,
$r=1$, $w_{1,1}=s_{1}s_{2}$ and $w_{1,2}=s_{2}s_{1}$. Thus $E=\left\{ s_{1}s_{2}=s_{2}s_{1}\right\} $
in this case. For notational convenience we sometimes denote $\eq X=s_{1}$,
$\eq Y=s_{2}$ and $\eq Z=s_{3}$.
\begin{defn}
	\label{def:wordNotation}For a word $w$ over $S^{\pm}$ and a tuple
	$\overline{\sigma}=\left(\sigma_{1},\dotsc,\sigma_{d}\right)\in\left(\Sym\left(n\right)\right)^{d}$,
	$n\in\NN$, write $w\left(\overline{\sigma}\right)$ for the permutation
	that results from applying the assignment $s_{j}\leftarrow \sigma_{j}$
	to the word $w$.
\end{defn}

\begin{exm} 
	If $S=\left\{ \eq X,\eq Y\right\} $, $w=\eq{XYX^{-1}Y^{-1}}$
	and $\overline{\sigma}=\left(\left(1\,2\,3\right),\left(1\,2\right)\right)\in\left(\Sym\left(3\right)\right)^{2}$
	(where the permutations are given in cycle notation), then $w\left(\overline{\sigma}\right)=\left(1\,2\,3\right)\left(1\,2\right)\left(1\,2\,3\right)^{-1}\left(1\,2\right)^{-1}=\left(1\,3\,2\right)$.
\end{exm}

Let $d_{n}^{\ham}$
denote\footnote{We shall omit the subscript $n$ when it is clear from the context.} the normalized Hamming metric on $\Sym\left(n\right)$.
That is,
\begin{equation}
	d_n^{\ham}\left(\sigma,\tau\right) = d^{\ham}\left(\sigma,\tau\right)=\frac{1}{n}\left|\left\{ x\in\left[n\right]\mid\sigma\left(x\right)\neq\tau\left(x\right)\right\} \right|\quad\forall\sigma,\tau\in\Sym\left(n\right)\,\,\text{.}\label{eq:intro-hamming-metric}
\end{equation}

\begin{defn}
	\label{def:solutions}Let $n\in\NN$. We say that $\overline{\sigma}\in\left(\Sym\left(n\right)\right)^{d}$
	is a \emph{solution} for $E$ in $\left(\Sym\left(n\right)\right)^{d}$ (or that $\overline \sigma$ \emph{satisfies} $E$)
	if $w_{i,1}\left(\overline{\sigma}\right)=w_{i,2}\left(\overline{\sigma}\right)$
	for each $1\leq i\leq r$, and write $\SOL_{E}\left(n\right)$ for
	the set of solutions for $E$ in $\left(\Sym\left(n\right)\right)^{d}$.
	For $\eps\ge0$, let
	\[
	\SOL_{E}^{\geq\eps}\left(n\right)=\left\{ \left(\sigma_{1},\dotsc,\sigma_{d}\right)\in\left(\Sym\left(n\right)\right)^{d}\mid\text{\ensuremath{\sum_{j=1}^{d}d^{\ham}\left(\sigma_{j},\tau_{j}\right)\geq\eps}\ensuremath{\quad}\ensuremath{\forall\left(\tau_{1},\dotsc,\tau_{d}\right)\in}}\SOL_{E}\left(n\right)\right\} \,\,\text{.}
	\]
\end{defn}

In the example where $S=\left\{ \eq X,\eq Y\right\} $ and $E=\left\{ \eq{XY=YX}\right\} $,
the space $\SOL_{E}\left(n\right)$ is the set of pairs $\left(A,B\right)\in\left(\Sym\left(n\right)\right)^{2}$
such that $AB=BA$, and the space $\SOL_{E}^{\geq\eps}\left(n\right)$
is the set of all pairs $(A,B)\in\left(\Sym\left(n\right)\right)^{2}$ that satisfy $d^H(A,A') + d^H(B,B') \ge \eps$ whenever $(A',B')\in\left(\Sym\left(n\right)\right)^{2}$ is a commuting pair.

The following is the main novel definition of this paper:
\begin{defn}
	[Testable system of relations]\label{def:testable-equations}An algorithm
	$\calM$ that takes $n\in\NN$ and a tuple $\overline{\sigma}=\left(\sigma_{1},\dotsc,\sigma_{d}\right)\in\left(\Sym\left(n\right)\right)^{d}$
	as input is an \emph{$\left(\eps,q\right)$-tester} for $E$ if it
	satisfies the following conditions:
	\begin{itemize}
		\item \textbf{Completeness:} if $\overline{\sigma}\in\SOL_{E}\left(n\right)$,
		the algorithm accepts with probability at least $0.99$.
		\item \textbf{$\eps$-soundness:} if $\overline{\sigma}\in\SOL_{E}^{\ge\eps}\left(n\right)$,
		the algorithm rejects with probability at least $0.99$.
		\item \textbf{Query efficiency:} the algorithm is only allowed to query $q$
		entries of $\sigma_{1},\dotsc,\sigma_{d}$ and their inverses.
	\end{itemize}
	If for every $\eps>0$ there are $q=q\left(\eps\right)\in\NN$ and
	an $\left(\eps,q\right)$-tester $\calM_{\eps}$ for $E=\left\{ w_{i,1}=w_{i,2}\right\} _{i=1}^{r}$
	then we say that $E$ is \emph{$q\left(\eps\right)$-testable} (or
	just \emph{testable}) and that $\eps\mapsto\calM_{\eps}$ is a \emph{family
		of testers} for $E$.
\end{defn}

Note that in Definition \ref{def:testable-equations} we allow the
algorithm to have oracle access both to the entries of the given permutations
and to the entries of their inverses. Oracle access to inverses of
permutations is in fact not always necessary. Appendix \ref{app:AvoidingInverses}
explains what can be done to circumvent the need to sample inverses.

In the sequel, it will be convenient to work with sets of \emph{relators} rather than systems of relations, as explained below.
A word $w$ over the alphabet $S^{\pm}$ is \emph{reduced} if it does
not contain any subword of the form $s_{i}s_{i}^{-1}$ or $s_{i}^{-1}s_{i}$,
$1\leq i\leq d$. Write $F_{S}$ for the set of reduced words over
$S^{\pm}$ (in Appendix \ref{app:FreeGroupAndPresentations} we
recall that $F_{S}$ has a natural group structure, making it the
\emph{free group} on $S$). Every word $w$ over $S^{\pm}$ is equivalent to a unique reduced word, obtained from $w$ by repeatedly removing subwords of the form $s_{i}s_{i}^{-1}$ or $s_{i}^{-1}s_{i}$,
$1\leq i\leq d$.

Let $w$ be a word over $S^{\pm}$. Write $w=s_{i_1}^{e_1}\dotsm s_{i_\ell}^{e_\ell}$, where $\ell\geq 0$, $e_j\in\{+1,-1\}$ and $1\leq i_j\leq d$ for all $1\leq j\leq \ell$. Then $w$ has a formal inverse $w^{-1}\coloneqq s_{i_\ell}^{-e_\ell} \dotsm s_{i_1}^{-e_1}$.

A system of relations $E=\left\{ w_{i,1}=w_{i,2}\right\} _{i=1}^{r}$ gives rise to a subset $R_E$ of $F_S$, defined to be the set of reduced words equivalent to $w_{1,2}^{-1}w_{1,1},\dotsc,w_{r,2}^{-1}w_{r,1}$.
We say that $R_E$ is the \emph{set of relators} corresponding to $E$ (in general, a set of relators is just a subset of $F_S$).
For example, if $E=\{\eq{XZ=ZX}, \eq{YZ=ZY}\}$ then $R_E=\{\eq{X^{-1}Z^{-1}XZ},\eq{Y^{-1}Z^{-1}YZ}\}$.
Clearly, the system of relations $E$ is equivalent to the system of relations $E'=\left\{w=1\mid w\in R_E\right\}$. Indeed, $\SOL_E(n) = \SOL_{E'}(n)$ and $\SOL_E^{\ge \eps}(n) = \SOL_{E'}^{\ge \eps}(n)$ for all $n\geq 1$ and $\eps>0$.
It is generally more convenient to work with $R_E$ rather than directly with $E$.

\subsection{A graph-theoretic view\label{sec:GraphView}}
It will be beneficial to encode a tuple of permutations as an $S$-graph, defined below. 
\begin{defn}
	An $S$-graph is an edge-labelled directed (not necessarily finite) graph\footnote{
	$S$-graphs are allowed to have self-loops and multiple edges, but as follows from
	the definition, two different edges directed from vertex $x$ to vertex $y$ must have
	distinct labels.}, where each directed
	edge is labelled by an element of $S$, and each vertex has exactly
	one outgoing and one incoming edge labelled $s$ for each $s\in S$.	
	
	Given $n\in \NN$, write $\calG_{S}\left(n\right)$ for the set of $S$-graphs on the
	vertex set $\left[n\right]$. In particular, $\calG_S(n)$ consists solely of finite $S$-graphs.
	
	When referring to the connected components of an $S$-graph
	$G$, or to whether or not $G$ is connected, we disregard edge orientation
	and labels, and treat $G$ as an undirected graph.
\end{defn}

Let $G$ be an $S$-graph. For $s\in S$ and a vertex $x$,
we write $s_{G}x=y$ for the unique vertex $y$ such that $x\overset{s}{\rightedge}y$
is an edge in $G$. We also define $s_{G}^{-1}y=x$. When there is no ambiguity about the graph in context, we write $sx$ for $s_Gx$ (for $s\in S^\pm$). For a word $w=w_{t}\dotsm w_{1}$
($w_{i}\in S^{\pm}$) we recursively define $w_{G}x=w_{t}\left(w_{t-1}\dotsc w_{1}x\right)$
whenever $t>1$. That is, $w_{G}x$ (or just $wx$) is the final vertex in the path
that starts at $x$ and follows $t$ edges labelled according to $w$.

We next show how to encode a tuple $\overline{\sigma}=\left(\sigma_{1},\dotsc,\sigma_{d}\right)\in\left(\Sym\left(n\right)\right)^{d}$
as an $S$-graph $G_{\overline \sigma}$. Let $G_{\overline{\sigma}}\in \calG_S(n)$ denote the $S$-graph with vertex
set $\left[n\right]$ and edge set $\left\{ x\overset{s_{i}}{\rightedge}\sigma_{i}x\mid x\in\left[n\right],1\leq i\leq d\right\} $. Clearly, the map $\overline{\sigma}\mapsto G_{\overline{\sigma}}\colon\left(\Sym\left(n\right)\right)^{d}\to\calG_{S}\left(n\right)$
is a bijection.

Let $E = \{w_{i,1}=w_{i,2}\}_{i=1}^r$ be a system of relations over $S^{\pm}$. An $S$-graph $G$ is said to belong to the class $\GSOL_E$ if $(w_{i,1})_Gx = (w_{i,2})_Gx$ for every $1\le i\le r$ and vertex $x$. Equivalently, $G\in \GSOL_E$ if $w_Gx=x$ holds for all $w\in R_E$ and every vertex $x$. In this case, we say that $G$ \emph{satisfies} $E$.

\begin{exm}
Let $S=\{\eq{X},\eq{Y}\}$ and $E=\{\eq{XY=YX}\}$. Let $m,n\geq 1$, and let $G$ be the graph on the vertex set $V=\{0,\dotsc,m-1\}\times\{0,\dotsc,n-1\}$, with the following edges: for each $v=(a,b)\in V$, the $\eq{X}$-labelled edge originating from $v$ terminates at $(a+1,b)$, and the $\eq{Y}$-labelled edge originating from $v$ terminates at $(a,b+1)$ (where the addition is taken modulo $m$ and $n$, respectively).
Then $G$ belongs to $\GSOL_E$ (note that $G$ can be embedded on a torus).
\end{exm}

Denote $\GSOL_E(n) = \GSOL_E \cap \calG_S(n)$. Given $\overline\sigma \in \left(\Sym(n)\right)^d$, note that $\overline \sigma \in \SOL_E(n)$ if and only $G_{\overline \sigma} \in \GSOL_E(n)$. In other words, $G_{\overline \sigma}$ satisfies $E$ if and only if $\overline \sigma$ does the same. Importantly, the inclusion $\bigcup_{n\in \NN}\GSOL_E(n)\subset \GSOL_E$ is strict since the latter set contains infinite $S$-graphs.

At this point, all notions in the statements of Theorems \ref{thm:MainPositive} and \ref{thm:MainNegative} have been fully defined. As these theorems demonstrate, the power of the correspondence $\SOL_E\left(n\right) \to \GSOL_E\left(n\right)$, $n\in\NN$, manifests in a connection between the testability of $E$ and the geometry of the
(finite and infinite) graphs in $\GSOL_E$. 

As we will be working extensively with the graph encoding of a tuple of permutations, it will be convenient to think of a tester as an algorithm whose input lies in $\calG_S(n)$, rather than $\left(\Sym(n)\right)^d$. We naturally define\footnote{In words, $d^{\ham}(G,G')$ counts (with a $\frac 1n$ normalization factor) pairs consisting of a vertex $x$ and letter $s$, such that the respective $s$-labelled edges originating from $x$, in the graphs $G$ and $G'$, disagree on their target vertex. Recall that the vertex set of both $G$ and $G'$ is $\{1,\dotsc, n\}$.}
\begin{equation}\label{eq:graphDistance}
d^{\ham}\left(G,G'\right)=\sum_{s\in S}\frac{1}{n}\left|\left\{ x\in\left[n\right]\mid s_{G}x\ne s_{G'}x\right\} \right|\quad\forall G,G'\in\calG_S(n)
\end{equation}
and
\[
\GSOL_{E}^{\geq\eps}\left(n\right)=\left\{ G\in\calG_{S}\left(n\right)\mid d^{\ham}\left(G,G'\right)\geq\eps\,\,\forall G'\in\GSOL_{E}\left(n\right)\right\} \,\,\text{.}
\]
Note that 
\[
d^{\ham}\left(G_{\overline{\sigma}},G_{\overline{\tau}}\right)=\sum_{i=1}^{d}d^{\ham}\left(\sigma_{i},\tau_{i}\right)
\]
for $\overline{\sigma}=\left(\sigma_{1},\dotsc,\sigma_{d}\right)$
and $\overline{\tau}=\left(\tau_{1},\dotsc,\tau_{d}\right)$ in $\left(\Sym\left(n\right)\right)^{d}$, 
and that
$$\GSOL_{E}^{\geq\eps}\left(n\right)=\left\{ G_{\overline{\sigma}}\mid\overline{\sigma}\in\SOL_{E}^{\geq\eps}\left(n\right)\right\}\,\,\text{.}$$

We can now restate Definition \ref{def:testable-equations} in terms of $S$-graphs.
\begin{defn}
	[Testable system of relations in terms of $S$-graphs]\label{def:testable-equations-graphs}An algorithm
	$\calM$ that takes $n\in\NN$ and an $S$-graph $G\in \calG_S(n)$
	as input is an \emph{$\left(\eps,q\right)$-tester} for $E=\left\{ w_{i,1}=w_{i,2}\right\} _{i=1}^{r}$ if it
	satisfies the following conditions:
	\begin{itemize}
		\item \textbf{Completeness:} if $G\in\GSOL_{E}\left(n\right)$,
		the algorithm accepts with probability at least $0.99$.
		\item \textbf{$\eps$-soundness:} if $G\in\GSOL_{E}^{\ge\eps}\left(n\right)$,
		the algorithm rejects with probability at least $0.99$.
		\item \textbf{Query efficiency:} the algorithm is only allowed to make $q$
		queries, where each query is of the form ``what is $s_Gx$?'', $x\in [n]$, $s\in S^\pm$. 
	\end{itemize}
	If for every $\eps>0$ there are $q=q\left(\eps\right)\in\NN$ and
	an $\left(\eps,q\right)$-tester $\calM_{\eps}$ for $E$
	then we say that $E$ is \emph{$q\left(\eps\right)$-testable} (or
	just \emph{testable}) and that $\eps\mapsto\calM_{\eps}$ is a \emph{family
		of testers} for $E$.
\end{defn}

\subsection{The \texttt{Sample and Substitute} algorithm, a first attempt at defining a tester}
We turn to discussing algorithms for testing a given system of relations. Fix a system of relations $E$. The \texttt{Sample and Substitute} algorithm with repetition factor $k$ (denoted $\SAS_k^E$), defined below, is arguably the most natural candidate for an algorithm that tests whether an $S$-graph $G$ lies in $\GSOL_E$.

\begin{algorithm}[H]\caption{\label{alg:SAS}
		\textsf{Sample and Substitute} for $E$ with repetition factor $k$\\
		\textbf{Input:} $n\in\NN$ and $G\in \calG_S(n)$}\begin{algorithmic}[1]
		
		\State Sample $\left(w_1,x_1\right),\dotsc,\left(w_k,x_k\right)$
		uniformly and independently from $R_E\times\left[n\right]$.
		
		\If{ $w_jx_{j}=x_{j}$ in the graph $G$
			for all $1\leq j\leq k$}
		
		\State Accept.
		
		\Else\State Reject.\EndIf
		
	\end{algorithmic}
	
\end{algorithm}

For example, when $S = \{\eq{X,Y}\}$ and $E = \{\eq{XY=YX}\}$ (and so $R_E = \left\{\eq{X^{-1}Y^{-1}XY}\right\}$), the algorithm $\SAS_k^E$ randomly samples $k$ vertices $x_1,\dotsc, x_k$ from $[n]$.
It then substitutes each $x_j$ into the permutation of $[n]$ defined by $\eq{X^{-1}Y^{-1}XY}$ and checks whether the result is $x_j$. 
In other words, for each $1\leq j\leq k$, the algorithm checks whether walking from $x_j$ along the edge labelled $\eq{Y}$ and then the edge labelled $\eq{X}$, and then in the reverse direction along the edge labelled $\eq{Y}$ and then the edge labelled $\eq{X}$, leads back to $x_j$. The algorithm accepts if and only if this check succeeds for all $1\leq j\leq k$. For this specific system of relations, the query complexity is $4k$, since, to compute $\eq{X^{-1}Y^{-1}XY} x_j$, the algorithm first needs to query the vertex $\eq {Y} x_i$, then the vertex $\eq{X}(\eq{Y}x_j)$, and so on, for a total of $4$ queries.
In terms of tuples of permutations (rather than graphs), the input is $(\sigma_1,\sigma_2)\in\left(\Sym(n)\right)^2$, and the run of the algorithm is equivalent to sampling $x_1,\dotsc,x_k$ uniformly and independently from $[n]$ and accepting if and only if $\sigma_1^{-1}\sigma_2^{-1}\sigma_1\sigma_2x_j=x_j$ (equivalently, $\sigma_1\sigma_2x_j=\sigma_2\sigma_1x_j$) for all $1\leq j\leq k$.

For a general system of relations $E$, the query complexity of $\SAS_{k}^{E}$ is $O_{E}\left(k\right)$
since, generalizing from the case $E=\left\{ \mathsf{XY=YX}\right\} $,
the algorithm requires $\left|w_j\right|$ queries in order
to compute $(w_j)_{G}x_j$,
where we write $\left|w\right|$ for the length of a reduced
word $w$. In particular, if $k$ is independent of $n$ then so is
the query complexity.
If $E$ admits a family of testers consisting of $\SAS_k^E$ algorithms, we say that $E$ is \emph{stable}. This is formalized in the following definition,
which is stated in combinatorial terms in \cite[Definition 1]{GlebskyRivera} and in
group-theoretic terms in \cite[Theorem 4.2]{ArzhantsevaPaunescu} and \cite[Definition 1.1]{Ioana}. Here the definition is given in computational terms:
\begin{defn}\label{def:StableRelations}
	\label{def:stable}The system of relations $E$ is \emph{stable }if
	there is a function $k\colon\RR_{>0}\rightarrow\NN$ such that $\eps\mapsto\SAS_{k\left(\eps\right)}^{E}$
	is a family of testers for $E$. In this case we also say that $E$
	is \emph{$k\left(\eps\right)$-stable}.
\end{defn}
In particular, every stable system of relations is testable. 

Note that for every $k\in \NN$, the algorithm $\SAS_k^E$ has perfect completeness, namely, when $G\in \GSOL_E(n)$, the algorithm always accepts. Thus, the question of whether a system $E$ is stable is equivalent to the question of whether there is a function $\epsilon\mapsto k\left(\eps\right)$ such that $\SAS_{k\left(\eps\right)}^E$ is $\eps$-sound for all $\eps>0$.

The question of which systems of relations are stable has recently received a lot of attention (see Section \ref{sec:survey-stability}), yielding both positive and negative results. For example, the aforementioned testability of $E = \{\eq{XY=YX}\}$ follows from the following stability result:

\begin{exm}[$\{\eq{XY=YX}\}$ is stable]\label{example:XY=YX}
	Consider the system $E = \{\eq{XY=YX}\}$. The main theorem of \textbf{\cite{ArzhantsevaPaunescu}}
	implies that there is a function $k\colon\RR_{>0}\rightarrow\NN$
	such that $\eps\mapsto\SAS_{k\left(\eps\right)}^{\eq{XY=YX}}$ is
	a family of testers for $\eq{XY=YX}$. A later work \cite{BeckerMosheiff} improved upon the result of \cite{ArzhantsevaPaunescu} by showing that we may take
	$k\left(\eps\right)=\left(\frac{1}{\eps}\right)^{O\left(1\right)}$,
	where the implied constant is absolute, and by providing a constructive\footnote{That is, the proof in \cite{BeckerMosheiff} provides an explicit algorithm that, given a graph $G\in\calG_S(n)$ that is accepted by $\SAS_{k(\eps)}^{\eq{XY=YX}}$ with high probability, produces a graph $G'\in\GSOL_E(n)$ close to $G$.}
	proof. Thus, not only is $E$ testable by $\textup{\texttt{Sample and Substitute}}$, but moreover, this testing can be done efficiently.
\end{exm}

One natural question from the testing perspective is whether \texttt{Sample and Substitute} is \emph{universal}, namely, does every testable system of relations have a family of \texttt{Sample and Substitute} testers? In other words, are testability and stability equivalent? As we shall see, Theorem \ref{thm:MainPositive} can be used to provide a negative answer to this question. In Section \ref{subsec:testableInstable} we give a concrete example of a system of relations which is testable but not stable (see also Appendix \ref{appendix:abels}).

Fortunately, however, it turns out that a universal tester does exist. We turn now to describing this tester, which we name $\texttt{Local Statistics Matcher}$\footnote{Readers familiar with Benjamini--Schramm convergence will recognize a close connection between  $\texttt{Local Statistics Matcher}$ and the Benjamini--Schramm metric.}.

\subsection{\label{subsec:intro-LSM}\textsf{Local Statistics Matcher}, a universal tester for relations between permutations}
The idea behind \textsf{Local Statistics Matcher} is as follows: Let $G$ be a finite $S$-graph. For a vertex $x$ of $G$ and $r\geq 1$, we consider the isomorphism class of the closed ball\footnote{The ball $B_G\left(x,r\right)$ is the graph whose vertex set $V$ consists of all vertices $wx$ of $G$, where $w$ is a word of length at most $r$ over $S^\pm$, and whose edge set
consists of all the edges of $G$ of the form $w'x\overset{s}{\rightedge}sw'x$, where $s\in S^\pm$ and $w'$ is a word of length at most $r-1$.}
$B_G\left(x,r\right)$
of radius $r$ centered at $x$ as a rooted directed edge-labelled graph (the root is $x$ and the orientations and edge labels are inherited from $G$, but the relevant notion of a graph isomorphism ignores vertex labels). The graph $G$ gives rise to a probability distribution on the set of such isomorphism classes: Write $N_{G,r}$ for the distribution of the isomorphism class of $B_G\left(x,r\right)$, where $x$ is sampled uniformly from the set of vertices of $G$.

Given a graph $G\in\calG_S(n)$, \textsf{Local Statistics Matcher} computes an approximation of the distribution $N_{G,r}$ (for a large enough $r$, independent of $n$), which we denote by $N_{G,r}^{\Emp}$. The algorithm accepts if and only if there exists $G'\in \GSOL_E(n)$ such that the distributions $N_{G,r}^{\Emp}$ and $N_{G',r}$ are at most $\delta$-far from each other (for a small enough $\delta>0$, independent of $n$). In other words, \textsf{Local Statistics Matcher} accepts with high probability if and only if $G$ and $G'$ are close together under the Benjamini--Schramm metric \cites{AldousLyons, BenjaminiSchramm}.

More precisely, \textsf{Local Statistics Matcher} is a parameterized family of algorithms, i.e., for $k$, $P$ and $\delta$ (see below) we have a \textsf{Local Statistics Matcher} algorithm $\LSM_{k,P,\delta}^E$.
As we shall see in Theorem \ref{thm:LSMUniversal}, \textsf{Local Statistics Matcher} is universal in the following sense: the system $E$ is testable if and only if for every $\eps>0$ there are $k$, $P$ and $\delta$ such that $\LSM_{k,P,\delta}^E$ is an $(\eps,q)$-tester for $E$, for some $q$ which depends only on $\eps$.

Rather than working with isomorphism classes of balls, it is more convenient (and essentially equivalent) to consider \emph{stabilizers}:
\begin{defn}
	\label{def:stab-and-Nsigma-P}For a vertex $x$ of an $S$-graph $G$, let 
	\[
	\stab_{G}\left(x\right)=\left\{ w\in F_{S}\mid w_Gx=x\right\} \,\,\text{.}
	\]
	For $P\subset F_{S}$, write $N_{G,P}$ for the distribution
	(over the set of subsets of $P$) of $\stab_{G}\left(x\right)\cap P$
	when $x$ is sampled uniformly from the set of vertices of $G$.
\end{defn}
If $P=\left\{ w\in F_{S}\mid\left|w\right|\le r\right\} $, $r\geq1$,
then $\stab_{G}\left(x\right)\cap P$ determines the
isomorphism class, as a rooted directed edge-labelled graph,
of the closed ball $B_{G}\left(x,r/2\right)$
of radius $r/2$ centered at $x$ in the graph $G$. Indeed, $\stab_{G}\left(x\right)\cap P$ determines the set
\begin{equation}
	\left\{ (w_1,w_2)\in F_{S}\times F_{S} \mid\left|w_1\right|,\left|w_2\right|\le r/2, w_1x=w_2x \right\} \,\,\text{.}\label{eq:colliding-pairs}
\end{equation}
because $w_1x=w_2x$ if and only if the reduced word equivalent to $w_{2}^{-1}w_{1}$ is in $\stab_{G}\left(x\right)\cap P$ (when $\left|w_1\right|,\left|w_2\right|\le r/2$). The set \eqref{eq:colliding-pairs} encodes the isomorphism class of $B_{G}\left(x,r/2\right)$.

Before formulating the algorithm we need some notation. Given a set $A$, we denote its power set by $\Subsets\left(A\right)$, and write $\FinSubsets\left(A\right)$ for the set of finite subsets of $A$. When $A$ is finite, we write $\U(A)$ for the uniform distribution over $A$.

Denote the total-variation distance between two distributions $\theta_{1},\theta_{2}$
over a finite set $\Omega$ by
\[
d_{\TV}\left(\theta_{1},\theta_{2}\right)\coloneqq\frac{1}{2}\sum_{x\in\Omega}\left|\theta_{1}\left(x\right)-\theta_{2}\left(x\right)\right|\,\,\text{.}
\]

The \texttt{Local Statistics Matcher} algorithm for a system of relations $E$ takes three parameters in addition to $E$: a \emph{repetition factor} $k\in \NN$; a finite \emph{word set} $P\in \FinSubsets(F_S)$; and a \emph{proximity parameter} $\delta > 0$. The algorithm, denoted $\LSM_{k,P,\delta}$, is defined as follows.

\begin{algorithm}[H]\caption{\label{alg:LSM}
		\textsf{Local Statistics Matcher} for $E$ with with repetition factor $k$, word set $P\subset F_S$ and proximity parameter $\delta$\\
		\textbf{Input:} $n\in\NN$ and $G\in \calG_S(n)$}\begin{algorithmic}[1]
		
		\State\label{step:LSM-sampling}Sample $x_{1},\dotsc,x_{k}$ uniformly
		and independently from $\left[n\right]$.
		
		\State For each $1\le j\le k$, compute the set $\stab_{G}\left(x_{j}\right)\cap P$
		by querying $G$.
		
		\State Let $N_{G,P}^{\Emp}$ be the distribution of $\stab_{G}\left(x_{j}\right)\cap P$
		where $j$ is sampled uniformly from $\left[k\right]$.
		
		\If{
			\begin{equation}
				\min\left\{ d_{\TV}\left(N_{G,P}^{\Emp},N_{H,P}\right)\mid H\in\GSOL_{E}\left(n\right)\right\} \le\delta\label{eq:LSMCondition}
			\end{equation}
		}
		
		\State Accept.
		
		\Else\State Reject.\EndIf
		
	\end{algorithmic}
	
\end{algorithm}

In our proof of the universaility of \texttt{Local Statistics Matcher}, we only make use of sets $P$ of the form $P=\left\{ w\in F_{S}\mid\left|w\right|\le r\right\}$, $r\geq 1$. This restriction does not hurt the universality of the algorithm. For sets $P$ of this form, the algorithm can essentially be seen as sampling balls of radius $\frac r2$.

The query complexity of $\LSM_{k,P,\delta}^{E}$ is $O\left(k\sum_{w\in P}\left|w\right|\right)$,
and in particular it is independent of $n$ whenever the same is true
for $k$, $P$ and $\delta$. Indeed, to determine $\stab_{G}\left(x_{j}\right)\cap P$,
$1\le j\le k$, it suffices to compute $w_Gx_{j}$
for each $w\in P$, and this can be done in $\left|w\right|$ queries
for any given $w$.

In contrast to \texttt{Sample and Substitute}, the \texttt{Local Statistics Matcher} algorithm is not necessarily perfectly complete. Namely, for some systems of relations $E$, the algorithm may reject (with some small probability) even if $G\in \GSOL_E(n)$. This can happen in a run of \texttt{Local Statistics Matcher} only if $N_{G,P}^{\Emp}$ is not a good approximation of $N_{G,P}$.

\begin{rem}
	\label{rem:RunningTimeLSM}In a naive implementation of $\LSM_{k,P,\delta}^{E}$,
	the algorithm explicitly enumerates the elements of $\GSOL_{E}\left(n\right)$
	in order to compute the set of distributions $\left\{ N_{H,P}\mid H\in\GSOL_{E}\left(n\right)\right\} $.
	This results in running time exponential in $n$. While the present
	work focuses on query complexity, we discuss possible significant
	time complexity improvements in Section \ref{sec:Runtime}.
\end{rem}

\subsubsection{Statistical distinguishability and the universality of \textup{\texttt{Local Statistics Matcher}}}

To prove the universality of \texttt{Local Statistics Matcher}, we introduce the notion of a \emph{statistically distinguishable} system of relations (Definition \ref{def:statisticallyDistinuishable}). We then show (Theorem \ref{thm:LSMUniversal}) that for a system of relations $E$, the following implications hold:
$$
E \text{ is testable} \implies E \text{ is statistically distinguishable} \implies E \text{ is testable by }\texttt{Local Statistics Matcher},$$
and thus the three statements are equivalent.

\begin{defn}\label{def:statisticallyDistinuishable}
	The system of relations $E$ is \emph{statistically distinguishable} if for every
	$\eps>0$ there exist $P\left(\eps\right)\in\FinSubsets\left(F_{S}\right)$
	and $\delta\left(\eps\right)>0$ such that $d_{\TV}\left(N_{G,P\left(\eps\right)},N_{H,P\left(\eps\right)}\right)\ge\delta\left(\eps\right)$
	for every $n\in\NN$, $G\in\GSOL_{E}\left(n\right)$
	and $H\in\GSOL_{E}^{\ge\eps}\left(n\right)$. In this
	case we also say that $E$ is \emph{$\left(P\left(\eps\right),\delta\left(\eps\right)\right)$-statistically-distinguishable}.
\end{defn}

\begin{rem}
	\label{rem:LocalVsGlobalMetric}Given graphs $G,H\in\calG_S(n)$,
	we sometimes refer to $d^{\ham}\left(G,H\right)$
	as their \emph{global distance}, and to $d_{\TV}\left(N_{G,P},N_{H,P}\right)$
	(where $P\in\FinSubsets\left(F_{S}\right)$) as their \emph{$P$-local
		distance}. Thus, $E$ is statistically distinguishable if whenever
	$H$ is $\eps$-far from $\GSOL_{E}\left(n\right)$ in
	the global metric, it is also $\delta$-far from $\GSOL_{E}\left(n\right)$
	in the $P$-local metric, where $\delta$ and $P$ may depend on $\eps$.
\end{rem}

\begin{thm}
	\label{thm:LSMUniversal}Fix a system of relations $E$ over an alphabet $S^\pm$. The following conditions are equivalent:
	\begin{enumerate}
		\item $E$ is testable.
		\item There exist functions $k\colon\RR_{>0}\to\NN$, $P\colon\RR_{>0}\to\FinSubsets\left(F_{S}\right)$
		and $\delta\colon\RR_{>}\to\RR_{>0}$ such that $\eps\mapsto\LSM_{k\left(\eps\right),P\left(\eps\right),\delta\left(\eps\right)}^{E}$
		is a family of testers for $E$. 
		\item $E$ is statistically distinguishable.
	\end{enumerate}
\end{thm}

We prove Theorem \ref{thm:LSMUniversal} in Section \ref{sec:LSM}. Note that, in addition to implying the universality of \texttt{Local Statistics Matcher}, Theorem \ref{thm:LSMUniversal} reduces the question of testability to the notion of statistical distinguishability. The latter is graph theoretic and geometric, rather than algorithmic. Our proofs of Theorems \ref{thm:MainPositive} and \ref{thm:MainNegative} rely on this reduction.

\subsection{Overview}

Section \ref{sec:SurveyAndFurther} contains a survey of previous work, and discusses directions for future research.
Section \ref{sec:GroupTheory} explains the powerful relationship between testability and group theory: each system of relations $E$ gives rise to a group $\Gamma(E)$, and the foundational Proposition \ref{prop:testability-group-property} says that the testability of $E$ depends only on $\Gamma(E)$.
Section \ref{sec:GroupTheory} also reviews the relevant group-theoretic notions in a combinatorial language.

In Section \ref{sec:kazhdan-not-testable} we prove Theorem \ref{thm:MainNegative} by a direct use of expansion in graphs.
In Section \ref{sec:amenable-are-testable} we prove Theorem \ref{thm:MainPositive} using deep results such as the Newman--Sohler Theorem \cite{NewmanSohler2013} (see also \cite{Elek2012}) and the Ornstein--Weiss Theorem \cites{OrnsteinWeiss, ConnesFeldmanWeiss}.
In Section \ref{sec:TestabilityIsAGroupProperty} we prove Proposition \ref{prop:testability-group-property}.

The proofs of Theorems \ref{thm:MainPositive} and \ref{thm:MainNegative} and Proposition \ref{prop:testability-group-property} use the reduction from testability to the geometry of graphs, as afforded by Theorem \ref{thm:LSMUniversal}, which is proved in Section \ref{sec:LSM}.

Appendix \ref{appendix:abels} uses Theorem \ref{thm:MainPositive} to provide an example of a testable instable system of relations.
Appendix \ref{appendix:SL} gives an example of a system satisfying the hypothesis of Theorem \ref{thm:MainNegative}.

\section{A survey of previous work and directions for further research}
\label{sec:SurveyAndFurther}

\subsection{\label{sec:survey-stability}A survey of previous work on stability}

The notion of testability of relations and the computational perspective
on stability are new concepts, introduced in this paper. Nevertheless,
there is a lot of recent research of stability in permutations from
group-theoretic and combinatorial points of view. Most of the existing literature on stability relies on the relation to group theory as explained in Section \ref{sec:GroupTheory}. We formulate these results in the graph-theoretic language developed in Section \ref{sec:GraphView}.

\subsubsection{\label{subsec:old-defs-of-stability}The combinatorial definition of stability}
Stability, as in Definition \ref{def:stable}, has several equivalent definitions (see  \cite[Definition 1]{GlebskyRivera}, \cite[Theorem 4.2]{ArzhantsevaPaunescu}, \cite[Definition 1.1]{Ioana}).
Below we recall one of them and prove the equivalence.
\begin{defn}
 	[Combinatorial definition of stability]\label{def:DefectStability}\cite[Definition 1]{GlebskyRivera}
	Let $E$
	be a system of relations over $S^{\pm}$, and let $G\in \calG_S(n)$.
	The \emph{local defect} of $G$ with respect to $E$
	is 
	\[
	L_{E}\left(G\right)\coloneqq\sum_{w\in R_E}\Pr_{x\sim \U\left([n]\right)}\left[w_Gx \ne x\right].
	\]
	We say that $E$ is \emph{stable} if
	\[
	\inf\left\{ L_{E}\left(G\right)\mid G\in\GSOL_{E}^{\ge\eps}\left(n\right),n\in\NN\right\} >0\quad\forall\eps>0\,\,\text{.}
	\]
\end{defn}
\begin{rem}
We use the notion of local defect in the proof of Proposition \ref{prop:testability-group-property} (see Section \ref{sec:TestabilityIsAGroupProperty}).
\end{rem}

For $G\in\calG_S(n)$, we
have $G\in\GSOL_{E}\left(n\right)$ if and only if
$L_{E}\left(G\right)=0$.
Another way to state Definition \ref{def:DefectStability} is:
The system $E$ is stable
if and only if for every $\eps>0$ there is $\delta>0$ such that
$G\notin\GSOL_{E}^{\geq\eps}\left(n\right)$ whenever
$n\in\NN$, $G\in\calG_S(n)$
and $L_{E}\left(G\right)<\delta$. That is (informally),
$E$ is stable if every graph which approximately satisfies $E$ is close to
a graph which fully satisfies $E$. 

For $G\in\calG_S\left(n\right)$, it is
easy to see that the probability that $\SAS_{1}^{E}$ rejects $G$
is $\frac{1}{\left|E\right|}L_{E}\left(G\right)$,
and that $\SAS_{k}^{E}$ can be implemented by running $\SAS_{1}^{E}$
for $k$ independent iterations and accepting if all iterations accept.
These observations are used in the proof of the following claim.
\begin{claim}
	A system of relations $E$ is stable in the sense of Definition \ref{def:DefectStability}
	if and only if it is stable in the sense of Definition \ref{def:stable}.
\end{claim}

\begin{proof}
	Suppose that $E$ is stable in the sense of Definition \ref{def:DefectStability}.
	For $\eps>0$, write 
	\[
	\delta\left(\eps\right)=\frac{1}{\left|E\right|}\inf\left\{ L_{E}\left(G\right)\mid G\in\GSOL_{E}^{\ge\eps}\left(n\right),n\in\NN\right\} >0
	\]
	and
	\[
	k\left(\eps\right)=\left\lceil \log_{1-\delta\left(\eps\right)}0.01\right\rceil \,\,\text{.}
	\]
	We claim that $\eps\mapsto\SAS_{k\left(\eps\right)}^{E}$ is a family
	of testers for $E$. Let $n\in\NN$ and $G\in\GSOL_{E}^{\ge\eps}\left(n\right)$.
	Then $\SAS_{1}^{E}$ rejects $G$ with probability
	$\frac{1}{\left|E\right|}L_{E}\left(G\right)\geq\delta\left(\eps\right)$.
	Thus, the probability that $\SAS_{k}^{E}$ accepts $G$
	is at most
	\[
	\left(1-\delta\left(\eps\right)\right)^{k}\leq\left(1-\delta\left(\eps\right)\right)^{\log_{1-\delta\left(\eps\right)}0.01}=0.01\,\,\text{.}
	\]
	
	Conversely, suppose that $E$ is stable in the sense of Definition
	\ref{def:stable}, and let $\eps\mapsto\SAS_{k\left(\eps\right)}^{E}$
	be a family of testers for $E$, $k\colon\RR_{>0}\to\NN$. Take $G\in\GSOL_{E}^{\ge\eps}\left(n\right)$.
	Write $p_{1}$ (resp. $p_{k\left(\eps\right)}$) for the probability
	that $\SAS_{1}^{E}$ (resp. $\SAS_{k\left(\eps\right)}^{E}$) rejects
	$G$. On one hand, $p_{1}=\frac{1}{\left|E\right|}L_{E}\left(G\right)$
	and $p_{k\left(\eps\right)}\geq0.99$. On the other hand, $p_{k\left(\eps\right)}\leq k\left(\eps\right)\cdot p_{1}$
	by the union bound, and thus $L_{E}\left(G\right)\geq\frac{0.99\left|E\right|}{k\left(\eps\right)}$.
	Hence
	\[
	\inf\left\{ L_{E}\left(G\right)\mid G\in\GSOL_{E}^{\ge\eps}\left(n\right),n\in\NN\right\} \geq\frac{0.99\left|E\right|}{k\left(\eps\right)}>0\,\,\text{.}
	\]
\end{proof}

\subsubsection{Stability under the hypothesis of Theorem \ref{thm:MainPositive}}
Let $E$ be a system of relations satisfying the hypothesis of Theorem \ref{thm:MainPositive}. That is, $\iota(G)=0$ for every $G\in\GSOL_E$.
Theorem \ref{thm:MainPositive} says that $E$ is testable. However, $E$ is not necessarily stable. In fact, there is a group-theoretic criterion that determines whether a system $E$, satisfying the hypothesis of Theorem \ref{thm:MainPositive}, is stable (see Theorem \cite[Theorem 1.3(ii)]{BLT}\footnote{
Theorem 1.3(ii) of \cite{BLT} characterizes stability among systems $E$ for which the group $\Gamma(E)$ is amenable (see Section \ref{sec:GroupTheory})}). See Section \ref{subsec:testableInstable} and Appendix \ref{appendix:abels} for an explicit example of a stable non-testable system.

\begin{exm}[Baumslag--Solitar relations] \label{example:BS-stability}
	Fix $m,n\in \ZZ$, and consider the system $E_{m,n}=\left\{ \eq{XY^{m}=Y^{n}X}\right\}$, consisting of a single relation. Then $E_{m,n}$ is stable if and only if $\left|m\right|\leq1$ or $\left|n\right|\leq1$. Indeed:
\begin{itemize}
	\item The case $m=n=0$ is clear.
	\item If $m=0$ and $n\neq0$ then $E_{m,n}$ is stable by \cite[Theorem 2]{GlebskyRivera}.
	\item If $\left|m\right|=1$ or $\left|n\right|=1$ then $E_{m,n}$ is
	stable by \cite[Theorem 1.2(ii)]{BLT}.
	\item If $\left|m\right|,\left|n\right|\geq2$ and $\left|m\right|\neq\left|n\right|$
	then $E_{m,n}$ is not stable by \cite[Example 7.3]{ArzhantsevaPaunescu}
	(see also \cite[Theorem 1.3(i)]{BLT}).
	\item If $\left|m\right|,\left|n\right|\geq2$
	and $\left|m\right|=\left|n\right|$ then $E_{m,n}$ is not stable
	by \cite[Corollary B(3)]{Ioana}.
\end{itemize}
	
	As for the testability of $E_{m,n}$, the stable cases are clearly testable. Testability in the case $|m|,|n|\ge 2$ and $|m|\ne |n|$ is an interesting open question (see Section \ref{subsec:characterization}).
	
	In the remaining case, $|m|,|n|\ge 2$ and $|m|=|n|$, it turns out that $E_{m,n}$ is not testable. Indeed, by \cite[Theorem D]{Ioana} there are $\eps_0>0$ and finite $S$-graphs in $\GSOL_{E_{m,n}}^{\geq \eps_0}$ with local statistics approximating a certain\footnote{Namely, $H$ is the Cayley graph of the Baumslag--Solitar group $\BS(m,n)$} infinite vertex-transitive $S$-graph $H$ arbitrarily well. On the other hand, it is well-known that there are also finite $S$-graphs in $\GSOL_{E_{m,n}}$ whose local statistics approximate the same infinite graph $H$ arbitrarily well\footnote{This follows from the group $\BS(m,n)$ being residually finite when $|m|=|n|$ \cite[Proposition 2.6(3)]{Levitt05}}.
Thus $E_{m,n}$ is not statistically distinguishable\footnote{This argument is similar to the proof of Theorem \ref{thm:kazhdan-not-distinguishable} using Lemma \ref{lem:action-size-changer} in Section \ref{sec:kazhdan-not-testable}}, and hence it is not testable by Theorem \ref{thm:LSMUniversal}.
In fact, $E_{m,n}$ is not even flexibly testable (see Section \ref{subsec:survery-flexible-stability}) because the argument above remains true even with a flexible definition of $\GSOL_{E_{m,n}}^{\geq \eps_0}$ by \cite[Theorem D]{Ioana}.
	
	By the above, the system of relations $E_{2,2}=\{\mathsf{XY^2=Y^2X}\}$, mentioned in the beginning of Section \ref{sec:intro}, is not flexibly testable. By a similar argument, also based on \cite[Theorem D]{Ioana} and Theorem \ref{thm:LSMUniversal}, the system $\{\mathsf{XZ=ZX}, \mathsf{YZ=ZY}\}$ is also not flexibly testable. 
\end{exm}

\subsubsection{\label{subsec:survery-flexible-stability}Flexible stability (and testability)}

A slightly weaker form of stability, called \emph{flexible stability}
\cite[Section 4.4]{BeckerLubotzky}, has led to fascinating research.
Here we define this notion in the language of the present paper. We also introduce a new, more general, notion of \emph{flexible testability}.

For $G\in \calG_S(n)$
and $G'\in \calG_S(N)$, $N\geq n$, let
\[
d^{\ham}\left(G,G'\right)=d^{\ham}\left(G',G\right)=\sum_{s\in S}\frac{\left|\left\{ x\in\left[n\right]\mid s_G x\neq s_{G'} x\right\} \right|+\left(N-n\right)}{N} \,\,\text{.}
\]
Then $d^{\ham}$ is a metric on the disjoint union $\coprod_{n\in\NN}\calG_S(n)$
\cite[Lemma A.1]{BeckerChapman}, extending (\ref{eq:graphDistance}).
Note that $d^{\ham}\left(G,G'\right)\in\left[0,|S|\right]$,
and that if $d^{\ham}\left(G,G'\right)$ is close to $0$ then
$\frac{n}{N}$ is close to $1$ (indeed, $1-\frac{d^{\ham}\left(G,G'\right)}{|S|}\leq\frac{n}{N}\leq1$).
For a system of relations $E$ over $S^{\pm}$, let
\[
\GSOL_{E}^{\geq\eps,\flex}\left(n\right)=\left\{ G\in\calG_S(n) \mid d^{\ham}\left(G,G'\right)\geq\eps\quad\forall G'\in\coprod_{m\in\NN}\GSOL_{E}\left(m\right)\right\} \,\,\text{.}
\]

\begin{defn}[Flexibly-testable system of relations]
An algorithm $\calM$ that takes $n\in\NN$ and an $S$-graph $G\in \calG_S(n)$
as input is an \emph{$\left(\eps,q\right)$-flexible-tester} for $E$
if it satisfies the following conditions:
\begin{itemize}
	\item \textbf{Completeness:} if $G\in\GSOL_{E}\left(n\right)$,
	the algorithm accepts with probability at least $0.99$.
	\item \textbf{$\eps$-soundness:} if $G\in\GSOL_{E}^{\geq\eps,\flex}\left(n\right)$,
	the algorithm rejects with probability at least $0.99$.
	\item \textbf{Query efficiency:}  the algorithm is only allowed to make $q$
	queries, where each query is of the form ``what is $s_Gx$?'', $x\in [n]$, $s\in S^\pm$. 
\end{itemize}
If for every $\eps>0$ there are $q=q\left(\eps\right)\in\NN$ and
an $\left(\eps,q\right)$-flexible-tester $\calM_{\eps}$ for $E$
then we say that $E$ is \emph{$q\left(\eps\right)$-flexibly-testable}
(or just \emph{flexibly testable}, or \emph{testable under the flexible
	model}) and that $\eps\mapsto\calM_{\eps}$ is a family of flexible
testers for $E$.
\end{defn}
Clearly, every testable system of relations is also flexibly testable. As shown in Example \ref{example:BS-stability}, not every system of relations is flexibly testable.

\begin{defn}[Flexibly-stable system of relations]
We say that $E$ is \emph{flexibly stable} if there is a function
$k\colon\RR_{>0}\to\NN$ such that $\eps\mapsto\SAS_{k\left(\eps\right)}^{E}$
is a family of flexible testers for $E$.
\end{defn}

Clearly, if $E$ is stable then it is flexibly stable. The converse
is still open:
\begin{problem}
	\label{prob:flex-stable-not-stable}Is there a flexibly-stable system
	of relations that is not stable?
\end{problem}
An open problem similar to Problem \ref{prob:flex-stable-not-stable}, in the context of testability, is stated and discussed in Section \ref{subsec:OpenFlexibleTestability}.

\begin{rem}
\label{rem:sofic-nonRF-is-not-flex-stable}
We note that the sufficient condition for instability in \cite[Theorem 1.3(i)]{BLT} is in fact sufficient for flexible instability
(this follows directly from the proof given in \cite{BLT} which was written before the notion of flexible stability was defined formally).
\end{rem}

The following deep result provides examples of flexibly-stable systems, none of which are known to be stable.
\begin{exm}
	\label{example:surface-groups}\cite[Theorem 1.1]{LLM} For $g\geq2$,
	the system $E_{g}=\left\{ \left[s_{1},s_{2}\right]\left[s_{3},s_{4}\right]\cdots\left[s_{2g-1},s_{2g}\right]=1\right\} $,
	consisting of a single relation, is flexibly stable. Here $\left[a,b\right]\coloneqq aba^{-1}b^{-1}$.
	
	More precisely, $\eps\mapsto\SAS_{k\left(\eps\right)}^{E_{g}}$ is
	a family of flexible testers for $E_{g}$ for $k\left(\eps\right)=O\left(\frac{1}{\eps}\log\frac{1}{\eps}\right)$.
\end{exm}

We refer the reader to \cite[Section 4.4]{BeckerLubotzky} and \cite[Section 3]{Ioana}
for further information regarding flexible stability.


\subsubsection{\label{subsec:survery-query-complexity}Query efficiency}

It is interesting to study the query complexity of a stable system
of relations $E$ in terms of $\eps$. That is, how small can $k\left(\eps\right)$
be such that $\eps\mapsto\SAS_{k\left(\eps\right)}^{E}$ is a family
of testers for $E$.

When $|S|=d$, $d\geq 1$, and
$$E = \left\{ ss'=s's\mid s,s'\in S\right\},$$ the main result of \cite{BeckerMosheiff} (see Example \ref{example:XY=YX})
gives an upper bound on $k\left(\eps\right)$ which is polynomial
in $\frac{1}{\eps}$.
On the other hand, \cite[Theorem 1.17]{BeckerMosheiff} shows that $k\left(\eps\right)\geq\Omega\left(\left(\frac{1}{\eps}\right)^{d}\right)$
if $d\geq2$.
As for a lower bound on $k\left(\eps\right)$ under the flexible model,
we know that $k\left(\eps\right)\geq\Omega\left(\left(\frac{1}{\eps}\right)^{2}\right)$
if $d\geq2$ (stated without
proof in \cite[Section 6]{BeckerMosheiff}). In particular, the bound
$k\left(\eps\right)\leq O\left(\frac{1}{\eps}\log\frac{1}{\eps}\right)$
in Example \ref{example:surface-groups} does not hold when $g=1$.

\subsubsection{Stability of an infinite system of relations}
Our definition of a system of relations $E$ requires $E$ to be finite. It is possible to define stability even when $E$ is infinite: $E$ is stable if $\eps\mapsto\SAS_{k\left(\eps\right)}^{E'\left(\eps\right)}$
is a family of testers for $E$ for some $k\colon\RR_{>0}\to\NN$
and $E'\colon\RR_{>0}\to\FinSubsets\left(E\right)$.
Examples of stable infinite systems, which
are not equivalent to any finite system, are given in \cite[Theorem 1.7]{Zheng},
\cite{LevitLubotzky1} and \cite{LevitLubotzky2}, using \cite[Theorem 1.3(ii)]{BLT}.

\subsubsection{Matrices instead of permutations}

The general question of whether approximate solutions are close to
solutions, in various contexts, was suggested by Ulam \cite[Chapter VI]{Ulam}.
Definition \ref{def:DefectStability} puts the notion of stability
in permutations into this framework, and follows the earlier notion
of stability in matrices. The classical question in the latter context
is whether for all $n\times n$ matrices $A$ and $B$ such that $\|AB-BA\|<\delta$
there are $n\times n$ matrices $A'$ and $B'$ such that $A'B'=B'A'$
and $\|A-A'\|+\|B-B'\|<\eps$, where $\eps=\eps\left(\delta\right)\overset{\delta\to0}{\longrightarrow}0$.
The answer depends on the type of matrices considered (self-adjoint,
unitary, etc.) and the matrix norm used. See the introduction of \cite{ArzhantsevaPaunescu}
for a short survey of stability of the relation $\eq{XY=YX}$ in matrices,
and \cite{BeckerLubotzky,Dadarlat,DGLT,ESS,EndersShulman,HadwinShulman,LubotzkyOppenheim,MoralesGlebsky,OppenheimBanach}
for newer works that consider more general relations.

\subsection{Directions for further research}

This work originated from the observation that stability, 
which has been studied extensively in the context of group theory,
admits a natural equivalent definition in the language of
property testing (Definition \ref{def:stable}).
This gave rise to the more general notion of testability of relations between permutations
(Definition \ref{def:testable-equations}).

We hope
that the present paper will stimulate more work on testability of
relations. While Theorems \ref{thm:MainPositive} and \ref{thm:MainNegative}
generate large families
of testable and non-testable relations, many fundamental questions
remain open. Below we present suggestions for further research (see Section \ref{sec:GroupTheory} for several additional open problems).

\subsubsection{A complete characterization of the testable systems of relations}\label{subsec:characterization}
The most direct goal of our line of research is to classify every system of relations as either testable or non-testable. We do not know whether this problem is computable, namely, it may be undecidable to determine, given a system of relations $E$ as input, whether $E$ is testable. Regardless of this potential obstacle, we seek a natural characterization for testability, and note that such a characterization need not necessarily be computable.
\begin{problem}{\label{prob:characterization}}
	Obtain a natural geometric characterization for testable systems of relations (i.e., a characterization in terms of properties of the graphs $\GSOL_E$)
\end{problem}

The present work advances us towards this goal via Theorems \ref{thm:MainPositive} and \ref{thm:MainNegative}, which yield large natural sets of positive and negative examples, and hint at a complete characterization in terms of the geometry of the graph family $\GSOL_E$.

A natural approach to Problem \ref{prob:characterization} is concentrating on systems of relations that have evaded our methods so far. Perhaps the most prominent candidates are the Baumslag-Solitar systems of relations (see Example \ref{example:BS-stability}):

\begin{problem}\label{prob:BS}
	Let $m,n\in\ZZ$. Is $E_{m,n}=\left\{ XY^{m}=Y^{n}X\right\} $ testable?
	If not, is it flexibly testable?
\end{problem}

The only open cases in Problem \ref{prob:BS} are when $|m|$ and $|n|$ are distinct and larger than $1$ (see Example \ref{example:BS-stability}). In all of these open cases, $E_{m,n}$ is not flexibly stable (see Example \ref{example:BS-stability} and
Remark \ref{rem:sofic-nonRF-is-not-flex-stable}).

\subsubsection{Query efficiency}\label{subsec:queryComplexity}

A testable system of relations $E$, as in Definition
\ref{def:testable-equations}, is a system
that admits an $\left(\eps,q\right)$-tester $\calM_{\eps}$ for all
$\eps>0$, where $q$ depends only on $\eps$ and $E$ (but not on
the input size $n$). From a computational perspective, it is desirable
to derive explicit good upper bounds on $q$ in terms of $\eps$. Examples \ref{example:XY=YX} and  \ref{example:surface-groups}
establish results
of this kind in the context of stability, as discussed in Section
\ref{subsec:survery-query-complexity}. It is interesting to study
the quantitative aspect in the broader context of testability. One natural question is whether Theorem \ref{thm:MainPositive} can be extended to a quantitative statement.

\begin{problem}\label{prob:MainPositiveTheoremQuantitative}
	Let $E$ satisfy the hypothesis of Theorem \ref{thm:MainPositive}. What is the minimum query complexity of a tester for $E$?
\end{problem}
See Section \ref{subsec:QuantitativeGroups} for a conjecture about Problem \ref{prob:MainPositiveTheoremQuantitative}.

In line with Examples \ref{example:XY=YX} and \ref{example:surface-groups},
it is interesting to find additional families of $\poly\left(1/\eps\right)$-testable
and $\poly\left(1/\eps\right)$-stable systems of relations.
\begin{problem}\label{prob:PolyTestability}
	Which systems of relations are $O\left(\poly\left(\frac 1\eps\right)\right)$-testable?
\end{problem}

The more refined question of the degree of polynomial stability or testability
is also interesting. The following is still open.
\begin{problem}
	What is the minimal\footnote{More precisely, we ask about the infimum of the set of all $D$ such that $E$ is $O\left(\left(\frac{1}{\eps}\right)^{D}\right)$-stable (we do not know if the minimum is attained).} $D$ such that $E=\left\{ \eq{XY=YX}\right\} $
	is $O\left(\left(\frac{1}{\eps}\right)^{D}\right)$-stable?
\end{problem}

As discussed in Section \ref{subsec:survery-query-complexity}, the
minimal $D$ is at least $2$. One can derive an explicit upper bound
on $D$ by following the proof of \cite[Theorem 1.16]{BeckerMosheiff}.

Another question is whether there are systems of relations that are stable, but where an algorithm different than \texttt{Sample and Substitute} would yield better query complexity. For example:
\begin{problem}
	Let $D$ be minimal such that $E = \{\eq{XY=YX}\}$ is  $O\left(\left(\frac{1}{\eps}\right)^{D}\right)$-stable. Is there some $D' < D$? such that $E$  is $O\left(\left(\frac{1}{\eps}\right)^{D'}\right)$-testable?
\end{problem}

In the context of lower bounds, the following is also still open.
\begin{problem}
	Is there a system of relations $E$ which is testable, but not $O\left(\left(\frac{1}{\eps}\right)^{D}\right)$-testable
	for any $D$? Similarly, is there $E$ which is stable but not $O\left(\left(\frac{1}{\eps}\right)^{D}\right)$-stable
	for any $D$?
\end{problem}

\subsubsection{Allowing the query complexity to depend on $n$}\label{subsec:QueriesDependOnN}

In this work, as in, e.g., \cite{NewmanSohler2013,AFNS09,AS08}, the
query complexity of a tester must not depend on the input size. In
general, however, testers whose query complexity depends on the input
size in a sublinear fashion are widely studied (see, e.g., \cite{Goldreich}).
Allowing such testers raises many interesting questions, such as the
following.
\begin{problem}
	Consider the system of relations $E=E_3$, discussed in Section \ref{subsec:non-testable} and defined in Appendix \ref{appendix:SL}. This system satisfies the hypothesis of Theorem \ref{thm:MainNegative}
	and thus it is not testable.
	Is there a probabilistic algorithm
	that distinguishes between elements of $\GSOL_{E_{3}}\left(n\right)$
	and $\GSOL_{E_{3}}^{\ge\eps}\left(n\right)$, in the sense of Definition
	\ref{def:testable-equations}, by making only $q=q\left(\eps,n\right)$
	queries? Here $q$ must be sublinear in $n$ for a result to be interesting,
	and the question is how small can $q$ be in terms of $n$ (the same question is open for every other system $E$ that satisfies the hypothesis of Theorem \ref{thm:MainNegative}).
\end{problem}

\subsubsection{Running time efficiency\label{sec:Runtime}}

Another aspect that we do not pursue in this work is the time complexity
of our testers. In the case of the \textsf{Sample and Substitute}
algorithm, the time complexity is easily seen to be identical to the
query complexity, so not much remains to be optimized. 

On the other hand, as explained in Remark \ref{rem:RunningTimeLSM},
a straightforward implementation of \textsf{Local Statistics Matcher}
requires time exponential in $n$ to determine whether Condition (\ref{eq:LSMCondition})
holds. However, when the graphs in $\GSOL_E$ are ``well behaved'', the set of distributions $\left\{ N_{G,P}\mid G\in\GSOL_{E}\left(n\right)\right\}$
may be structured enough to allow this condition to be checked much
more efficiently, perhaps even in time independent of $n$. 
\begin{problem}
	\label{prob:time-efficient-LSM}For which testable systems of relations
	$E$, having $\eps\mapsto\LSM_{k\left(\eps\right),P\left(\eps\right),\delta\left(\eps\right)}^{E}$
	as a family of testers, does there exist an implementation of $\LSM_{k\left(\eps\right),P\left(\eps\right),\delta\left(\eps\right)}^{E}$
	that is time-efficient in terms on $n$ (or better, has running time
	independent of $n$)?
\end{problem}

Additionally, despite the universality of \textsf{Local Statistics
	Matcher} (see Theorem \ref{thm:LSMUniversal}), it may also be worthwhile
to seek other, more time-efficient, testing algorithms.

In the context of Problem \ref{prob:time-efficient-LSM}, it is worth
noting that for a $k\left(\eps\right)$-stable system $E$, $\eps\mapsto\LSM_{k\left(\eps\right),R_E,0}^{E}$
is a family of testers
(here we plug $R_E$ itself into $\LSM$ in the role of the word-set parameter $P$).
Calculating the left-hand side of (\ref{eq:LSMCondition})
in $\LSM_{k\left(\eps\right),R_E,0}^{E}$ is trivial because $N_{H,R_E}$
is the Dirac distribution concentrated on $R_E$ for each $H\in\GSOL_{E}\left(n\right)$,
$n\in\NN$. Thus, in this case the running time discussed in Problem
\ref{prob:time-efficient-LSM} does not depend on $n$. Problem \ref{prob:time-efficient-LSM}
asks if there are also testable instable systems with this property,
or at least with running time depending weakly on $n$.

\subsubsection{Uniform testability}
\label{sec:unfirom-testability}

Following \cite[Definition 2.4]{NewmanSohler2013}, our notion of
testability from Definition \ref{def:testable-equations} is \emph{nonuniform
	in $\eps$, }in the sense that it requires a family of testers, one
for each $\eps$. In the more standard notion of testability, which
we refer to as \emph{uniform testability} (see \cite[Definition 1.6]{Goldreich}),
there is just a single tester, and it takes $\epsilon$ as input\emph{.
}To upgrade our testability to uniform testability, we need the function
$\epsilon\mapsto\calM_{\eps}$, that chooses the tester according
to $\eps$, to be computable.

We would like to understand which systems of relations are uniformly
testable, and whether there is a gap between the uniform and nonuniform
notions. In particular,
\begin{problem}
\label{prob:UniformMainPositive}
	Can Theorem \ref{thm:MainPositive} be strengthened
	to prove uniform testability for some systems of relations?
\end{problem}
This problem is further discussed in Section \ref{subsec:UniformAndGroups}.
\subsubsection{POT-testability and fixed-radius sampling}
Proximity Oblivious Testability (POT) is a stronger form of testability.
A system of relations $E$ is \emph{POT-testable} if it can be tested
by repeating the same random boolean subroutine $A$ for $k\left(\eps\right)$
independent iterations, and accepting if the number of iterations
in which $A$ returns \textsf{true} is above a certain threshold.
Notably, $A$ itself must not depend on $\eps$. See \cite[Definition 1.7]{Goldreich}
for a precise definition of POT testability in a general context.
A stable system of relations is POT-testable because $\SAS_{k}^{E}$
can be implemented by running $\SAS_{1}^{E}$ for $k$ independent
iterations, and accepting if all iterations accept.

A related notion is that of \emph{fixed-radius sampling} (FRS). We
say that a system of relations $E$ is \emph{FRS-testable} if it admits
a family of testers $\eps\mapsto\calM_{\eps}$ such that $\calM_{\eps}$,
in its run on $G\in\calG_S(n)$,
repeats the same subroutine $A$ for $k\left(\eps\right)$ independent
iterations and accepts or rejects according to the results of these
iterations. Again, $A$ must be independent of $\eps$, but unlike
the case of POT-testability, $A$ does not have to be a boolean subroutine,
and $\calM_{\eps}$ does not have to decide on its return value according
to a threshold. The term \emph{fixed radius} comes from the fact that
each iteration of $A$ on its input $G$ may only
examine a constant (i.e., independent of $\eps$) number of vertices
of the graph $G$.
Without loss of generality, we may assume that it examines a constant
number of balls in this graph of some fixed radius $r$. Thus, the
tester $\cM_{\eps}$ has information about $k\left(\eps\right)$ balls
of radius $r$ in the graph $G$. 

Clearly, every POT-testable system of relations is also FRS-testable.
In addition, if $E$ is testable by the family $\eps\mapsto\LSM_{k\left(\eps\right),P,\delta\left(\eps\right)}^{E}$
where $P$ is constant, $k\colon\RR_{>0}\to\NN$ and  $\delta\colon\RR_{>0}\to\RR_{>0}$,
then $E$ is FRS-testable.

We have thus established the hierarchy:
\begin{equation}
	\text{stability}\implies\text{POT-testability}\implies\text{FRS-testability}\implies\text{testability}\,\,.\label{eq:Hierarchy}
\end{equation}
As discussed earlier, there are systems of relations that are testable
but are not stable (see Section \ref{subsec:testableInstable}). Thus, the class of stable
systems of relations is strictly contained in the class of testable systems
of relations.
\begin{problem}
	Determine which of the inclusions arising from (\ref{eq:Hierarchy}),
	if any, are equalities.
\end{problem}

\subsubsection{The flexible model}\label{subsec:OpenFlexibleTestability}

Consider the flexible model of stability and testability as in Section
\ref{subsec:survery-flexible-stability}.
In line with Problem \ref{prob:flex-stable-not-stable}, we ask:
\begin{problem}
\label{prob:flex-testable-but-not-testable}
	Are there flexibly-testable systems of relations which are not testable?
\end{problem}

We note that Theorem \ref{thm:MainNegative} provides 
non-testable (and thus also instable) systems in the strict model, as in Definition
\ref{def:testable-equations}, but the proof of Theorem \ref{thm:MainNegative} does not rule out flexible testability. Finding a flexibly-testable system satisfying the hypothesis of Theorem \ref{thm:MainNegative} will solve Problem \ref{prob:flex-testable-but-not-testable}.

Another possible source for a positive answer to Problem \ref{prob:flex-testable-but-not-testable} is Example \ref{example:surface-groups}, which provides systems which are flexibly
stable (and hence are flexibly testable). It is not known whether these systems are testable.

\section{The connection to group theory}\label{sec:GroupTheory}
Most of the previous study of stability of relations between permutations has been done through a connection between stability and group theory. This section reviews this connection, introduces the connection between testability and group theory, and gives a brief introduction to the relevant group-theoretic notions.

Let $E$ be a system of relations over $S^\pm$.
Then $E$ gives rise to a finitely-presented group $\Gamma(E)$ by means of a group presentation.
For example, the system $E=\left\{ \eq{XY=YX}\right\} $ over $S=\left\{ \eq X,\eq Y\right\}$
gives rise to the group $\Gamma\left(E\right)=\langle\eq X,\eq Y\mid\eq{XY=YX}\rangle\cong\ZZ^{2}$.
More generally, we define $\Gamma\left(E\right)\coloneqq\langle S\mid E\rangle$
(this is the group generated by $S$ subject to the relations $E$; see Appendix \ref{app:FreeGroupAndPresentations} for a reminder
about group presentations). The association $E\mapsto\Gamma(E)$ is many-to-one.
That is, different systems of relations may give rise to isomorphic groups.
The starting point of the group-theoretic approach to stability is
the following observation from \cite{ArzhantsevaPaunescu} (see also
\cite[Proposition 1.11]{BeckerMosheiff}).
\begin{prop}
	\label{prop:stability-group-property}\cite[Section 3]{ArzhantsevaPaunescu}
	Let $E_{1}$ and $E_{2}$ be systems of relations (over possibly different
	sets of variables) such that the groups $\Gamma\left(E_{1}\right)$
	and $\Gamma\left(E_{2}\right)$ are isomorphic. Then $E_{1}$ is stable
	if and only if $E_{2}$ is stable.
\end{prop}

In Section \ref{sec:TestabilityIsAGroupProperty} we lay the foundations
for the study of testability via group theory by proving the following
analogue of Proposition \ref{prop:stability-group-property}.
\begin{prop}
	\label{prop:testability-group-property}
	Let $E_{1}$ and $E_{2}$ be systems of relations (over possibly different
	sets of variables) such that the groups $\Gamma\left(E_{1}\right)$
	and $\Gamma\left(E_{2}\right)$ are isomorphic.
	Then $E_{1}$ is testable if and only if $E_{2}$ is testable.
\end{prop}

Propositions \ref{prop:stability-group-property} and \ref{prop:testability-group-property} suggest that one can study the stability and testability of the system $E$ by studying the group $\Gamma(E)$.
The following theorem is an example of this strategy.
\begin{thm}
	\label{thm:intro-abelian}\cite{ArzhantsevaPaunescu} If the group
	$\Gamma\left(E\right)$ is abelian, then $E$ is stable.
	
	Furthermore \cite[Theorem 1.16]{BeckerMosheiff}, in this case $\eps\mapsto\SAS_{k\left(\eps\right)}^{E}$
	is a family of testers for $E$, where $k\left(\eps\right)\leq C\cdot\left(\frac{1}{\eps}\right)^{D}$.
	Here $C$ depends on $E$, and $D$ depends only on the isomorphism
	class of the group $\Gamma\left(E\right)$.
\end{thm}

Theorem \ref{thm:intro-abelian} applies to the system of relations
\begin{equation}
	E_{\comm}^{d}\coloneqq\left\{ s_{i}s_{j}=s_{j}s_{i}\mid i,j\in\left[d\right]\right\} \,\,\text{,}\label{eq:def-of-E-comm}
\end{equation}
since $\Gamma\left(E_{\comm}^{d}\right)\cong\ZZ^{d}$ is abelian.
In particular, the theorem applies to $E_{\comm}^{2}=\left\{ \eq{XY=YX}\right\} $.

Classical group properties can also be used to prove instability.
For example, the instability of $E_{2,3}\coloneqq\left\{ \eq{XY^{2}=Y^{3}X}\right\}$ \cite[Theorem 2]{GlebskyRivera} follows from properties of the group $\BS(2,3)\coloneqq\langle \eq{X},\eq{Y}\mid \eq{XY^{2}=Y^{3}X}\rangle$ (more preciesely, $E_{2,3}$ is not stable because $\BS(2,3)$ is sofic but not residually finite).

Our main theorems, namely, Theorems \ref{thm:MainPositive} and \ref{thm:MainNegative}, can be formulated in group-theoretic terms, and these formulations allow us to find many systems of relations to which the theorems apply.
For this we need the notions of amenability and property $\ptau$ (the latter is a variant of the well-known Kazhdan property $\T$).
One of the reasons for the wealth of examples and applications of these notions is that each of them has many equivalent definitions.
For groups of the form $\Gamma(E)$, i.e., for finitely-presented groups,
we give simple definitions, using the isoperimetric quantities presented in the introduction, as follows\footnote{The definitions given here of amenability and property $\ptau$ for a group of the form $\Gamma(E)$ refer to $E$ itself.
There are equivalent definitions of amenability and property $\ptau$ for a group $\Delta$ that are intrinsic to the group and do not refer to a system of relations $E$ such that $\Delta\cong\Gamma(E)$.
See \cite[Chapter 18]{DrutuKapovich}, \cite{BHV} and \cite{LubotzkyZuk} for more information on amenability, property $\T$ and property $\ptau$, respectively}.

The group $\Gamma(E)$ is amenable if and only if $\iota(G)=0$ for every $G\in\GSOL_E$. Thus Theorem \ref{thm:MainPositive} is equivalent to the following theorem:
\begin{customthm}{1'}[Main positive theorem in group terms]
\label{thm:MainPositivePrime}
If the group $\Gamma(E)$ is amenable then $E$ is testable.
\end{customthm}
It is worth noting that the Cayley graph $C$ of the group $\Gamma(E)$ belongs to $\GSOL_E$, and that if $\iota(C)=0$ then $\iota(G)=0$ for each $G\in\GSOL_E$. Thus the hypothesis of Theorem 1 can be reduced to an assumption about the Cayley graph of $\Gamma(E)$ only, rather than the family $\GSOL_E$ of graphs.

The group $\Gamma(E)$ has property $\ptau$ if and only if $\inf\{h(G)\mid G\in \FGSOL_E\} > 0$. Furthermore, $\FGSOL_E$ is infinite if and only if the group $\Gamma(E)$ has infinitely many finite quotients. Thus, Theorem \ref{thm:MainNegative} is equivalent to the following theorem:
\begin{customthm}{2'}[Main negative theorem in group terms]\label{thm:MainNegativePrime}
If the group $\Gamma(E)$ has property $\ptau$ and infinitely many finite quotients then $E$ is non-testable.
\end{customthm}


The rest of this section discusses applications of Theorems \ref{thm:MainPositive} and \ref{thm:MainNegative} and other aspects of testability related to the connection to group theory.

\subsection{Theorem \ref{thm:MainPositive} yields testable instable systems of relations}\label{subsec:testableInstable}

The class of amenable groups is vast and contains all solvable groups, and thus, by Theorem \ref{thm:MainPositivePrime}, $E$ is testable whenever $\Gamma(E)$ is solvable. But even in this case, $E$ is not necessarily stable.
Indeed, \cite[Theorem 1.3(ii)]{BLT} provides a group-theoretic condition $\text{($\ast$)}$ such that if $\Gamma(E)$ is amenable then $E$ is stable if and only if $\Gamma(E)$ satisfies $\text{($\ast$)}$.
Using this characterization of stability, \cite[Theorem 1.2(iii)]{BLT} provides an instable system of relations $E_p$, for each prime number $p$, such that $\Gamma(E_p)$ is solvable. But $E_p$ is testable by Theorem \ref{thm:MainPositivePrime}. Thus we have infinitely many examples of instable testable systems of relations. An explicit description of $E_p$ and $\Gamma(E_p)$ is given in Appendix \ref{appendix:abels}. 

\subsection{Theorem \ref{thm:MainNegative} yields non-testable systems of relations}\label{subsec:non-testable}

The class of finitely-presented groups with property $\ptau$ is also vast, and contains all finitely-presented groups that have property $\T$.
For example, the finitely-presented group $\SL_{m}\ZZ$, $m\geq 3$, has infinitely many finite quotients and has property $\T$, and thus, by Theorem \ref{thm:MainNegativePrime}, $E$ is non-testable if $\Gamma(E)\cong\SL_{m}\ZZ$.
In Appendix \ref{appendix:SL}
we describe a system $E_{m}$ such that $\Gamma\left(E_{m}\right)\cong\SL_{m}\ZZ$.

\subsection{Query efficiency and group theory}\label{subsec:QuantitativeGroups}

Let $E_{1}$ and $E_{2}$ be systems of relations such that
$\Gamma\left(E_{1}\right)\cong\Gamma\left(E_{2}\right)$.
Proposition \ref{prop:stability-group-property}
states that if $E_1$ is stable then so is $E_{2}$. Furthermore, \cite[Proposition 1.11]{BeckerMosheiff}
strengthens Proposition \ref{prop:stability-group-property} by showing
that if $E_{1}$ is $q\left(\eps\right)$-stable then $E_{2}$ is
$cq\left(\eps\right)$-stable for a constant $c=c\left(E_{1},E_{2}\right)$
(except for in the trivial case where $\GSOL_{E_{1}}\left(n\right)=\calG_S(n)$ for all $n\in \NN$).

Assuming that $E_1$ is testable (and thus, so is $E_2$), Proposition \ref{prop:DistinguishabilityIsAGroupProperty} describes a quantitative relationship between the statistical-distinguishability parameters of $E_1$ and $E_2$. It would be interesting to find an explicit quantitative relationship between the query complexities required to test these systems of relations. Such a result may be relevant to the problems discussed in Section \ref{subsec:queryComplexity}.

In the context of Problem \ref{prob:MainPositiveTheoremQuantitative}, it is worthwhile to mention the \emph{F\o{}lner function} \cite{VershikAmenability} of the group $\Gamma\left(E\right)$
with respect to $S$. Informally, assuming that $E$ satisfies the hypothesis of Theorem \ref{thm:MainPositive}, the F\o{}lner function measures ``how quickly'' $E$ does so. We conjecture that $E$ as above is $q\left(\eps\right)$-testable for some function $q(\eps)$, given in terms of $\left|S\right|$, $\sum_{w\in R_{E}}|w|$ and the F\o{}lner function.

\subsection{Uniform testability and group theory}\label{subsec:UniformAndGroups}
For systems $E_{1}$ and $E_{2}$ such that $\Gamma\left(E_{1}\right)\cong\Gamma\left(E_{2}\right)$,
if $E_{1}$ is uniformly testable (see Section \ref{sec:unfirom-testability}) then the same is true for $E_{2}$.
This follows from the method of Section \ref{sec:TestabilityIsAGroupProperty}
(see the explicit bounds in Proposition \ref{prop:DistinguishabilityIsAGroupProperty}).
It is interesting to study which systems of relations are uniformly testable. In particular, we would like to know if there is a uniform version of Theorem \ref{thm:MainPositive}, as asked in Problem \ref{prob:UniformMainPositive}. We suspect that if the F\o{}lner function of $\Gamma\left(E\right)$ is computable then $E$ is uniformly testable.

\section{\label{sec:kazhdan-not-testable}Proof of Theorem \ref{thm:MainNegative}}

Here we prove Theorem \ref{thm:MainNegative}, which is equivalent,
in light of Theorem \ref{thm:LSMUniversal}, to the following result. 
\begin{thm}
	\label{thm:kazhdan-not-distinguishable}Fix a finite alphabet $S$. Let $E$ be a system of relations over $S^{\pm}$ such that $\FGSOL_E$ is infinite and \begin{equation}\label{eq:propTau}
		\inf\{h(G)\mid G\in \FGSOL_E\}>0.
	\end{equation} Then $E$ is not statistically distinguishable. 
\end{thm}

The rest of this section is devoted to proving Theorem \ref{thm:kazhdan-not-distinguishable}.
Our proof strengthens the argument of \cite[Theorem 1.4]{BeckerLubotzky},
which proves, under similar assumptions, that $E$ is not stable.

We start with a proof sketch. Let $E$ be as in Theorem \ref{thm:kazhdan-not-distinguishable}.
For every $P\in\FinSubsets\left(F_{S}\right)$, we need to produce
a sequence of graphs $\left(G_{m}\right)_{m=1}^{\infty}$, $G_{m}\in\calG_{S}\left(n_{m}\right)$,
$n_{m}\in\NN$, such that
\begin{equation}
	d_{\TV}\left(N_{G_{m},P},N_{G'_{m},P}\right)\overset{m\to\infty}{\longrightarrow}0\label{eq:tau-sketch-1}
\end{equation}
for some $G'_{m}\in\GSOL_{E}\left(n_{m}\right)$, but
\begin{equation}
	d^{\ham}\left(G_{m},G''_{m}\right)\geq\eps_{0}\quad\forall m\in\NN\,\,\forall G''_{m}\in\GSOL_{E}\left(n_{m}\right)\,\,\text{,}\label{eq:tau-sketch-2}
\end{equation}
where $\eps_{0}>0$ depends only on $E$. That is, $G_{m}$ is close
to some solution for $E$ under the $P$-local metric, but far from
every solution under the global metric (see Remark \ref{rem:LocalVsGlobalMetric}).

To produce $G_{m}$, we take a connected graph $\tilde{G}_{m}\in\GSOL_{E}\left(n_{m}+1\right)$,
$n_{m}\in\NN$, $n_{m}\overset{m\to\infty}{\longrightarrow}\infty$,
make a local change such that $n_{m}+1$ becomes an isolated vertex,
and then remove the vertex $n_{m}+1$ and obtain $G_{m}\in\calG_{S}\left(n_{m}\right)$.
Then
\begin{equation}
	d_{\TV}\left(N_{G_{m},P},N_{\tilde{G}_{m},P}\right)\overset{m\to\infty}{\longrightarrow}0\,\,\text{,}\label{eq:tau-sketch-3}
\end{equation}
but, as shown by Lemma \ref{lem:restriction-and-cheeger} below, 
\begin{equation}
	d^{\ham}\left(G_{m},G''_{m}\right)\geq\frac{1}{2\left|S\right|}h\left(E\right)\quad\forall G''_{m}\in\GSOL_{E}\left(n_{m}\right)\label{eq:tau-sketch-4}
\end{equation}
for all large enough $m$, where
$$h(E) = \inf\{h(G)\mid G\in \FGSOL_E\}.$$
Set $\eps_{0}=\frac{1}{2}h(E)$. By \eqref{eq:propTau}, $\eps_{0}>0$, and thus (\ref{eq:tau-sketch-3})
and (\ref{eq:tau-sketch-4}) almost prove the desired (\ref{eq:tau-sketch-1})
and (\ref{eq:tau-sketch-2}), with the caveat that $\tilde{G}_{m}\in\GSOL_{E}\left(n_{m}+1\right)$,
while we need this graph to be in $\GSOL_{E}\left(n_{m}\right)$.
Lemma \ref{lem:action-size-changer} helps us overcome this final difficulty.
We now give a complete proof based on the proof sketch above.


For $n\geq2$, define a function $\res_{n}\colon\calG_{S}\left(n\right)\to\calG_{S}\left(n-1\right)$
by letting $\hat{G}=\res_{n}\left(G\right)$ be the graph
on the vertex set $\left[n-1\right]$ such that
\[
s_{\hat{G}}x=\begin{cases}
	s_{G}x & s_{G}x\ne n\\
	s_{G}s_{G}x & s_{G}x=n
\end{cases}
\]
for all $x\in\left[n-1\right]$ and $s\in S$. That is, the edge set
of $\res_{n}\left(G\right)$ consists of all edges of $\hat{G}$ that
are not incident to the vertex $n$, and all edges of the form $s_{G}^{-1}n\overset{s}{\rightedge}s_{G}n$
for all $s\in S$ such that $s_{G}n\neq n$.

The following lemma shows that $\res_{n}\left(G\right)$ is far from
$\GSOL_{E}\left(n-1\right)$ whenever $E$ satisfies $\eqref{eq:propTau}$.
The lemma and its proof can be seen as a combinatorial version of ideas from
\cite{BeckerLubotzky}.
\begin{lem}
	\label{lem:restriction-and-cheeger}Let $E$ be a system of relations
	over $S^{\pm}$. Take a connected graph $G\in\GSOL_{E}\left(n\right)$,
	$n\ge2$, and denote $\hat{G}=\res_{n}\left(G\right)$. Then
	\[
	d^{\mathrm{H}}\left(\hat{G},G'\right)\geq\frac{h\left(E\right)}{\left|S\right|}-\frac{1}{n-1}\quad\forall G'\in\GSOL_{E}\left(n-1\right)\,\,\text{.}
	\]
\end{lem}

\begin{proof}
	Take $G'\in\GSOL_{E}\left(n-1\right)$. Consider the product $S$-graph
	$G'\times G$. This is the $S$-graph on the vertex set $\left[n-1\right]\times\left[n\right]$
	such that
	\[
	s_{G'\times G}\left(x,y\right)=\left(s_{G'}x,s_{G}y\right)\,\,\text{.}
	\]
	
	Write $X_{1},\dotsc,X_{c}$ for the connected components of $G'\times G$,
	and let $D=\left\{ \left(x,x\right)\mid x\in\left[n-1\right]\right\} \subset\left[n-1\right]\times\left[n\right]$.
	We claim that
	\begin{equation}
		\left|D\cap X_{i}\right|\leq\frac{1}{2}\left|X_{i}\right|\quad\forall1\leq i\leq c\,\,\text{.}\label{eq:cheeger-lemma-1}
	\end{equation}
	Let $1\leq i\leq c$. If $D\cap X_{i}=\emptyset$ we are done. Otherwise,
	fix an arbitrary $\left(x,x\right)\in X_{i}$, $x\in\left[n-1\right]$.
	The connected component $Y$ of $x$ in $G'$ clearly has at most
	$n-1$ vertices. On the other hand, the connected component of $x$
	in $G$ has exactly $n$ vertices because $G$ is connected. Accordingly,
	$\left[F_{S}:\stab_{G'}\left(x\right)\right]=\left|Y\right|\leq n-1$ and $\left[F_{S}:\stab_{G}\left(x\right)\right]=n$ by the Orbit--Stabilizer Theorem\footnote{
		The theorem says, in our terminology, that the index $\left[F_S: \stab_{H}(x)\right]$ of the subgroup $\stab_{H}(x)$ of $F_S$ is equal to the number of vertices of $H$ for every connected $S$-graph $H$ and vertex $x$ of $H$.
	},
	and so $\stab_{G'}\left(x\right)$
	is not contained in $\stab_{G}\left(x\right)$. Thus, the inclusion
	in
	\[
	\stab_{G'\times G}\left(\left(x,x\right)\right)=\stab_{G'}\left(x\right)\cap\stab_{G}\left(x\right)\subset\stab_{G'}\left(x\right)
	\]
	is strict, and therefore $m\coloneqq\left[\stab_{G'}\left(x\right):\stab_{G'\times G}\left(\left(x,x\right)\right)\right]$
	is at least $2$. Using the Orbit--Stabilizer Theorem again,
	
	\[
	\left|X_{i}\right|=\left[F_{S}:\stab_{G'\times G}\left(\left(x,x\right)\right)\right]=m\left[F_{S}:\stab_{G'}\left(x\right)\right]\geq2\left|Y\right|\,\,\text{.}
	\]
	Thus (\ref{eq:cheeger-lemma-1}) follows since $\left|D\cap X_{i}\right|\leq\left|Y\right|$
	(because $Y$ is image of $X_{i}$ under the projection onto the first
	coordinate $\left[n-1\right]\times\left[n\right]\to\left[n-1\right]$,
	and this projection is injective on $D$).
	
	Since the graph $X_{i}$ is a connected solution for $E$ (namely, it is in $\FGSOL_E$), (\ref{eq:cheeger-lemma-1})
	implies that
	\[
	\sum_{s\in S}\left|\left(s_{G'\times G}\left(D\cap X_{i}\right)\right)\setminus\left(D\cap X_{i}\right)\right|\geq h\left(X_{i}\right)\left|D\cap X_{i}\right|\quad\forall1\leq i\leq c.
	\]
	Thus
	\begin{align}
		\sum_{s\in S}\left|\left(s_{G'\times G}D\right)\setminus D\right| & =\sum_{i=1}^{c}\sum_{s\in S}\left|\left(s_{G'\times G}\left(D\cap X_{i}\right)\right)\setminus\left(D\cap X_{i}\right)\right|\nonumber \\
		& \geq\sum_{i=1}^{c}h\left(X_{i}\right)\left|D\cap X_{i}\right|\ge h\left(E\right)\underbrace{\left|D\right|}_{=n-1}\,\,\text{,}\label{eq:tau-1}
	\end{align}
	where the equality on the first line follows since $X_{1},\dotsc,X_{c}$
	are disjoint and $s_{G'\times G}X_{i}=X_{i}$ for all $1\leq i\leq c$.
	
	On the other hand,
	\begin{align}
		\sum_{s\in S}\left|\left(s_{G'\times G}D\right)\setminus D\right| & =\sum_{s\in S}\left|\left\{ x\in\left[n-1\right]\mid s_{G'}x\ne s_{G}x\right\} \right|\nonumber \\
		& \leq\sum_{s\in S}\left|\left\{ x\in\left[n-1\right]\mid\text{\ensuremath{s_{G'}x\ne s_{\hat{G}}x} or \ensuremath{s_{\hat{G}}x\neq s_{G}x}}\right\} \right|\nonumber \\
		& \leq\sum_{s\in S}\left(\underbrace{\left|\left\{ x\in\left[n-1\right]\mid s_{G'}x\neq s_{\hat{G}}x\right\} \right|}_{=\left(n-1\right)d^{\ham}\left(\hat{G},G'\right)}+\underbrace{\left|\left\{ x\in\left[n-1\right]\mid\ensuremath{s_{\hat{G}}x\neq s_{G}x}\right\} \right|}_{\leq1}\right)\nonumber \\
		& \le\left|S\right|\left(\left(n-1\right)d^{\ham}\left(\hat{G},G'\right)+1\right)\,\,\text{.}\label{eq:tau-2}
	\end{align}
	The claim follows from (\ref{eq:tau-1}) and (\ref{eq:tau-2}).
\end{proof}
Let $G\in\calG_{S}\left(n\right)$, write $\hat{G}=\res_{n}\left(G\right)$,
and take $P\in\FinSubsets\left(F_{S}\right)$. Denote 
\[
C=\left\{ v_{G}^{-1}n\mid v\text{ is a suffix of at least one }w\in P\right\} \,\,\text{.}
\]
Then $w_{G}x=w_{\hat{G}}x$ for each $x\in\left[n-1\right]\setminus C$,
and thus $\stab_{G}\left(x\right)\cap P=\stab_{\hat{G}}\left(x\right)\cap P$.
Since $\left|C\right|\le\left(\TotalSize\left(P\right)\right)^{2}$,
one can easily deduce the following lemma.
\begin{lem}
	\label{lem:resPreservesLocalStatistics} Let $G\in\calG_{S}\left(n\right)$,
	write $\hat{G}=\res_{n}\left(G\right)$, and take $P\in\FinSubsets\left(F_{S}\right)$.
	Then 
	\[
	d_{\TV}\left(N_{G,P},N_{\hat{G},P}\right)\le O_{P}\left(\frac{1}{n}\right)\,\,\text{.}
	\]
\end{lem}

Let $P\in\FinSubsets\left(F_{S}\right)$. For an $S$-graph $G$,
we view $N_{G,P}$ as a vector in $\RR^{\Subsets\left(P\right)}$,
belonging to the compact set $\Prob\left(\Subsets(P)\right)\coloneqq\left\{ f\colon\Subsets\left(P\right)\to\left[0,1\right]\mid\sum_{Q\in\Subsets\left(P\right)}f\left(Q\right)=1\right\} $.
For a sequence of $S$-graphs $\left(G_{k}\right)_{k=1}^{\infty}$,
the sequence $\left(N_{G_{k},P}\right)_{k=1}^{\infty}$ is said to
\emph{converge} if it converges as a sequence of vectors in $\RR^{\Subsets\left(P\right)}$.
Intuitively, this means that the $S$-graphs $\left(G_{k}\right)_{k=1}^{\infty}$
tend toward having the same $P$-local statistics. This notion of
convergence is used in the proof of Theorem \ref{thm:kazhdan-not-distinguishable}
below. We use the fact that $\prod_{P\in\FinSubsets\left(F_{S}\right)}\Prob\left(\Subsets(P)\right)$
is compact\footnote{This is equivalent to the fact that the space of $S$-graphs
is compact under Benjamini--Schramm convergence, and to the fact
that the space $\IRS\left(F_{S}\right)$ of invariant random subgroups is compact.}.

We also use the following lemma.
\begin{lem}
	\label{lem:action-size-changer}\cite[Lemma 7.6]{BLT} Let $E$ be
	system of relations over $S^{\pm}$, and let $\left(G_{k}\right)_{k=1}^{\infty}$,
	$G_{k}\in\GSOL_{E}\left(l_{k}\right)$, $l_{k}\in\NN$, be a sequence
	of $S$-graphs such that the local statistics sequence $\left(N_{G_{k},P}\right)_{k=1}^{\infty}$
	converges for each $P\in\FinSubsets\left(F_{S}\right)$. Take a sequence
	$\left(m_{k}\right)_{k=1}^{\infty}$ of integers such that $m_{k}\overset{k\rightarrow\infty}{\longrightarrow}\infty$.
	Then there is a sequence $\left(H_{k}\right)_{k=1}^{\infty}$, $H_{k}\in\GSOL_{E}\left(m_{k}\right)$,
	such that $\left(N_{G_{k},P}\right)_{k=1}^{\infty}$ and $\left(N_{H_{k},P}\right)_{k=1}^{\infty}$
	converge to the same limit for each $P\in\FinSubsets\left(F_{S}\right)$. 
\end{lem}

\begin{proof}
	[Proof of Theorem \ref{thm:kazhdan-not-distinguishable}]
	
	By our assumption that $\FGSOL_E$ is infinite
	(and thus contains graphs of unbounded cardinality),
	and the compactness of $\prod_{P\in\FinSubsets\left(F_{S}\right)}\Prob\left(\Subsets(P)\right)$,
	there is a sequence of connected $S$-graphs $\left(G_{k}\right)_{k=1}^{\infty}$,
	$G_{k}\in\GSOL_{E}\left(n_{k}\right)$, $n_{k}\overset{k\rightarrow\infty}{\longrightarrow}\infty$,
	such that $\left(N_{G_{k},P}\right)_{k=1}^{\infty}$ converges for
	each $P\in\FinSubsets\left(F_{S}\right)$. Let $\hat{G}_{k}=\res_{n_{k}}\left(G_{k}\right)$.
	By Lemma \ref{lem:resPreservesLocalStatistics},
	\[
	\lim_{k\to\infty}N_{\hat{G}_{k},P}=\lim_{k\to\infty}N_{G_{k},P}\,\,\text{.}
	\]
	Write $h=h\left(E\right)$. By Lemma \ref{lem:restriction-and-cheeger},
	\begin{equation}
		\hat{G}_{k}\in\GSOL_{E}^{\geq h/\left|S\right|-1/\left(n_{k}-1\right)}\left(n_{k}-1\right)\qquad\forall k\in\NN\,\,\text{.}\label{eq:property-t-bounded-away-from-solutions}
	\end{equation}
	By applying Lemma \ref{lem:action-size-changer} to the sequence $\left(G_{k}\right)_{k=1}^{\infty}$,
	with $m_{k}=n_{k}-1$, we obtain a sequence of graphs $\left(H_{k}\right)_{k=1}^{\infty}$
	such that
	\begin{equation}
		H_{k}\in\GSOL_{E}\left(n_{k}-1\right)\label{eq:property-t-H-is-a-solution}
	\end{equation}
	and
	\begin{equation}
		\lim_{k\to\infty}N_{H_{k},P}=\lim_{k\to\infty}N_{G_{k},P}\quad\forall P\in\FinSubsets\left(F_{S}\right)\,\,\text{.}\label{eq:property-t-similar-statistics}
	\end{equation}
	Thus
	\begin{equation}
		\lim_{k\to\infty}N_{H_{k},P}=\lim_{k\to\infty}N_{\hat{G}_{k},P}\quad\forall P\in\FinSubsets\left(F_{S}\right)\,\,\text{.}\label{eq:property-t-similar-statistics-2}
	\end{equation}
	
	The system $E$ is not statistically distinguishable by (\ref{eq:property-t-bounded-away-from-solutions}),
	(\ref{eq:property-t-H-is-a-solution}) and (\ref{eq:property-t-similar-statistics-2}).
\end{proof}

\section{\label{sec:amenable-are-testable}Proof of Theorem \ref{thm:MainPositive}}

Let
$S$ be a finite alphabet and let $E$ be a system of relations over
$S^{\pm}$.

The proof of Theorem \ref{thm:MainPositive} relies on two deep results: The Ornstein--Weiss Theorem
\cite{OrnsteinWeiss} (see also \cite{ConnesFeldmanWeiss}) and the
Newman--Sohler Theorem \cite{NewmanSohler2013} (see also \cite[Theorem 5]{Elek2012}).
These theorems are used together in \cite[Proposition 6.8]{BLT} in
a simple manner to prove the following theorem (stated here in the
language of the present paper).
\begin{thm}
	\label{thm:Newman-Sohler-Ornstein-Weiss}Assume that $\iota(G)=0$ for every $G\in \GSOL_E$. Then for every $\eps>0$ there are $r\in\NN$ and $\delta>0$
	such that $d^{\ham}\left(G,G'\right)<\eps$ whenever $n\in\NN$, $G\in\GSOL_{E}\left(n\right)$,
	$G'\in\calG_{S}\left(n\right)$ and $d_{\TV}\left(N_{G,B_{r}},N_{G',B_{r}}\right)<\delta$
	(where $B_{r}=\left\{ w\in F_{S}\mid\left|w\right|\leq r\right\} $).
\end{thm}

Theorem \ref{thm:MainPositive} follows immediately
from Theorems \ref{thm:LSMUniversal} and \ref{thm:Newman-Sohler-Ornstein-Weiss}.
\begin{proof}
	[Proof of Theorem \ref{thm:MainPositive}]
	
	Assume that $\iota(G) = 0$ for every $G\in \GSOL_E$.
	By Theorem \ref{thm:Newman-Sohler-Ornstein-Weiss},
	there are functions $r\colon\RR_{>0}\to\NN$ and $\delta\colon\RR_{>0}\to\RR_{>0}$
	such that $E$ is $\left(B_{r\left(\eps\right)},\delta\left(\eps\right)\right)$-statistically-distinguishable.
	Thus $E$ is testable by Theorem \ref{thm:LSMUniversal}.
\end{proof}

\section{\label{sec:TestabilityIsAGroupProperty}Proof of Proposition \ref{prop:testability-group-property}}

Fix finite alphabets $S=\left\{ s_{1},\dotsc,s_{d_{1}}\right\} $
and $T=\left\{ t_{1},\dotsc,t_{d_{2}}\right\} $, and consider two
systems of relations $E_{1}$ and $E_{2}$ over $S^{\pm}$ and $T^{\pm}$,
respectively, such that the groups $\Gamma\left(E_{1}\right)$ and
$\Gamma\left(E_{2}\right)$ are isomorphic. By Theorem \ref{thm:LSMUniversal},
in order to prove Proposition \ref{prop:testability-group-property},
we need to prove that if $E_{2}$ is statistically distinguishable
then so is $E_{1}$. This is achieved by Proposition \ref{prop:DistinguishabilityIsAGroupProperty}
below. Our strategy is to show how to map $\calG_S(n)$ to $\calG_T(n)$, and vice versa, for all $n\in \NN$, in a way that preserves certain desirable properties. We define these maps in Sections \ref{subsec:HomView} and \ref{sec:Lambda1AndLambda2}, and then prove the relevant invariants in Section \ref{subsec:PropertiesOfLambdaStars}

Given a finite set of words $P$, denote 
\[
\TotalSize\left(P\right)=\sum_{x\in P}\left|x\right|\,\,\text{.}
\]

\begin{prop}
\label{prop:DistinguishabilityIsAGroupProperty}If $E_{2}$ is $\left(P_{2}\left(\eps\right),\delta_{2}\left(\eps\right)\right)$-statistically-distinguishable
for $P_{2}\colon\RR_{>0}\to\FinSubsets\left(F_{T}\right)$ and $\delta_{2}\colon\RR_{>0}\to\RR_{\geq0}$,
then $E_{1}$ is $\left(P_{1}\left(\eps\right),\delta_{1}\left(\eps\right)\right)$-statistically-distinguishable
for some $P_{1}\colon\RR_{>0}\to\FinSubsets\left(F_{S}\right)$ and $\delta_{1}\colon\RR_{>0}\to\RR_{\geq0}$
such that
\begin{equation}
\TotalSize\left(P_{1}\left(\eps\right)\right)=O_{E_{1},E_{2}}\left(\TotalSize\left(P_{2}\left(\eps\right)\right)\right)\label{eq:group-property-statement-1}
\end{equation}
 and 
\begin{equation}
\delta_{1}\left(\eps\right)=\Omega_{E_{1},E_{2}}\left(\min\left(\delta_{2}\left(\eps\right),\eps\right)\right)\label{eq:group-property-statement-2}
\end{equation}
for every $\eps>0$.

Explicitly,
\[
P_{1}\left(\eps\right)=\lambda_{2}\left(P_{2}\left(\eps\right)\right)\cup R_{E_{1}}
\]
and 
\[
\delta_{1}\left(\eps\right)=\min\left(\delta_{2}\left(\frac{\eps}{2C_{1}}\right),\frac{\eps}{2C_{2}}\right)\,\,\text{,}
\]
for $\lambda_{2}$ as defined below and positive constants $C_{1}$
and $C_{2}$ that depend on $E_{1}$ and $E_{2}$.
\end{prop}

The rest of this section is devoted to proving Proposition \ref{prop:DistinguishabilityIsAGroupProperty}.

\subsection{A homomorphism view of $\calG_S(n)$}\label{subsec:HomView}

For $n\in\NN$, let $\calH_{S}\left(n\right)$ be the set of homomorphisms
$F_{S}\to\Sym\left(n\right)$. For the sake of proving Proposition \ref{prop:DistinguishabilityIsAGroupProperty}, it will be convenient to encode an $S$-graph in  $\calG_S(n)$ as a homomorphism in $\calH_S(n)$, in a manner which we now describe.

For $G\in \calG_S(n)$,
let $f_{G}\in\calH_{S}\left(n\right)$ be the 
$F_{S}\to\Sym\left(n\right)$ homomorphism that maps $w\in F_S$ to the permutation $i\mapsto w_Gi$. Note that the map $G\mapsto f_{G}\colon\calG_S(n)\to\calH_{S}\left(n\right)$
is a bijection\footnote{
This can be seen easily using the
universal property of the free group $F_S$, which amounts to a bijection
$(\Sym(n))^d\to\calH_S(n)$, as recalled in Appendix 
\ref{app:FreeGroupAndPresentations},
and the bijection
$\overline{\sigma}\mapsto G_{\overline{\sigma}}\colon(\Sym(n))^d\to\calG_S(n)$
defined in Section \ref{sec:GraphView}.}

We translate some of the notions from
the introduction to the language of homomorphisms. Given $f,g\in\calH_{S}\left(n\right)$,
define 
\[
d_{n}^{\ham}\left(f,g\right)=d^{\ham}\left(f,g\right)=\sum_{s\in S}\frac{1}{n}\left|\left\{ x\in\left[n\right]\mid f\left(s\right)x\ne g\left(s\right)x\right\} \right|\,\,\text{.}
\]
Let 
\[
\HSOL_{E}\left(n\right)=\left\{ f\in\calH_{S}\left(n\right)\mid f\left(w\right)=\id\,\,\forall w\in R_{E}\right\} \,\,\text{,}
\]
\[
\HSOL_{E}^{\ge\eps}\left(n\right)=\left\{ f\in\calH_{S}\left(n\right)\mid d^{\ham}\left(f,f'\right)\ge\eps\,\,\forall f'\in\HSOL_{E}\left(n\right)\right\} \,\,\text{.}
\]
For $x\in\left[n\right]$, let $\stab_{f}\left(x\right)=\left\{ w\in F_{S}\mid f\left(w\right)x=x\right\} $.
For $P\in\FinSubsets\left(F_{S}\right)$, let $N_{f,P}$ be the distribution
of $\stab_{f}\left(x\right)\cap P$, where $x$ is sampled uniformly
from $\left[n\right]$. Then $\HSOL_{E}\left(n\right)=\left\{ f_{G}\mid G\in\GSOL_{E}\left(n\right)\right\} $,
$\HSOL_{E}^{\ge\eps}\left(n\right)=\left\{ f_{G}\mid G\in\GSOL_{E}^{\geq\eps}\left(n\right)\right\} $,
$\stab_{f_{G}}\left(x\right)=\stab_{G}\left(x\right)$
and $N_{f_{G},P}=N_{G,P}$ for $G\in\calG_S(n)$. We also denote $\HSOL_{E}^{<\eps}\left(n\right)\coloneqq\calH_{S}\left(n\right)\setminus\HSOL_{E}^{\geq\eps}\left(n\right)$.

\subsection{The maps $\lambda_1^*$ and $\lambda_2^*$}\label{sec:Lambda1AndLambda2}
To prove Proposition \ref{prop:DistinguishabilityIsAGroupProperty}
we define two maps, $\lambda_{1}^{*}\colon\calH_{T}\left(n\right)\to\calH_{S}\left(n\right)$
and $\lambda_{2}^{*}\colon\calH_{S}\left(n\right)\to\calH_{T}\left(n\right)$,
that behave nicely (see Section \ref{subsec:PropertiesOfLambdaStars})
with respect to the spaces $\HSOL_{E_{1}}\left(n\right)$ and $\HSOL_{E_{2}}\left(n\right)$,
and also with respect to $\HSOL_{E_{1}}^{<\eps}\left(n\right)$ and
$\HSOL_{E_{2}}^{<\eps}\left(n\right)$, $\eps>0$.

We begin with the definitions of $\lambda_{1}^{*}$ and $\lambda_{2}^{*}$.
Write $\pi_{1}\colon F_{S}\twoheadrightarrow\Gamma\left(E_{1}\right)$
and $\pi_{2}\colon F_{T}\twoheadrightarrow\Gamma\left(E_{2}\right)$
for the quotient maps (see Appendix \ref{app:FreeGroupAndPresentations}),
and fix an isomorphism $\theta\colon\Gamma\left(E_{1}\right)\to\Gamma\left(E_{2}\right)$.
We ``lift'' the group isomorphisms $\theta\colon\Gamma\left(E_{1}\right)\to\Gamma\left(E_{2}\right)$
and $\theta^{-1}\colon\Gamma\left(E_{2}\right)\to\Gamma\left(E_{1}\right)$
to group homomorphisms

\[
\text{\ensuremath{\lambda_{1}\colon F_{S}\to F_{T}} and \ensuremath{\lambda_{2}\colon F_{T}\to F_{S}}\,\,\ensuremath{.}}
\]
More precisely, we fix homomorphisms $\lambda_{1}$ and $\lambda_{2}$
such that each of the two squares in the following diagram commutes
(i.e., $\pi_{2}\circ\lambda_{1}=\theta\circ\pi_{1}$ and $\pi_{1}\circ\lambda_{2}=\theta^{-1}\circ\pi_{2}$):
\begin{equation}
\xymatrix{F_{S}\ar[rr]^{\lambda_{1}}\ar[d]^{\pi_{1}} &  & F_{T}\ar[rr]^{\lambda_{2}}\ar[d]^{\pi_{2}} &  & F_{S}\ar[d]^{\pi_{1}}\\
\Gamma\left(E_{1}\right)\ar[rr]^{\theta} &  & \Gamma\left(E_{2}\right)\ar[rr]^{\theta^{-1}} &  & \Gamma\left(E_{1}\right)
}
\,\,\text{.}\label{eq:GroupPropDiagram}
\end{equation}
Such maps $\lambda_{1}$ and $\lambda_{2}$ exist (but are generally
not unique). Indeed, for $1\leq i\leq d_{1}$ we set $\lambda_{1}\left(s_{i}\right)\in F_{T}$
to be an arbitrary word such that $\pi_{2}\left(\lambda_{1}\left(s_{i}\right)\right)=\theta\left(\pi_{1}\left(s_{i}\right)\right)$,
and note that $\lambda_{1}$ extends uniquely to a group homomorphism.
We construct $\lambda_{2}$ similarly.

The maps $\lambda_{1}$ and $\lambda_{2}$ are not mutual inverses
in general, but they enjoy inverse-like properties. By the commutativity
of the diagram, $\pi_{1}\circ\lambda_{2}\circ\lambda_{1}=\theta^{-1}\circ\pi_{2}\circ\lambda_{1}=\theta^{-1}\circ\theta\circ\pi_{1}=\pi_{1}$.
In particular, $\left(\lambda_{2}\circ\lambda_{1}\right)\left(s_{i}\right)\in F_{S}$
and $s_{i}\in F_{S}$ belong to the same left coset of $\lla R_{E_{1}}\rra$
in $F_{S}$ for each $1\leq i\leq d_{1}$. That is,
\begin{equation}
\left(\lambda_{2}\circ\lambda_{1}\right)\left(s_{i}\right)=s_{i}\prod_{j=1}^{m_{i}}v_{i,j}r_{i,j}^{\eps_{i,j}}v_{i,j}^{-1}\,\,\text{,}\label{eq:pseudo-inverses}
\end{equation}
where $m_{i}\geq0$, $v_{i,j}\in F_{S}$, $r_{i,j}\in R_{E_{1}}$
and $\eps_{i,j}\in\left\{ \pm1\right\} $. Set $Q_{i}\coloneqq\left\{ v_{i,j}\mid1\leq j\leq m_{i}\right\} $.

For $n\in\NN$, the homomorphisms $\lambda_{1}\colon F_{S}\to F_{T}$
and $\lambda_{2}\colon F_{T}\to F_{S}$ give rise to maps
\begin{align*}
\lambda_{1}^{*}\colon\calH_{T}\left(n\right) & \to\calH_{S}\left(n\right)\,\,\text{and}\\
\lambda_{2}^{*}\colon\calH_{S}\left(n\right) & \to\calH_{T}\left(n\right)
\end{align*}
given by
\begin{align*}
\lambda_{1}^{*}h & =h\circ\lambda_{1}\qquad\forall h\in\calH_{T}\left(n\right)\,\,\text{and}\\
\lambda_{2}^{*}f & =f\circ\lambda_{2}\qquad\forall f\in\calH_{S}\left(n\right)\,\,\text{.}
\end{align*}

Under the bijections $\calH_{S}\left(n\right)\leftrightarrow\calG_{S}\left(n\right)$
and $\leftrightarrow\calH_{T}\left(n\right)\leftrightarrow\calG_{T}\left(n\right)$
of Section \ref{subsec:HomView}, $\lambda_{1}^{*}$ and $\lambda_{2}^{*}$
give rise to maps between $\calG_{S}\left(n\right)$
and $\calG_{T}\left(n\right)$. For this section, the homomorphism
view suffices, but the reader may find it instructive to spell out
the definitions of $\lambda_{1}^{*}$ and $\lambda_{2}^{*}$ in terms
of graphs.

\subsection{Properties of $\lambda_{1}^{*}$ and $\lambda_{2}^{*}$\label{subsec:PropertiesOfLambdaStars}}

We analyze the behavior of $\lambda_{1}^{*}$ and $\lambda_{2}^{*}$
with regard to both the global metric $d^{\ham}\left(f,g\right)$,
and the $P$-local metrics $d_{\TV}\left(N_{f,P},N_{g,P}\right)$,
$P\in\FinSubsets\left(F_{S}\right)$ (see Remark \ref{rem:LocalVsGlobalMetric}).
Due to symmetry, the claims in this section remain true if we swap
the roles of $S,E_{1},\lambda_{1}$ and $T,E_{2},\lambda_{2}$.

\subsubsection{The global metric}
\begin{defn}
Given $f\in\calH_{S}\left(n\right)$, let $\Bad_{E_{1}}\left(f\right)=\left\{ x\in\left[n\right]\mid\exists w\in R_{E_{1}}\,\,f\left(w\right)x\ne x\right\} \,\,$.
\end{defn}

\begin{rem}
\label{rem:LocalDefectAsNeighborhoodTVDistance}For $f\in\HSOL_{E_{1}}\left(n\right)$,
the probability distribution $N_{F,R_{E_{1}}}$ over $\Subsets\left(R_{E_{1}}\right)$
assigns probability $1$ to $R_{E_{1}}$. Thus, 
\[
d_{\TV}\left(N_{f,R_{E_{1}}},N_{g,R_{E_{1}}}\right)=\Pr_{x\sim \U\left(\left[n\right]\right)}\left(\stab_{g}\left(x\right)\cap R_{E_{1}}\ne R_{E_{1}}\right)=\frac{\left|\Bad_{E_{1}}\left(g\right)\right|}{n}
\]
for all $g\in\calH_{S}\left(n\right)$. 
\end{rem}

The following lemma shows that $\lambda_{1}^{*}$ and $\lambda_{2}^{*}$
enjoy inverse-like properties. More precisely, the lemma gives a tool
for bounding the distance between $f$ and $\lambda_{1}^{*}\lambda_{2}^{*}f$
for $f\in\calH_{S}\left(n\right)$.
\begin{lem}
\label{lem:pseudo-inverse}Let $f\in\calH_{S}\left(n\right)$. Then
$\left(\lambda_{1}^{*}\lambda_{2}^{*}f\right)\left(s_{i}\right)x=f\left(s_{i}\right)x$
for every $1\leq i\leq d_{1}$ and $x\in\left[n\right]$ such that
$x\notin\bigcup_{v\in Q_{i}}\left(f\left(v\right)\Bad_{E_{1}}\left(f\right)\right)$.
\end{lem}

\begin{proof}
For $1\leq i\leq d_{1}$ and $x\in\left[n\right]$,
\begin{align*}
\left(\lambda_{1}^{*}\lambda_{2}^{*}f\right)\left(s_{i}\right)x & =f\left(\left(\lambda_{2}\circ\lambda_{1}\right)\left(s_{i}\right)\right)x\\
 & =f\left(s_{i}\right)\left(\prod_{j=1}^{m_{i}}f\left(v_{i,j}\right)f\left(r_{i,j}^{\eps_{i,j}}\right)f\left(v_{i,j}^{-1}\right)\right)x\,\,\text{.} & \text{by (\ref{eq:pseudo-inverses})}
\end{align*}
Thus, $\left(\lambda_{1}^{*}\lambda_{2}^{*}f\right)\left(s_{i}\right)x=f\left(s_{i}\right)x$
if 
\[
f\left(r_{i,j}^{\eps_{i,j}}\right)f\left(v_{i,j}^{-1}\right)x=f\left(v_{i,j}^{-1}\right)x\qquad\forall1\leq j\leq m_{i}\,\,\text{.}
\]
The latter condition is equivalent to
\[
f\left(r_{i,j}\right)f\left(v_{i,j}^{-1}\right)x=f\left(v_{i,j}^{-1}\right)x\qquad\forall1\leq j\leq m_{i}\,\,\text{.}
\]
This holds whenever $f\left(v\right)^{-1}x\notin\Bad_{E_{1}}\left(f\right)$
for all $v\in Q_{i}$, i.e., when $x\notin\bigcup_{v\in Q_{i}}\left(f\left(v\right)\Bad_{E_{1}}\left(f\right)\right)$.
\end{proof}
We conclude the following.
\begin{cor}
\label{cor:InverseOnSolutions}The image of $\HSOL_{E_{1}}\left(n\right)$
under $\lambda_{2}^{*}$ is contained in $\HSOL_{E_{2}}\left(n\right)$.
Furthermore, the restriction of $\lambda_{2}^{*}$ to $\HSOL_{E_{1}}\left(n\right)$
is a bijection $\lambda_{2}^{*}\mid_{\HSOL_{E_{1}}\left(n\right)}\colon\HSOL_{E_{1}}\left(n\right)\to\HSOL_{E_{2}}\left(n\right)$
whose inverse is $\lambda_{1}^{*}\mid_{\HSOL_{E_{2}}\left(n\right)}$.
\end{cor}

\begin{proof}
Let $f\in\HSOL_{E_{1}}\left(n\right)$ and $v\in R_{E_{2}}$. Then
$\pi_{1}\left(\lambda_{2}\left(v\right)\right)=\theta^{-1}\left(\pi_{2}\left(v\right)\right)=\theta^{-1}\left(1_{\Gamma\left(E_{2}\right)}\right)=1_{\Gamma\left(E_{1}\right)}$
by (\ref{eq:GroupPropDiagram}). Thus $\lambda_{2}\left(v\right)\in\ker\pi_{1}=\lla E_{1}\rra$,
and so $\left(\lambda_{2}^{*}f\right)\left(v\right)=f\left(\lambda_{2}\left(v\right)\right)=1$.
Hence $\lambda_{2}^{*}f\in\HSOL_{E_{2}}\left(n\right)$. We proved
that 
\[
\lambda_{2}^{*}\left(\HSOL_{E_{1}}\left(n\right)\right)\subseteq\HSOL_{E_{2}}\left(n\right)\,\,\text{.}
\]
Similarly,
\[
\lambda_{1}^{*}\left(\HSOL_{E_{2}}\left(n\right)\right)\subset\HSOL_{E_{1}}\left(n\right)\,\,\text{.}
\]

Since $f\in\HSOL_{E_{1}}\left(n\right)$ we have $\Bad_{E_{1}}\left(f\right)=\emptyset$.
Consequently, Lemma \ref{lem:pseudo-inverse} implies that $\lambda_{1}^{*}\lambda_{2}^{*}f=f$.
Similarly, $\lambda_{2}^{*}\lambda_{1}^{*}h=h$ for $h\in\HSOL_{E_{2}}\left(n\right)$,
and the claim follows.
\end{proof}
If $f\in\calH_{S}\left(n\right)$ is not necessarily a solution for
$E_{1}$, but the set $\Bad_{E_{1}}\left(f\right)$ is small, the
following corollary shows that $\lambda_{1}^{*}\lambda_{2}^{*}f$
is close to $f$.
\begin{cor}
\label{cor:pseudo-inverse}For $f\in\calH_{S}(n)$, 
\[
d^{\ham}\left(f,\lambda_{1}^{*}\lambda_{2}^{*}f\right)\leq\left(\sum_{i=1}^{d_{1}}\left|Q_{i}\right|\right)\frac{\left|\Bad_{E_{1}}\left(f\right)\right|}{n}\,\,\text{.}
\]
\end{cor}

\begin{proof}
~
\begin{align*}
d^{\ham}\left(f,\lambda_{1}^{*}\lambda_{2}^{*}f\right) & =\sum_{i=1}^{d_{1}}d^{\ham}\left(f\left(s_{i}\right),\left(\lambda_{1}^{*}\lambda_{2}^{*}f\right)\left(s_{i}\right)\right)\\
 & \leq\sum_{i=1}^{d_{1}}\frac{1}{n}\left|\bigcup_{v\in Q_{i}}\left(f\left(v\right)\Bad_{E_{1}}\left(f\right)\right)\right| & \text{by Lemma \ref{lem:pseudo-inverse}}\\
 & \le\sum_{i=1}^{d_{1}}\sum_{v\in Q_{i}}\frac{1}{n}\left|f\left(v\right)\Bad_{E_{1}}\left(f\right)\right|\\
 & =\left(\sum_{i=1}^{d_{1}}\left|Q_{i}\right|\right)\frac{\left|\Bad_{E_{1}}\left(f\right)\right|}{n}\,\,\text{.}
\end{align*}
\end{proof}
Next, we study the interaction between $\lambda_{1}^{*}$ and the
Hamming metric.
\begin{lem}
\label{lem:pull-close-maps}Let $h,h'\in\calH_{T}\left(n\right)$.
Then $d^{\ham}\left(\lambda_{1}^{*}h,\lambda_{1}^{*}h'\right)\leq Cd^{\ham}\left(h,h'\right)$,
where $C=\left(\sum_{i=1}^{d_{1}}\left|\lambda_{1}\left(s_{i}\right)\right|\right)$.
\end{lem}

\begin{proof}
For $w\in F_{T}$, let $A_{w}=\left\{ x\in\left[n\right]\mid h\left(w\right)x=h'\left(w\right)x\right\} $
and $D_{w}=\left[n\right]\setminus A_{w}$. Write $\suff_{j}\left(w\right)$
for the suffix of length $j$ of $w$ (e.g., $\suff_{2}\left(ab^{-1}c^{-1}\right)=b^{-1}c^{-1}$).
Write $\last_{j}\left(w\right)$ for the $j$-th letter of $w$, counting
from the end (e.g., $\last_{2}\left(ab^{-1}c^{-1}\right)=b^{-1}$).

For $1\leq i\leq d_{1}$, 
\begin{align*}
 & d^{\ham}\left(\left(\lambda_{1}^{*}h\right)\left(s_{i}\right),\left(\lambda_{1}^{*}h'\right)\left(s_{i}\right)\right)\\
= & \frac{1}{n}\left|D_{\lambda_{1}\left(s_{i}\right)}\right|\\
\le & \frac{1}{n}\left|\bigcup_{j=1}^{\left|\lambda_{1}\left(s_{i}\right)\right|}\left(A_{\suff_{j-1}\left(\lambda_{1}\left(s_{i}\right)\right)}\cap D_{\suff_{j}\left(\lambda_{1}\left(s_{i}\right)\right)}\right)\right|\\
 & \qquad\text{\qquad\qquad\qquad\qquad(since \ensuremath{D_{w}\subseteq\bigcup_{j=1}^{\left|w\right|}A_{\suff_{j-1}\left(w\right)}\cap D_{\suff_{j}\left(w\right)}} for \ensuremath{w\in F_{T}}})\\
\le & \frac{1}{n}\sum_{j=1}^{\left|\lambda_{1}\left(s_{i}\right)\right|}\left|A_{\suff_{j-1}\left(\lambda_{1}\left(s_{i}\right)\right)}\cap D_{\suff_{j}\left(\lambda_{1}\left(s_{i}\right)\right)}\right|\\
= & \frac{1}{n}\sum_{j=1}^{\left|\lambda_{1}\left(s_{i}\right)\right|}\left|A_{\suff_{j-1}\left(\lambda_{1}\left(s_{i}\right)\right)}\cap h\left(\suff_{j-1}\left(\lambda_{1}\left(s_{i}\right)\right)\right)^{-1}D_{\last_{j}\left(\lambda_{1}\left(s_{i}\right)\right)}\right|\\
 & \qquad\qquad\qquad\qquad\qquad\text{(since \ensuremath{x\in A_{w}}}\text{ implies \ensuremath{x\in D_{sw}\Leftrightarrow x\in h\left(w\right){}^{-1}D_{s}} for \ensuremath{w\in F_{T},s\in S^{\pm}})}\\
\leq & \sum_{j=1}^{\left|\lambda_{1}\left(s_{i}\right)\right|}\underbrace{\frac{1}{n}\left|D_{\last_{j}\left(\lambda_{1}\left(s_{i}\right)\right)}\right|}_{\leq d^{\ham}\left(h,h'\right)}\\
\leq & \left|\lambda_{1}\left(s_{i}\right)\right|d^{\ham}\left(h,h'\right)\,\,\text{.}
\end{align*}

Thus
\[
d^{\ham}\left(\lambda_{1}^{*}h,\lambda_{1}^{*}h'\right)=\sum_{i=1}^{d_{1}}d^{\ham}\left(\left(\lambda_{1}^{*}h\right)\left(s_{i}\right),\left(\lambda_{1}^{*}h'\right)\left(s_{i}\right)\right)\leq\left(\sum_{i=1}^{d_{1}}\left|\lambda_{1}\left(s_{i}\right)\right|\right)d^{\ham}\left(h,h'\right)\,\,\text{.}
\]
\end{proof}
\begin{lem}
\label{lem:DistancePreservingBothSides}For $f\in\calH_{S}\left(n\right)$
and $h\in\calH_{T}\left(n\right)$,
\[
d^{\ham}\left(f,\lambda_{1}^{*}h\right)\leq C_{1}d^{\ham}\left(\lambda_{2}^{*}f,h\right)+C_{2}\frac{\left|\Bad_{E_{1}}\left(f\right)\right|}{n}\,\,\text{,}
\]
where $C_{1}=\sum_{i=1}^{d_{1}}\left|\lambda_{1}\left(s_{i}\right)\right|$
and $C_{2}=\sum_{i=1}^{d_{1}}\left|Q_{i}\right|$.
\end{lem}

\begin{proof}
By the triangle inequality, Corollary \ref{cor:pseudo-inverse} and
Lemma \ref{lem:pull-close-maps},

\begin{align*}
d^{\ham}\left(f,\lambda_{1}^{*}h\right) & \leq d^{\ham}\left(f,\lambda_{1}^{*}\lambda_{2}^{*}f\right)+d^{\ham}\left(\lambda_{1}^{*}\lambda_{2}^{*}f,\lambda_{1}^{*}h\right)\\
 & \leq\left(\sum_{i=1}^{d_{1}}\left|Q_{i}\right|\right)\frac{\left|\Bad_{E_{1}}\left(f\right)\right|}{n}+\left(\sum_{i=1}^{d_{1}}\left|\lambda_{1}\left(s_{i}\right)\right|\right)d^{\ham}\left(\lambda_{2}^{*}f,h\right)\,\,\text{.}
\end{align*}
\end{proof}

\subsubsection{The $P$-local metric}

For a probability distribution $\theta$ over a set $\Omega$ and
a function $\varphi\colon\Omega\to\Omega'$, write $\varphi_{*}\theta$
for the distribution of $\varphi\left(x\right)$ when $x\sim\theta$.
\begin{lem}
\label{lem:dTVInformationProcessing}Let $\theta$ and $\theta'$
be probability distributions over the finite sets $\Omega$ and $\Omega'$,
respectively, and take a function $\varphi\colon\Omega\to\Omega'$.
Then

\[
d_{\TV}\left(\varphi_{*}\theta,\varphi_{*}\theta'\right)\le d_{\TV}\left(\theta,\theta'\right)\,\,\text{.}
\]
\end{lem}

\begin{proof}
By the triangle inequality,
\begin{align*}
d_{\TV}\left(\varphi_{*}\theta,\varphi_{*}\theta'\right) & =\frac{1}{2}\sum_{y\in\Omega'}\left|\varphi_{*}\theta\left(y\right)-\varphi_{*}\theta'\left(y\right)\right|=\frac{1}{2}\sum_{y\in\Omega'}\left|\sum_{\substack{x\in\Omega\\
f(x)=y
}
}\theta\left(x\right)-\sum_{\substack{x\in\Omega\\
f(x)=y
}
}\theta'\left(x\right)\right|\\
 & \le\frac{1}{2}\sum_{y\in\Omega'}\sum_{\substack{x\in\Omega\\
f(x)=y
}
}\left|\theta\left(x\right)-\theta'\left(x\right)\right|=\frac{1}{2}\sum_{x\in\Omega}\left|\theta\left(x\right)-\theta'\left(x\right)\right|=d_{\TV}\left(\theta,\theta'\right).
\end{align*}
\end{proof}
In the rest of the section we make frequent use of the notation $N_{G,P}$
(see Definition \ref{def:stab-and-Nsigma-P}).
\begin{cor}
\label{cor:d_TVIncreasingInP}Let $S$ be a finite alphabet and take
$P,P^{*}\in\FinSubsets\left(F_{S}\right)$ such that $P^{*}\subseteq P$.
Then
\[
d_{\TV}\left(N_{G,P^{*}},N_{G',P^{*}}\right)\le d_{\TV}\left(N_{G,P},N_{G',P}\right)
\]
for all $n\in\NN$ and $G,G'\in\calG_S(n)$. 
\end{cor}

\begin{proof}
The claim follows from Lemma \ref{lem:dTVInformationProcessing} because
$N_{G,P^{*}}=\varphi_{*}N_{G,P}$
and $N_{G',P^{*}}=\varphi_{*}N_{G',P}$
for $\varphi\colon\Subsets\left(P\right)\to\Subsets\left(P^{*}\right)$
given by $\varphi\left(M\right)=M\cap P^{*}$.
\end{proof}
\begin{lem}
\label{lem:lambda2^*Preserevesd_P}Let $P_{2}\in\FinSubsets\left(F_{T}\right)$.
Then 
\[
d_{\TV}\left(N_{\lambda_{2}^{*}f,P_{2}},N_{\lambda_{2}^{*}f',P_{2}}\right)\le d_{\TV}\left(N_{f,\lambda_{2}\left(P_{2}\right)},N_{f',\lambda_{2}\left(P_{2}\right)}\right)
\]
for all $f,f'\in\calH_{S}\left(n\right)$.
\end{lem}

\begin{proof}
Define $\varphi\colon\Subsets\left(\lambda_{2}\left(P_{2}\right)\right)\to\Subsets\left(P_{2}\right)$
by 
\begin{equation}
\varphi\left(V\right)=\left\{ w\in P_{2}\mid\lambda_{2}\left(w\right)\in V\right\} =\lambda_{2}^{-1}\left(V\right)\cap P_{2}\quad\forall V\subset\lambda_{2}\left(P_{2}\right)\text{.}\label{eq:lambda2^*Prserervesd_PDefinitionOfPhi}
\end{equation}
By Lemma \ref{lem:dTVInformationProcessing}, it suffices to show
that $N_{\lambda_{2}^{*}f,P_{2}}=\varphi_{*}N_{f,\lambda_{2}\left(P_{2}\right)}$
and $N_{\lambda_{2}^{*}f',P_{2}}=\varphi_{*}N_{f',\lambda_{2}\left(P_{2}\right)}$.
Thus, it is enough to prove that $\text{\ensuremath{\stab}}_{\lambda_{2}^{*}h}\left(x\right)\cap P_{2}=\varphi\left(\stab_{h}\left(x\right)\cap\lambda_{2}\left(P_{2}\right)\right)$
for all $h\in\calH_{S}\left(n\right)$ and $x\in\left[n\right]$.

First, $\stab_{\lambda_{2}^{*}h}\left(x\right)=\lambda_{2}^{-1}\left(\stab_{h}\left(x\right)\right)$
because $\left(\lambda_{2}^{*}h\right)\left(w\right)x=x\iff h\left(\lambda_{2}\left(w\right)\right)x=x$
for $w\in F_{T}$. Thus,
\begin{align*}
\text{\ensuremath{\stab}}_{\lambda_{2}^{*}h}\left(x\right)\cap P_{2} & =\lambda_{2}^{-1}\left(\stab_{h}\left(x\right)\right)\cap P_{2}\\
 & =\lambda_{2}^{-1}\left(\stab_{h}\left(x\right)\cap\lambda_{2}\left(P_{2}\right)\right)\cap P_{2}\\
 & =\varphi\left(\stab_{h}\left(x\right)\cap\lambda_{2}\left(P_{2}\right)\right)\,\,\text{.}
\end{align*}
\end{proof}

\subsection{Proof of Proposition \ref{prop:DistinguishabilityIsAGroupProperty}}

Suppose that $E_{2}$ is $\left(P_{2}\left(\eps\right),\delta_{2}\left(\eps\right)\right)$-statistically-distinguishable.
We claim that $E_{1}$ is $\left(P_{1}\left(\eps\right),\delta_{1}\left(\eps\right)\right)$-statistically-distinguishable,
where

\[
P_{1}\left(\eps\right)=\lambda_{2}\left(P_{2}\left(\eps\right)\right)\cup R_{E_{1}}
\]
and 
\[
\delta_{1}\left(\eps\right)=\min\left(\delta_{2}\left(\frac{\eps}{2C_{1}}\right),\frac{\eps}{2C_{2}}\right)\,\,\text{,}
\]
for $C_{1}$ and $C_{2}$ as in Lemma \ref{lem:DistancePreservingBothSides}.
First, note that 
\begin{align*}
\TotalSize\left(P_{1}\left(\eps\right)\right) & =\sum_{w\in P_{1}\left(\eps\right)}\left|w\right|\\
 & \le\TotalSize\left(R_{E_{1}}\right)+\sum_{w\in P_{2}\left(\eps\right)}\left|\lambda_{2}\left(w\right)\right|\\
 & \le\TotalSize\left(R_{E_{1}}\right)+\sum_{w\in P_{2}\left(\eps\right)}\underbrace{\max_{t\in T}\left|\lambda_{2}\left(t\right)\right|}_{\eqqcolon C_{3}}\left|w\right|\\
 & =\TotalSize\left(R_{E_{1}}\right)+C_{3}\cdot\TotalSize\left(P_{2}\left(\eps\right)\right)\,\,\text{.}
\end{align*}
Hence (\ref{eq:group-property-statement-1}) and (\ref{eq:group-property-statement-2})
are satisfied. Take $\eps>0$, $f\in\calH_{S}\left(n\right)$ and
$f'\in\HSOL_{E_{1}}\left(n\right)$ such that 
\[
d_{\TV}\left(N_{f,P_{1}\left(\eps\right)},N_{f',P_{1}\left(\eps\right)}\right)<\delta_{1}\left(\eps\right)\,\,\text{.}
\]
It suffices to show that $f\in\HSOL_{E_{1}}^{<\eps}\left(n\right)$.
Let $h$ be an element of $\HSOL_{E_{2}}\left(n\right)$ that minimizes
$d^{\ham}\left(\lambda_{2}^{*}f,h\right)$. Then $\lambda_{1}^{*}h\in\HSOL_{E_{1}}\left(n\right)$
by Corollary \ref{cor:InverseOnSolutions}, and thus it suffices to
show that $d^{\ham}\left(f,\lambda_{1}^{*}h\right)<\eps$.

Now, $\lambda_{2}^{*}f'\in\HSOL_{E_{2}}\left(n\right)$ by Corollary
\ref{cor:InverseOnSolutions}. Furthermore, 
\begin{align*}
d_{\TV}\left(N_{\lambda_{2}^{*}f,P_{2}\left(\eps\right)},N_{\lambda_{2}^{*}f',P_{2}\left(\eps\right)}\right) & \le d_{\TV}\left(N_{f,\lambda_{2}\left(P_{2}\left(\eps\right)\right)},N_{f',\lambda_{2}\left(P_{2}\left(\eps\right)\right)}\right) & \text{by Lemma \ref{lem:lambda2^*Preserevesd_P}}\\
 & \le d_{\TV}\left(N_{f,P_{1}\left(\eps\right)},N_{f',P_{1}\left(\eps\right)}\right) & \text{by Corollary \ref{cor:d_TVIncreasingInP}}\\
 & <\delta_{1}\left(\eps\right)\le\delta_{2}\left(\frac{\eps}{2C_{1}}\right)\,\,\text{,}
\end{align*}
and thus $\lambda_{2}^{*}f\in\HSOL_{E_{2}}^{<\frac{\eps}{2C_{1}}}\left(n\right)$
since $E_{2}$ is $\left(P_{2}\left(\eps\right),\delta_{2}\left(\eps\right)\right)$-statistically-distinguishable.
In other words, 
\begin{equation}
d^{\ham}\left(\lambda_{2}^{*}f,h\right)<\frac{\eps}{2C_{1}}\,\,\text{.}\label{eq:hh'Distance}
\end{equation}

Now,
\begin{align}
\frac{\left|\Bad_{E_{1}}\left(f\right)\right|}{n} & =d_{\TV}\left(N_{f,R_{E_{1}}},N_{f',R_{E_{1}}}\right) & \text{by Remark \ref{rem:LocalDefectAsNeighborhoodTVDistance}}\nonumber \\
 & \le d_{\TV}\left(N_{f,P_{1}\left(\eps\right)},N_{f',P_{1}(\eps)}\right) & \text{by Corollary \ref{cor:d_TVIncreasingInP}}\nonumber \\
 & <\delta_{1}\left(\eps\right)\le\frac{\eps}{2C_{2}}\,\,\text{.}\label{eq:fLocalDefectBound}
\end{align}
Hence,
\begin{align*}
d^{\ham}\left(f,\lambda_{1}^{*}h\right) & \le C_{1}d^{\ham}\left(\lambda_{2}^{*}f,h\right)+C_{2}\frac{\left|\Bad_{E_{1}}\left(f\right)\right|}{n} & \text{by Lemma \ref{lem:DistancePreservingBothSides}}\\
 & <\eps\,\,\text{.} & \text{by \eqref{eq:hh'Distance} and \eqref{eq:fLocalDefectBound}}
\end{align*}

\section{Analysis and universality of the LSM Algorithm\label{sec:LSM}}

Let $S=\left\{ s_{1},\dotsc,s_{d}\right\} $ be an alphabet, and fix
a system of relations $E$ over $S^{\pm}$. The goal of this section
is to prove Theorem \ref{thm:LSMUniversal}, which follows immediately
from the following two propositions.
\begin{prop}
	[Statistical distinguishability implies testability]\label{prop:LSMPositive}If
	$E$ is $\left(P,\delta\right)$-statistically-distinguishable then
	$\eps\mapsto\LSM_{k\left(\eps\right),P(\eps),\frac{\delta\left(\eps\right)}{2}}^{E}$
	is a family of testers for $E$, where $k\left(\eps\right)=\left\lceil \frac{100\cdot2^{|P\left(\eps\right)|}}{\delta\left(\eps\right)^{2}}\right\rceil $.
\end{prop}

In particular, by the query-complexity analysis in Section \ref{subsec:intro-LSM}, Proposition
\ref{prop:LSMPositive} implies that $E$ is $O\left(\frac{\sum_{w\in P\left(\eps\right)}|w|\cdot2^{|P|}}{\delta\left(\eps\right)^{2}}\right)$-testable. We prove Proposition \ref{prop:LSMPositive} by reduction to hypothesis testing.
\begin{prop}
	[Testability implies statistical distinguishability]\label{prop:LSMNegative}Suppose
	that $E$ is $q\left(\eps\right)$-testable, $q\colon\RR_{>0}\to\NN$.
	Then, $E$ is $\left(B_{2q\left(\eps\right)}\cup E,q\left(\eps\right){}^{-c\cdot q\left(\eps\right)}\right)$-statistically-distinguishable,
	where $c>0$ is a universal constant.
\end{prop}

Here $B_{k}$ denotes the ball of radius $k$ in $F_{S}$, namely,
\[
B_{k}=\left\{ w\in F_{S}\mid\left|w\right|\le k\right\} \,\,\text{.}
\]
To prove Proposition \ref{prop:LSMNegative}, we first use a variation of \emph{Yao's minimax inequality} to reduce testability to a statement about deterministic (rather than randomized) testers (Lemma \ref{lem:minimax}). We then formalize the intuitive idea that runs of a determinstic tester of small query complexity only ``care about'' a small part of the input $S$-graph (Lemma \ref{lem:SufficientConditionForMToHaveSameResult}). This allows us to relate these runs to statistical distinguishability, which is a local notion.

\subsection{Proof of Proposition \ref{prop:LSMPositive}}

Assume that $E$ is $\left(P,\delta\right)$-statistically-distinguishable,
$P\colon\RR_{>0}\to\FinSubsets\left(F_{S}\right)$, $\delta\colon\RR_{>0}\to\RR_{>0}$.
Fix $\eps_{0}>0$ and denote $P_{0}=P\left(\eps_{0}\right)$ and $\delta_{0}=\delta\left(\eps_{0}\right)$.
Let $k=\left\lceil \frac{100\cdot2^{|P_{0}|}}{\delta_{0}{}^{2}}\right\rceil $.
Fix $G\in\calG_{S}\left(n\right)$ such that either $G\in\GSOL_{E}$$\left(n\right)$
or $G\in\GSOL_{E}^{\geq\eps_{0}}\left(n\right)$. We wish to show
that $\LSM_{k,P_{0},\frac{\delta_{0}}{2}}^{E}$ returns a correct
result in its run on $G$ with probability at least $0.99$.

If $G\in\GSOL_{E}\left(n\right)$, set $A_{0}=\GSOL_{E}\left(n\right)$
and $A_{1}=\GSOL_{E}^{\geq\eps_{0}}\left(n\right)$, and otherwise
set $A_{0}=\GSOL_{E}^{\geq\eps_{0}}\left(n\right)$ and $A_{1}=\GSOL_{E}\left(n\right)$.
Then
\begin{equation}
d_{\TV}\left(N_{G,P_{0}},N_{G',P_{0}}\right)\geq\delta_{0}\quad\forall G'\in A_{1}\label{eq:LSM-positive-distinguishability}
\end{equation}
because $E$ is $\left(P\left(\eps\right),\delta\left(\eps\right)\right)$-statistically-distinguishable.

Let $x_{1},\dotsc x_{k}\in[n]$ denote the random variables sampled
in Step \ref{step:LSM-sampling} of a run of $\LSM_{k,P_{0},\frac{\delta_{0}}{2}}^{E}$ (see Algorithm \ref{alg:LSM})
on $G$. Denote $\overline{x}=\left(x_{1},\dotsc,x_{k}\right)\in\left[n\right]^{k}$,
and let $N_{\overline{x}}^{\Emp}$ be the distribution, over $\Subsets\left(P_{0}\right)$,
of $\stab_{G}\left(x_{i}\right)\cap P_{0}$ where $i\sim \U\left(\left[k\right]\right)$.
Recall that this run of $\LSM_{k,P_{0},\frac{\delta_{0}}{2}}^{E}$\emph{
}accepts $G$ if and only if Condition (\ref{eq:LSMCondition}) is
satisfied,\textbf{ }namely, if 
\begin{equation}
\min\left\{ d_{\TV}\left(N_{\overline{x}}^{\Emp},N_{H,P_{0}}\right)\mid H\in\GSOL_{E}\left(n\right)\right\} \le\frac{\delta_{0}}{2}\,\,\text{.}\label{eq:LSM-positive-AcceptCondition}
\end{equation}
In particular, the run returns a correct result on $G$ whenever 
\begin{equation}
d_{\TV}\left(N_{\overline{x}}^{\Emp},N_{G,P_{0}}\right)<\delta_{0}/2\,\,\text{.}\label{eq:LSM-positive-closeEnough}
\end{equation}
Indeed, if $G\in\GSOL_{E}\left(n\right)$, then (\ref{eq:LSM-positive-AcceptCondition})
holds by (\ref{eq:LSM-positive-closeEnough}), and so the run is accepting.
If $G\in\GSOL_{E}^{\ge\eps_{0}}\left(n\right)$ then (\ref{eq:LSM-positive-AcceptCondition})
does not hold by (\ref{eq:LSM-positive-distinguishability}), (\ref{eq:LSM-positive-closeEnough})
and the triangle inequality, and so the run is rejecting. Thus, it
suffices to show that

\[
\Pr_{\overline{x}\sim \U\left(\left[n\right]^{k}\right)}\left(d_{\TV}\left(N_{\overline{x}}^{\Emp},N_{G,P_{0}}\right)<\delta_{0}/2\right)\ge0.99\,\,\text{.}
\]
For a fixed $Y\in\Subsets\left(P\right)$, the random variable $k\cdot N_{\overline{x}}^{\Emp}\left(Y\right)\colon\left[n\right]^{k}\to\ZZ_{\geq0}$ (that is, ``$k$ times $N_{\overline{x}}^{\Emp}\left(Y\right)$'')
is distributed $\Binomial\left(k,N_{G,P_{0}}\left(Y\right)\right)$
when $x\sim \U\left(\left[n\right]^{k}\right)$. In particular, $\EE_{\overline{x}\sim \U\left(\left[n\right]^{k}\right)}\left[k\cdot N_{\overline{x}}^{\Emp}\left(Y\right)\right]=kN_{G,P_{0}}\left(Y\right)$.
Hence,
\begin{align*}
\EE_{\overline{x}\sim \U\left(\left[n\right]^{k}\right)}\left(\|N_{\overline{x}}^{\Emp}-N_{G,P_{\eps}}\|_{2}^{2}\right) & =\sum_{Y\in\Subsets\left(P_{0}\right)}\underbrace{\EE_{\overline{x}\sim \U\left(\left[n\right]^{k}\right)}\left(\left(N_{\overline{x}}^{\Emp}\left(Y\right)-N_{G,P_{0}}\left(Y\right)\right)^{2}\right)}_{=\frac{1}{k^{2}}\Var_{\overline{x}\sim \U\left(\left[n\right]^{k}\right)}\left(kN_{\overline{x}}^{\Emp}\left(Y\right)\right)}\\
& =\sum_{Y\in\Subsets\left(P_{0}\right)}\frac{N_{G,P_{0}}\left(Y\right)\left(1-N_{G,P_{0}}\left(Y\right)\right)}{k}\\
& \leq\frac{1}{k}\sum_{Y\in\Subsets\left(P_{0}\right)}N_{G,P_{0}}\left(Y\right)\\
& =\frac{1}{k}\,\,\text{.}
\end{align*}
By the Cauchy--Schwartz inequality,
\begin{align*}
d_{\TV}\left(N_{\overline{x}}^{\Emp},N_{G,P_{0}}\right) & =\frac{1}{2}\|N_{\overline{x}}^{\Emp}-N_{G,P_{0}}\|_{1}\le\frac{1}{2}\left|\Subsets\left(P_{0}\right)\right|^{\frac{1}{2}}\cdot\|N_{\overline{x}}^{\Emp}-N_{G,P_{0}}\|_{2}\\
& =\frac{1}{2}\cdot2^{\frac{\left|P_{0}\right|}{2}}\cdot\|N_{\overline{x}}^{\Emp}-N_{G,P_{0}}\|_{2}\,\,\text{.}
\end{align*}
Thus,
\begin{align*}
\Pr_{\overline{x}\sim \U\left(\left[n\right]^{k}\right)}\left(d_{\TV}\left(N_{\overline{x}}^{\Emp},N_{G,P_{0}}\right)\ge\delta_{0}/2\right) & \le\Pr_{\overline{x}\sim \U\left(\left[n\right]^{k}\right)}\left(\|N_{\overline{x}}^{\Emp}-N_{G,P_{0}}\|_{2}^{2}\ge2^{-\left|P_{0}\right|}\delta_{0}^{2}\right)\\
& \le\frac{2^{|P_{0}|}}{k\delta_{0}^{2}}\qquad\text{(by Markov's inequality})\\
& \le0.01\,\,\text{.}
\end{align*}

\subsection{Testability implies statistical distinguishability}

We prove Proposition \ref{prop:LSMNegative} in Section \ref{subsec:ProofLSMNegative}
after making the necessary preparations.

\subsubsection{A minimax principle}

The following lemma, a variant of \emph{Yao's minimax inequality}
\cite{Yao77}, enables us to reduce Proposition \ref{prop:LSMNegative}
to a statement about deterministic, rather than randomized, algorithms.
\begin{lem}
	\label{lem:minimax} Fix $\eps>0$ and $q\in\NN$. The following conditions
	are equivalent: 
	\begin{enumerate}
		\item There exists an $\left(\eps,q\right)$-tester for $E$. 
		\item For every distribution $D$ on $\bigcup_{n\in\NN}\left(\GSOL_{E}(n)\cup\GSOL_{E}^{\ge\eps}\left(n\right)\right)$,
		there exists a deterministic algorithm $\cM$, limited to making
		at most $q$ queries of $G$, such that
		\[
		\Pr_{G\sim D}\left[\text{\ensuremath{\calM} solves the input \ensuremath{G} correctly}\right]\ge0.99\,\,\text{.}
		\]
	\end{enumerate}
\end{lem}

\begin{proof}
	Let $\calA$ denote the set of all deterministic algorithms,
	with input in $\bigcup_{n\in\NN}\calG_{S}\left(n\right)$ and boolean
	output, limited to making at most $q$ queries in a run. Let $\calB=\bigcup_{n\in\NN}\left(\GSOL_{E}\left(n\right)\cup\GSOL_{E}^{\ge\eps}\left(n\right)\right)$.
	Let $\calD_{\calA}$ (resp. $\calD_{\calB}$) denote the set of all
	probability distributions over $\calA$ (resp. $\calB$). For $\cM\in\calA$,
	$G\in\calB$, let 
	\[
	c_{\cM,G}=\begin{cases}
	1 & \text{if \ensuremath{\calM} solves the input \ensuremath{G} correctly}\\
	0 & \text{otherwise}
	\end{cases}\,\,\text{.}
	\]
	A randomized algorithm can be viewed as a distribution over $\calA$.
	Hence, the first condition in the statement of the proposition is
	equivalent to 
	\begin{equation}
	\max_{D\in\calD_{\calA}}\min_{G\in\calB}\Pr_{M\sim D}\left[c_{\cM,G}=1\right]\ge0.99\,\,\text{,}\label{eq:minimaxPrimal}
	\end{equation}
	while the second condition can be written as
	\begin{equation}
	\min_{D\in\calD_{B}}\max_{\cM\in\calA}\Pr_{G\sim D}\left[c_{\cM,G}=1\right]\ge0.99\,\,\text{.}\label{eq:minimaxDual}
	\end{equation}
	Observe that the left-hand side of (\ref{eq:minimaxPrimal}) is equal
	to
	\begin{equation}
	\max_{D\in\calD_{\calA}}\min_{D'\in\calD_{\calB}}\Pr_{\substack{G\sim D\\
			\cM\sim D'
		}
	}\left[c_{\cM,G}=1\right]\,\,\text{.}\label{eq:minimaxPrimalWholeSimplex}
	\end{equation}
	Indeed, $\Pr_{\substack{G\sim D\\
			\cM\sim D'
		}
	}\left[c_{\cM,G}=1\right]$ is linear in $D'$, and hence the minimum in (\ref{eq:minimaxPrimal})
	is obtained at a vertex of the simplex $\calD_{B}$. Similarly, the
	left-hand side of (\ref{eq:minimaxDual}) is equal to 
	\begin{equation}
	\min_{D'\in\calD_{B}}\max_{D\in\calD_{\calA}}\Pr_{\substack{G\sim D\\
			\cM\sim D'
		}
	}\left[c_{\cM,G}=1\right]\,\,\text{.}\label{eq:minimaxDualWholeSimplex}
	\end{equation}
	Finally, (\ref{eq:minimaxPrimalWholeSimplex}) and (\ref{eq:minimaxDualWholeSimplex})
	are equal by the von Neumann Minimax Theorem \cite[Thm. 1]{Du1995}. 
\end{proof}

\subsubsection{The $d_{\protect\TV}$ metric}

We recall the following standard facts about the $d_{\TV}$ metric.
\begin{fact}[{\cite[Prop. 4.2]{LP17Markov}}]
	\label{fact:dTVMaxEvent}Let $\theta$ and $\theta'$ be two distributions
	over the same finite base set $\Omega$. Then
	\[
	d_{\TV}\left(\theta,\theta'\right)=\max_{A\subset\Omega}\left|\theta\left(A\right)-\theta'\left(A\right)\right|\,\,\text{.}
	\]
\end{fact}

\begin{fact}[{Coupling Lemma \cite[Prop. 4.7]{LP17Markov}}]
	\label{fact:dTVCoupling}For any two distributions $\theta$ and
	$\theta'$ over the same finite base set $\Omega$, there exists a
	distribution $D$ over $\Omega\times\Omega$ such that:
	\begin{enumerate}
		\item The marginal distributions of $D$ are equal, respectively, to $\theta$
		and $\theta'$. 
		\item 
		\[
		d_{\TV}\left(\theta,\theta'\right)=\Pr_{\left(x,y\right)\sim D}\left[x\ne y\right]\,\,\text{.}
		\]
	\end{enumerate}
\end{fact}

\begin{fact}[{\cite[Ex. 4.4]{LP17Markov}}]
	\label{fact:dTVProduct}Let $\theta=\prod_{i=1}^{t}\theta_{i}$ and
	$\theta'=\prod_{i=1}^{t}\theta_{i}'$ be product measures. Then,
	\[
	d_{\TV}\left(\theta,\theta'\right)\le\sum_{i=1}^{t}d_{\TV}\left(\theta_{i},\theta_{i}'\right)
	\]
\end{fact}

\subsubsection{Partial $S$-graphs and runs of algorithms}
\begin{defn}
	[Partial $S$-graph]Let $S$ be a finite set. A \emph{partial $S$-graph}
	is a directed graph $H$ whose edges are labelled by $S$, such that
	for every $s\in S$, every vertex of $H$ has at most one outgoing
	edge labelled $s$ and at most one incoming edge labelled $s$.
	
	We denote the vertex set of $H$ by $V\left(H\right)$, and $E\left(H\right)$
	will be the set of ($S$-labelled) edges in $H$. For vertices $x_{1},x_{2}\in V(H)$ and $s\in S$, we say that $x_{1}\overset{s^{-1}}{\rightedge}x_{2}$
	is an edge of $H$ if the same is true for $x_{2}\overset{s}{\rightedge}x_{1}$. When computing the
	cardinality $\left|E\left(H\right)\right|$ of $E\left(H\right)$,
	the edges $x_{2}\overset{s}{\rightedge}x_{1}$ and $x_{1}\overset{s^{-1}}{\rightedge}x_{2}$
	count as one edge.
\end{defn}

We denote the set of all partial $S$-graphs $H$ such that $V\left(H\right)\subset\left[n\right]$
by $\PG_{S}\left(n\right)$. Given $H\in\PG_{S}\left(n\right)$, we
write $r\left(H\right)=\left|V\left(H\right)\right|-k$, where $k$
is the number of connected components of $H$. For example, $r\left(H\right)=\left|V\left(H\right)\right|-1$
if $H$ is connected, and $r\left(H\right)=0$ if $H$ has no edges.

\begin{defn}
	[Paths in a partial $S$-graph] Let $H$ be a partial $S$-graph,
	$x\in V(H)$ and $w\in F_{S}$. Suppose that the reduced form of $w$
	is given by $w_{k}\cdots w_{1}$ ($w_{i}\in S^{\pm}$ for each $1\le i\le k$).
	We write $w_{H}x=y$, and say that $w$ is an \emph{$H$-path} for $x$, if there exist $x_{0},\ldots,x_{k}\in V(H)$
	such that $x_{0}=x$, $x_{k}=y$, and $H$ contains a $w_{i}$-labelled
	edge $x_{i-1}\rightedge x_{i}$ for each $1\le i\le k$ (in this case
	$x_{0},\dotsc,x_{k}$ are uniquely defined). If, in addition, the
	vertices $x_{0},\dotsc,x_{k}$ are distinct, except for the possibility
	that $x_{0}=x_{k}$, we say that $w$ is a\emph{ simple $H$-path}
	for $x$.
	
	Let 
	\[
	\paths_{H}\left(x\right)=\left\{ w\in F_{S}\mid w\text{ is an }H\text{-path for }x\right\} \,\,\text{,}
	\]
	\[
	\spaths_{H}\left(x\right)=\left\{ w\in F_{S}\mid w\text{ is a simple }H\text{-path for }x\right\} \,\,\text{,}
	\]
	\[
	\bp_{H}\left(x\right)=\left\{ w'^{-1}\cdot w\mid w,w'\in\spaths_{H}\left(x\right)\right\} 
	\]
	and
	\[
	\pstab_{H}\left(x\right)=\left\{ w\in\paths_H(x)\cap\bp_{H}\left(x\right)\mid w_{H}x=x\right\} \,\,\text{.}
	\]
\end{defn}

The definition of $\bp_{H}\left(x\right)$ ensures that the set $\pstab_H\left(x\right)$
contains all the information about simple paths emerging from $x$
and ending at the same vertex. Indeed, for $w,w'\in\spaths_{H}\left(x\right)$,
we have $w_{H}x=w'_{H}x$ if and only if $w'^{-1}w\in\pstab_H\left(x\right)$.
\begin{defn}
	[Inclusion of a partial $S$-graph]Let $G\in\calG_{S}\left(n\right)$
	and $H\in\PG_{S}\left(n\right)$. Write $H\subset G$ if $G$ contains
	all the labelled directed edges of $H$. Equivalently, $H\subset G$
	if $w_{G}x=w_{H}x$ for all $x\in V\left(H\right)$ and $w\in\spaths_{H}\left(x\right)$.
\end{defn}

\begin{rem}
	Crucially, in contrast to other common definitions of inclusion of
	graphs, here we care about vertex labels. For example, if $H\subset G$
	and $H$ has the edge $3\overset{s}{\rightedge}8$, then $G$ must
	also have the edge $3\overset{s}{\rightedge}8$. 
\end{rem}

Let $\cM$ be a deterministic algorithm that makes exactly $q$
queries on its input $G\in\calG_{S}\left(n\right)$. Let $H_{\cM,G}\in\PG_{S}\left(n\right)$
be the partial $S$-graph whose edge set is the set of edges of $G$
queried by $\cM$ in its run on $G$, with vertex set consisting of
the vertices incident to these edges. More formally, in graph-theoretic
terms each query of $\calM$ is of the form: ``what is the label
of the vertex $s_{G}x$?'' for given $s\in S^{\pm}$ and $x\in\left[n\right]$.
We denote this query by $\left(x,s\right)$. Denote the sequence of
queries that $\cM$ makes in its run on $G$ by $\left(x^{\cM,G,1},s^{\cM,G,1}\right),\ldots,\left(x^{\cM,G,q},s^{\cM,G,q}\right)$.
Then the vertex set of $H_{\cM,G}$ is $\left\{ x^{\cM,G,1},\ldots,x^{\cM,G,q}\right\} \cup\left\{ s_{G}^{\cM,G,1}x^{\cM,G,1},\ldots,s_{G}^{\cM,G,q}x^{\cM,G,q}\right\} $,
and its edges are $x^{\cM,G,i}\stackrel{s^{\cM,G,i}}{\rightedge}s_{G}^{\cM,G,i}x^{\cM,G,i}$,
$1\le i\le q$. 

Clearly $H_{\cM,G}\subset G$. The key observation is that if an $S$-graph
$G'\in\calG_{S}\left(n\right)$ contains $H_{\cM,G}$ then $\cM$
has an identical run on $G$ and on $G'$. Indeed, since $\cM$ is
deterministic, it always starts with the same query. Since both $G$
and $G'$ contain the edge $x^{\cM,G,1}\stackrel{s^{\cM,G,1}}{\rightedge}s_{G}^{\cM,G,1}x^{\cM,G,1}$,
the answer to this query is the same in both runs. Consequently, the
second query of $\cM$ is also the same, and so on.

We denote 
\[
\calR_{\cM}\left(n\right)=\left\{ H_{\cM,G}\mid G\in\calG_{S}\left(n\right)\right\} 
\]
and 
\begin{equation}
\calR_{\cM,r}\left(n\right)=\left\{ H\in\calR_{\cM}\left(n\right)\mid r\left(H\right)=r\right\} \label{eq:DefOfR_M,r(n)}
\end{equation}
for $r\geq0$. For $G\in\calG_{S}\left(n\right)$, the partial $S$-graph
$H=H_{\cM,G}$ is the unique $H\in\calR_{\cM}\left(n\right)$ such
that $H\subset G$. If $H_{\cM,G}=H_{\cM,G'}$ then $\cM$ has identical
runs on $G$ and $G'$. In particular, these runs terminate with the
same result. We proved:
\begin{lem}
	\label{lem:SufficientConditionForMToHaveSameResult}Let $\cM$ be
	a deterministic algorithm and let $G,G'\in\calG_{S}\left(n\right)$.
	Suppose that $H_{\cM,G}=H_{\cM,G'}$. Then
	\[
	\cM\text{ accepts }G\iff\cM\text{ accepts }G'\,\,\text{.}
	\]
\end{lem}

\subsubsection{Vertex relabelling of an $S$-graph}
\begin{defn}
	[$S$-graph relabelling] \label{def:GraphRerelabelling}Let $G\in\calG_{S}\left(n\right)$
	and let $\pi$ be a permutation over $\left[n\right]$. The \emph{vertex-relabelled
		graph} $G\pi\in\calG_{S}\left(n\right)$ is the graph with vertex
	set $\left[n\right]$ and edge set $\left\{ \pi^{-1}x\overset{s}{\rightedge}\pi^{-1}\left(s_{G}x\right)\mid x\in\left[n\right]\right\} $. 
\end{defn}

\begin{rem}
	\label{rem:Rerelabelling}For each $w\in F_{S}$, the following diagram
	commutes:
	\[
	\xymatrix{[n]\ar[r]^{w_{G\pi}}\ar[d]^{\pi} & [n]\ar[d]^{\pi}\\{}
		[n]\ar[r]^{w_{G}} & [n]
	}
	\,\,\text{.}
	\]
	In other words, $w_{\left(G\pi\right)}x=\pi^{-1}\left(w_{G}\pi x\right)$
	for $x\in V\left(G\pi\right)$.
	
	We think of $\pi$ as a relabelling function, in the sense that $G\pi$
	is obtained from $G$ by changing the name of each vertex $\pi\left(x\right)\in V\left(G\right)$
	to $x$.
\end{rem}

\begin{rem}
	\label{rem:rerelabellingPreserves} Let $G\in\calG_{S}\left(n\right)$
	and $\pi\in\Sym\left(n\right)$. It is straightforward to verify that
	$G\in\GSOL_{E}\left(n\right)$ implies $G\pi\in\GSOL_{E}\left(n\right)$,
	and that $G\in\GSOL_{E}^{\ge\eps}\left(n\right)$ implies $G\pi\in\GSOL_{E}^{\ge\eps}\left(n\right)$
	($\eps>0$).
\end{rem}

\begin{rem}
	For each $\pi\in\Sym\left(n\right)$, the map $G\mapsto G\pi\colon\calG_{S}\left(n\right)\to\calG_{S}\left(n\right)$
	defines a right action of $\Sym\left(n\right)$ on $\calG_{S}\left(n\right)$.
	For $G\in\calG_{S}\left(n\right)$, the set $\left\{ G\pi\mid\pi\in\Sym\left(n\right)\right\} $
	is an orbit of this action. The distribution of $G\pi$ when $\pi\sim\U\left(\Sym\left(n\right)\right)$
	is the uniform distribution on this orbit.
\end{rem}

\subsubsection{Proving Proposition \ref{prop:LSMNegative} \label{subsec:ProofLSMNegative}}

We shall prove the following technical lemma.
\begin{lem}
	\label{lem:LSMNegativeMainTechnical}Fix two $S$-graphs $G_{0},G_{1}\in\calG_{S}\left(n\right)$.
	Fix $q\in\NN$ with $1\le q\le n^{\frac{1}{4}}$ and let $\cM$ be
	a deterministic algorithm, limited to making at most $q$ queries.
	Let $\theta_{i}$, $i\in\left\{ 0,1\right\} ,$ be the distribution (over $\calR_{\calM}(n)$)
	of the partial $S$-graph $H_{\cM,G_{i}\pi}$, where $\pi\sim \U\left(\Sym(n)\right)$.
	Then
	\[
	d_{\TV}\left(\theta_{1},\theta_{2}\right)\le\left(2q\right)^{q+2}\left(d_{\TV}\left(N_{G_{0},B_{2q}},N_{G_{1},B_{2q}}\right)+cn^{-\frac{1}{4}}\right)\,\,\text{,}
	\]
	where $c>0$ is a universal constant.
\end{lem}

Before proving Lemma \ref{lem:LSMNegativeMainTechnical}, we show
that it implies Proposition \ref{prop:LSMNegative}.
\begin{proof}[Proof of Proposition \ref{prop:LSMNegative}]
	By the hypothesis, $E$ is $q\left(\eps\right)$-testable. Fix $\eps>0$
	and write $q_{\eps}=q\left(\eps\right)$. Let $P=B_{2q_{\eps}}\cup E$.
	Fix $G_{0}\in\GSOL_{E}\left(n\right)$ and $G_{1}\in\GSOL_{E}^{\ge\eps}\left(n\right)$,
	and denote $\delta=d_{\TV}\left(N_{G_{0},P},N_{G_{1},P}\right)$.
	Our goal is to show that
	\[
	\delta\ge q_{\eps}^{-O\left(q_{\eps}\right)}\,\,\text{,}
	\]
	where the implied constant is universal. We first deal with the case
	\[
	n<\max\left\{ \frac{\left(2q_{\eps}\right)^{4q_{\eps}+8}\cdot c^{4}}{0.057},q_{\eps}^{4}\right\} \,\,\text{.}
	\]
	where $c$ is the constant from Lemma \ref{lem:LSMNegativeMainTechnical}.
	We note that, since $G_{0}\in\GSOL_{E}\left(n\right)$, we have $\stab_{G_{0},E}\left(x\right)=E$
	for every $x\in\left[n\right]$. However, the same is not true for
	$\stab_{G_{1},E}(x)$ since $G_{1}\notin\GSOL_{E}\left(n\right)$. Since $E\subseteq P$, Corollary \ref{cor:d_TVIncreasingInP} yields (for the leftmost inequality)
	\[
	\delta\ge d_{\TV}\left(N_{G_{0},E},N_{G_{1},E}\right)\ge \Pr_{x\sim \U([n])}\left[\stab_{G_0,E}(x)=E\right]-\Pr_{x\sim \U([n])}\left[\stab_{G_1,E}(x)=E\right]\ge\frac{1}{n}\ge q_{\eps}^{-O\left(q_{\eps}\right)}\,\,\text{.}
	\]
	
	We proceed, assuming that 
	\begin{equation}
	n\ge\max\left\{ \frac{\left(2q_{\eps}\right)^{4q_{\eps}+8}\cdot c^{4}}{0.057},q_{\eps}^{4}\right\} \,\,\text{.}\label{eq:LSMNegativenIsLarge}
	\end{equation}
	Let $D$ be the distribution of the random $S$-graph $G_{i}\pi$,
	where $i$ is the result of a fair coin flip, and $\pi$ is independently
	sampled from $\U\left(\Sym\left(n\right)\right)$ uniformly. Applying
	Lemma \ref{lem:minimax} to $D$ yields a deterministic algorithm
	$\cM$, limited to making $q_{\eps}$ queries, such that
	\begin{align}
	\Pr_{G\sim D}\left[\text{\ensuremath{\calM} solves the input \ensuremath{G} correctly}\right]\ge0.99\,\,\text{.}\label{eq:MIsAGoodAlgorithm}
	\end{align}

	Let $\theta_{i}$, $i\in\left\{ 0,1\right\} ,$ be the distribution
	of the partial $S$-graph $H_{\cM,G_{i}\pi}$, where $\pi\sim \U\left(\Sym(n)\right)$. We claim that
	\begin{equation}\label{eq:LSMNegativedThetasLowerBound}
	d_{\TV}(\theta_0,\theta_1) \ge 0.49\,\,\text{.}
	\end{equation}
	
	Before proving \eqref{eq:LSMNegativedThetasLowerBound}, we show that it implies the proposition. Note that every $H\in\calR_{\cM}\left(n\right)$ has $\left|E\left(H\right)\right|\le q_{\eps}\le n^{\frac{1}{4}}$,
	where the rightmost inequality is due to (\ref{eq:LSMNegativenIsLarge}).
	Hence,
	\begin{align*}
	d_{\TV}\left(\theta_{0},\theta_{1}\right)\le & \left(2q_{\eps}\right)^{q_{\eps}+2}\left(d_{\TV}\left(N_{G_{0},B_{2q}},N_{G_{1},B_{2q}}\right)+cn^{-\frac{1}{4}}\right) & \text{by Lemma \ref{lem:LSMNegativeMainTechnical}}\\
	\le & \left(2q_{\eps}\right)^{q_{\eps}+2}\left(\delta+cn^{-\frac{1}{4}}\right) & \text{by Corollary \ref{cor:d_TVIncreasingInP}}\\
	\le & \left(2q_{\eps}\right)^{q_{\eps}+2}\delta+0.057^{1/4}\,\,\text{.} & \text{by \eqref{eq:LSMNegativenIsLarge}}
	\end{align*}
	Together with \eqref{eq:LSMNegativedThetasLowerBound}, this yields
	\[
	\delta\ge\left(2q_{\eps}\right)^{-q_{\eps}-2}\left(0.49-\underbrace{0.057^{1/4}}_{< 0.49}\right)\ge q_{\eps}^{-O\left(q_{\eps}\right)}\,\,\text{,}
	\]
	and the proposition follows.
	
	We turn to proving \eqref{eq:LSMNegativedThetasLowerBound}. Fact \ref{fact:dTVCoupling} yields a distribution $Q$ over $\calR_{\cM}\left(n\right)\times\calR_{\cM}\left(n\right)$
	with respective marginal distributions $\theta_{0}$ and $\theta_{1}$,
	such that 
	\[
	\Pr_{\left(H_{0},H_{1}\right)\sim Q}\left[H_{0}\ne H_{1}\right]=d_{\TV}\left(\theta_{0},\theta_{1}\right)\,\,\text{.}
	\]
	Let $\overline{Q}$ be a distribution over $\left(\Sym(n)\right)^{2}$
	which generates a pair of permutations $\left(\pi_{0},\pi_{1}\right)$
	by first sampling $\left(H_{0},H_{1}\right)\sim Q$, and then independently
	sampling $\pi_{0}\sim \U\left(\left\{ \pi\in\Sym(n)\mid H_{\cM,G_{0}\pi}=H_{0}\right\} \right)$
	and $\pi_{1}\sim \U\left(\left\{ \pi\sim\Sym(n)\mid H_{\cM,G_{1}\pi}=H_{1}\right\} \right)$.
	Observe that both marginal distributions of $\overline{Q}$ are equal
	to the uniform distribution on $\Sym(n)$. Indeed, for $i\in\{0,1\}$
	and $\alpha\in\Sym(n)$,
	\begin{align*}
	\Pr_{\left(\pi_{0},\pi_{1}\right)\sim\overline{Q}}\left[\pi_{i}=\alpha\right] & =\Pr_{\left(H_{0},H_{1}\right)\sim Q}\left[H_{i}=H_{\cM,G_{i}\alpha}\right]\cdot \frac{1}{\left|\left\{ \pi\in\Sym(n)\mid H_{\cM,G_{i}\pi}=H_{\calM,G_i\alpha}\right\} \right|}\\
	& =\Pr_{H_i\sim \theta_i}\left[H_i=H_{\cM,G_i\alpha}\right]\cdot \frac 1{\left|\left\{ \pi\in\Sym(n)\mid H_{\cM,G_{i}\pi}=H_{\cM,G_i\alpha}\right\} \right|}\\&=
	\frac{\left|\left\{ \pi\in\Sym(n)\mid H_{\cM,G_{i}\pi}=H_{\cM,G_i\alpha}\right\} \right|}{\left |\Sym(n)\right|}\cdot \frac 1{\left|\left\{ \pi\in\Sym(n)\mid H_{\cM,G_{i}\pi}=H_{\cM,G_i\alpha}\right\} \right|} \\
	&=\frac{1}{\left|\Sym(n)\right|}\,\,\text{.}
	\end{align*}
	Furthermore,
	
	\[
	\Pr_{\left(\pi_{0},\pi_{1}\right)\sim\overline{Q}}\left[H_{\cM,G_{0}\pi_{0}}\ne H_{\cM,G_{1}\pi_{1}}\right]=\Pr_{\left(H_{0},H_{1}\right)\sim Q}\left[H_{0}\ne H_{1}\right]=d_{\TV}\left(\theta_{0},\theta_{1}\right)\,\,\text{.}
	\]
	Then 
	\begin{align*}
	0.99 & \le\Pr_{G\sim D}\left[\text{\ensuremath{\calM} correctly solves the input \ensuremath{G}}\right] & \text{by (\ref{eq:MIsAGoodAlgorithm})}\nonumber \\
	& =\frac{1}{2}\left(\Pr_{\pi_{0}\sim \U\left(\Sym(n)\right)}\left[\cM\text{ correctly solves }G_{0}\pi_{0}\right]+\Pr_{\pi_{1}\sim \U\left(\Sym(n)\right)}\left[\cM\text{ correctly solves }G_{1}\pi_{1}\right]\right)\nonumber \\
	& =\frac 12 \left(\Pr_{\left(\pi_{0},\pi_{1}\right)\sim\overline{Q}}
	\left[\cM\text{ correctly solves }G_{0}\pi_{0}\right]+\Pr_{\left(\pi_{0},\pi_{1}\right)\sim\overline{Q}}
	\left[\cM\text{ correctly solves }G_{1}\pi_{1}\right]\right)\nonumber \\
	& =\frac{1}{2}\Pr_{\left(\pi_{0},\pi_{1}\right)\sim\overline{Q}}\left[\cM\text{ correctly solves exactly one of }G_{0}\pi_{0}\text{ and }G_{1}\pi_{1}\right]\nonumber \\
	& \,\,\,\,\,+\Pr_{\left(\pi_{0},\pi_{1}\right)\sim\overline{Q}}\left[\cM\text{ correctly solves both }G_{0}\pi_{0}\text{ and }G_{1}\pi_{1}\right]\nonumber \\
	& \le\frac{1}{2}+\Pr_{\left(\pi_{0},\pi_{1}\right)\sim\overline{Q}}\left[\cM\text{ correctly solves both }G_{0}\pi_{0}\text{ and }G_{1}\pi_{1}\right]\nonumber \\
	& =\frac{1}{2}+\Pr_{\left(\pi_{0},\pi_{1}\right)\sim\overline{Q}}\left[\cM\text{ accepts }G_{0}\pi_{0}\text{ and rejects }G_{1}\pi_{1}\right] & \text{by Remark \ref{rem:rerelabellingPreserves}}\nonumber \\
	& \le\frac{1}{2}+\Pr_{\left(\pi_{0},\pi_{1}\right)\sim\overline{Q}}\left[\cM\text{ returns different results on }G_{0}\pi_{0}\text{ and }G_{1}\pi_{1}\right] & \text{}\nonumber \\
	& \le\frac{1}{2}+\Pr_{\left(\pi_{0},\pi_{1}\right)\sim\overline{Q}}\left[H_{\cM,G_{0}\pi_{0}}\ne H_{\cM,G_{1}\pi_{1}}\right] & \text{by Lemma \nonumber \ref{lem:SufficientConditionForMToHaveSameResult}}\\
	& =\frac{1}{2}+d_{\TV}\left(\theta_{0},\theta_{1}\right)\,\,\text{,}\nonumber
	\end{align*}
	and \eqref{eq:LSMNegativedThetasLowerBound} follows.
\end{proof}

The rest of this section is devoted to proving Lemma \ref{lem:LSMNegativeMainTechnical}.
The following lemma bounds the cardinality of the set $\calR_{\calM,r}\left(n\right)$
(defined in (\ref{eq:DefOfR_M,r(n)})). 
\begin{lem}
	\label{lem:BoundedNumberOfRuns}Let $\cM$ be a deterministic algorithm, limited to making at most $q$ queries per run. Let
	$n,r\in\NN$. Then 
	\[
	\left|\calR_{\cM,r}\left(n\right)\right|\le\binom{q}{r}n^{r}\left(2q\right)^{q-r}\,\,\text{.}
	\]
\end{lem}

\begin{proof}
	For $G\in\calG_{S}\left(n\right)$ and $0\leq i\leq q$, let $H_{i}^{G}$
	be the partial $S$-graph such that the edges of $H_{i}^{G}$ are
	the edges of $G$ queried by $\cM$ in the first $i$ queries, and
	the vertices of $H_{i}^{G}$ are those that touch these edges. In
	particular, $H_{0}^{G}$ is the empty partial $S$-graph, and $H_{q}^{G}=H_{\cM,G}$.
	Note that the query $\left(x^{\cM,G,i},s^{\cM,G,i}\right)$ is determined
	by $H_{i-1}^{G}$. The answer to this query, for which there are at
	most $n$ possibilities, depends on $G$.
	
	Clearly, $r\left(H_{0}^{G}\right)=0$ and $r\left(H_{q}^{G}\right)=r\left(H_{\cM,G}\right)$.
	Furthermore, in each step, $r\left(H_{i}\right)=r\left(H_{i-1}\right)+1$
	or $r\left(H_{i}\right)=r\left(H_{i-1}\right)$. Respectively, we
	say that the $i$-th step is \emph{increasing} or \emph{nonincreasing}.
	If the $i$-th step is nonincreasing then the vertex given as an answer
	to the $i$-th query is in $\left\{ x^{\cM,G,j}\mid1\le j\leq i\right\} \cup\left\{ s^{\cM,G,j}x^{\cM,G,j}\mid1\le j\leq i-1\right\} $.
	In particular, there are at most $2q$ possible answers to the query
	in a nonincreasing step. Write $\calR_{\cM,r',i}\left(n\right)=\left\{ H_{i}^{G}\mid\text{\ensuremath{G\in\calG_{S}\left(n\right)} and \ensuremath{r\left(H_{i}^{G}\right)=r'}}\right\} $.
	By the above,
	
	\[
	\left|\calR_{\cM,r',i}\left(n\right)\right|\le n\left|\calR_{\cM,r'-1,i-1}\left(n\right)\right|+2q\left|\calR_{\cM,r',i-1}\left(n\right)\right|\quad\forall1\leq i\leq q\,\,\,\forall 0\le r'\le r \,\,\text{,}
	\]
	and thus
	\[
	\left|\calR_{\cM,r}\left(n\right)\right|=\left|\calR_{\cM,r,q}\left(n\right)\right|\le\binom{q}{r}n^{r}\left(2q\right){}^{q-r}\,\,\text{.}
	\]
\end{proof}
The next lemma provides a condition equivalent to $H\subset G$ when
$H$ is a connected partial $S$-graph.
\begin{lem}
	[An equivalent condition for  inclusion of a  partial $S$-graph] \label{lem:SubgraphCharacterization}Fix
	$G\in\calG_{S}\left(n\right)$ and $H\in\PG_{S}\left(n\right)$, where
	$H$ is connected and nonempty. Treating $H$ as undirected, fix an undirected spanning tree $T$ for it. Fix some $y_{0}\in V\left(H\right)$. For every $y\in V\left(H\right)$, fix a word
	 $w^{y}\in\spaths_{H}\left(y_{0}\right)$ such that $w_{H}^{y}y_{0}=y$, and the simple $H$-path from $y_0$ to $y$, induced by $w^y$, proceeds along edges of $T$. Then, $H\subset G$ if and only if
	
	\begin{equation}
	\stab_{G}\left(y_{0}\right)\cap\bp_{H}\left(y_{0}\right)=\pstab_H\left(y_{0}\right)\label{eq:EventADef}
	\end{equation}
	and
	
	\begin{equation}
	w_{G}^{y}y_{0}=y\quad\forall y\in V\left(H\right)\,\,\text{.}\label{eq:EventB_uDef}
	\end{equation}
\end{lem}

\begin{proof}
	Suppose that $H\subset G$. Then $w_{G}y_{0}=w_{H}y_{0}$ for all
	$w\in\spaths_{H}\left(y_{0}\right)$, and (\ref{eq:EventB_uDef}) follows.
	The $\supset$ inclusion in (\ref{eq:EventADef}) is clear. To prove
	the $\subset$ inclusion, let $w\in\stab_{G}\left(y_{0}\right)\cap\bp_{H}\left(y_{0}\right)$
	and write $w=v'^{-1}v$ where $v,v'\in\spaths_{H}\left(y_{0}\right)$.
	Then $y_{0}=w_{G}y_{0}=v_{G}'^{-1}v_{G}y_{0}$, and thus $v'_{G}y_{0}=v_{G}y_{0}$.
	Hence $v'_{H}y_{0}=v_{H}y_{0}$, and therefore $w=v'^{-1}v\in\pstab_H\left(y_{0}\right)$.
	
	Conversely, suppose that (\ref{eq:EventADef}) and (\ref{eq:EventB_uDef})
	hold, and take an edge $y_{1}\overset{s}{\rightedge}y_{2}$ of $H$. Without loss of generality, assume that $y_2$ is not an internal vertex in the path, in $T$, from $y_0$ to $y_1$. Then, $sw^y_1\in \spaths_H(y_0)$, and so $\left(w^{y_{2}}\right)^{-1}sw^{y_{1}}\in \bp_H(y_0)$. 
	Note that $\left(sw^{y_{1}}\right)_{H}y_{0}=y_{2}=w_{H}^{y_{2}}y_{0}$.
	Thus, $\left(w^{y_{2}}\right)^{-1}sw^{y_{1}}\in\pstab_H\left(y_{0}\right)$.
	By (\ref{eq:EventADef}), $\left(sw^{y_{1}}\right)_{G}y_{0}=w_{G}^{y_{2}}y_{0}$.
	Consequently, (\ref{eq:EventB_uDef}) implies that $s_{G}y_{1}=y_{2}$,
	so $G$ contains the edge $y_{1}\overset{s}{\rightedge}y_{2}$.
\end{proof}
Next, for fixed $G\in\calG_{S}\left(n\right)$ and connected $H\in\PG_{S}\left(n\right)$,
we study
\begin{equation}
\Pr_{\pi\sim \U\left(\Sym\left(n\right)\right)}\left(H\subset G\pi\right)\,\,\text{.}\label{eq:connected-inclusion-event}
\end{equation}
We use Lemma \ref{lem:SubgraphCharacterization}. Fix some $y_0\in V(H)$ and a set of words $\{w^y\}_{y\in V(H)}$ as in the lemma. The probability that
\begin{equation}\label{eq:y0InCorrectClass}
\stab_{G\pi}(y_0)\cap \bp_H(y_0) = \pstab_H(y_0)
\end{equation} is $N_{G,\bp_{H}\left(y_{0}\right)}\left(\pstab_H\left(y_{0}\right)\right)$, since $\pi\left(y_{0}\right)$ is distributed uniformly
on $[n]$. Conditioned on the event \eqref{eq:y0InCorrectClass}, the probability that $$w_{G\pi}^yy_0=y\quad\forall y\in V(H)$$ is approximately
$n^{-\left(V\left(H\right)-1\right)}$. Lemma \ref{lem:ProbOfConnectedSubgraph}
below provides a more precise estimate of \eqref{eq:connected-inclusion-event}. Furthermore, the lemma shows
that (\ref{eq:connected-inclusion-event}) does not change by much
if we condition on the event $\pi\mid_{L}=f$ for a small subset $L$
of $\left[n\right]$ and an injective function $f\colon L\to\left[n\right]$.
\begin{lem}
	\label{lem:ProbOfConnectedSubgraph}Fix $G\in\calG_{S}\left(n\right)$,
	$H\in\PG_{S}\left(n\right)$, and suppose that $H$ is connected and has at least one edge. Let $y_{0}$ be an arbitrary vertex of $H$. Then
	\begin{equation}
	\Pr_{\pi\sim \U\left(\Sym\left(n\right)\right)}\left(H\subset G\pi\right)=\frac{\left(n-\left|V\left(H\right)\right|\right)!}{\left(n-1\right)!}N_{G,\bp_{H}\left(y_{0}\right)}\left(\pstab_H\left(y_{0}\right)\right)\,\,\text{.}\label{eq:ProbContainHConnected}
	\end{equation}
	Moreover, let $L\subset\left[n\right]\setminus V\left(H\right)$ and
	let $f\colon L\to\left[n\right]$ be an injective function. Then
	\begin{equation}
	\left|\Pr_{\pi\sim \U\left(\Sym\left(n\right)\right)}\left(H\subset G\pi\mid\pi\mid_{L}=f\right)-p\right|\le\frac{\left(n-|L|-\left|V(H)\right|\right)!}{\left(n-|L|-1\right)!}\cdot\frac{|L|\cdot\left|V(H)\right|}{n}\,\,\text{,}\label{eq:ProbContainHConnectedConditioned}
	\end{equation}
	where $\pi_{\mid L}$ denotes the restriction of $\pi$ to the set
	$L$ and 
	\[
	p=\frac{\left(n-\left|V\left(H\right)\right|-|L|\right)!}{\left(n-1-|L|\right)!}\cdot N_{G,\bp_{H}\left(y_{0}\right)}\left(\pstab_H\left(y_{0}\right)\right)\,\,\text{.}
	\]
\end{lem}
\begin{rem}\label{rem:ProbOfConnectedSubgraphSpecial}
	In the setting of Lemma \ref{lem:ProbOfConnectedSubgraph}, when $L\le O(n^{1/4})$ and $|V(H)|\le O(n^{1/4})$, \eqref{eq:ProbContainHConnectedConditioned} yields
	$$\Pr_{\pi\sim \U\left(\Sym\left(n\right)\right)}\left(H\subset G\pi\mid\pi\mid_{L}=f\right) = n^{-(V(H)-1)}\cdot \left(N_{G,\bp_{H}\left(y_{0}\right)}\left(\pstab_H\left(y_{0}\right)\right)\pm O\left(n^{-\frac 12}\right)\right). $$
\end{rem}

\begin{proof}[Proof of Lemma \ref{lem:ProbOfConnectedSubgraph}]
	Since (\ref{eq:ProbContainHConnected}) is a special case of (\ref{eq:ProbContainHConnectedConditioned}),
	we only prove the latter. Let $D=\left\{ \pi\in\Sym\left(n\right)\mid\pi\mid_{L}=f\right\} $,
	and note that the distribution of $\pi\sim \U\left(\Sym\left(n\right)\right)$,
	conditioned on the event $\pi\mid_{L}=f$, is the uniform distribution
	on $D$.
	
	Fix $y_{0}\in V\left(H\right)$ and write $V\left(H\right)=\left\{ y_{0},y_{1},\ldots,y_{|V\left(H\right)|-1}\right\} $. For every $0\le i\le |V(H)|-1$, fix a word $w^{y_i}\in \spaths_{H}(y_0)$ such that $w_H^{y_i}y_0=y_i$. 
	Let
	\[
	A=\left\{ \pi\in\Sym\left(n\right)\mid\stab_{G\pi}\left(y_{0}\right)\cap\bp{}_{H}\left(y_{0}\right)=\pstab_H\left(y_{0}\right)\right\} \,\,\text{.}
	\]
	For $0\le i\le\left|V\left(H\right)\right|-1$, let
	\[
	B_{i}=\left\{ \pi\in\Sym\left(n\right)\mid w_{G\pi}^{y_{i}}y_{0}=y_{i}\right\} \,\,\text{.}
	\]
	
	By Lemma \ref{lem:SubgraphCharacterization},
	\begin{equation}
	\Pr_{\pi\sim \U\left(D\right)}\left(H\subset G\pi\right)=\Pr_{\pi\sim \U\left(D\right)}\left(\pi\in A\cap\bigcap_{i=0}^{|V(H)|-1}B_{i}\right)\,\,\text{.}\label{eq:connected-subgraph-proof-first-formulation}
	\end{equation}
	Let 
	\[
	\overline{A}=\left\{ \pi\in A\mid w_{G\pi}^{y_{i}}y_{0}\notin L\text{ for all }1\le i\le\left|V(H)\right|-1\right\} \,\,\text{.}
	\]
	Observe that
	\[
	A\cap\bigcap_{i=1}^{|V(H)|-1}B_{i}\subset\overline{A}\,\,\text{.}
	\]
	Indeed, $\pi\in B_{i}$ implies that $w_{G\pi}^{y_{i}}y_{0}=y_{i}$,
	and in particular $w_{G\pi}^{y_{i}}y_{0}\notin L$. Hence, since $\overline A\subseteq A$,
	\begin{equation}
	A\cap\bigcap_{i=1}^{|V(H)|-1}B_{i}=\overline{A}\cap\bigcap_{i=1}^{|V(H)|-1}B_{i}\,\,\text{.}\label{eq:connected-subgraph-olA-vs-A}
	\end{equation}
	
	Next, we show that $\ol A\subset B_{0}$. It always holds that $w_{H}^{y_{0}}(y_{0})=y_{0}$,
	so $w^{y_{0}}\in\pstab_H\left(y_{0}\right)$. Consequently, $\pi\in\ol A$
	implies $w^{y_{0}}\in\stab_{G\pi}\left(y_{0}\right)$, which implies
	$\pi\in B_{0}$. Therefore, by (\ref{eq:connected-subgraph-proof-first-formulation})
	and (\ref{eq:connected-subgraph-olA-vs-A}), 
	\begin{equation}
	\Pr_{\pi\sim \U\left(D\right)}\left(H\subset G\pi\right)=\Pr_{\pi\sim \U\left(D\right)}\left(\pi\in\overline{A}\right)\cdot\prod_{i=1}^{\left|V\left(H\right)\right|-1}\Pr_{\pi\sim \U\left(D\right)}\left(\pi\in B_{i}\mid\pi\in\overline{A}\cap\bigcap_{j=0}^{i-1}B_{j}\right)\,\,\text{.}\label{eq:connected-subgraph-proof-0}
	\end{equation}
	We claim that 
	\begin{equation}
	\left|\Pr_{\pi\sim \U\left(D\right)}\left(\pi\in\overline{A}\right)-N_{G,\bp_{H}\left(y_{0}\right)}\left(\pstab_H\left(y_{0}\right)\right)\right|\leq\frac{\left|L\right|\cdot|V(H)|}{n}\label{eq:connected-subgraph-proof-1}
	\end{equation}
	and
	\begin{equation}
	\Pr_{\pi\sim \U\left(D\right)}\left(\pi\in B_{i}\mid\pi\in\overline{A}\cap\bigcap_{j=0}^{i-1}B_{j}\right)=\frac{1}{n-|L|-i}\quad\forall1\leq i\leq\left|V\left(H\right)\right|-1\,\,\text{,}\label{eq:connected-subgraph-proof-2}
	\end{equation}
	and thus (\ref{eq:ProbContainHConnectedConditioned}) follows from
	(\ref{eq:connected-subgraph-proof-0}).
	
	To prove (\ref{eq:connected-subgraph-proof-1}), let 
	\[
	W_{1}=\left\{ x\in[n]\mid\stab_{G}\left(x\right)\cap\bp{}_{H}\left(y_{0}\right)=\pstab_H\left(y_{0}\right)\right\} 
	\]
	and
	\[
	W_{2}=\left\{ x\in[n]\mid w_{G}^{y_{i}}x\in f\left(L\right)\text{ for some }1\le i\le\left|V(H)\right|-1\right\} \,\,\text{,}
	\]
	and note that $\pi\in\ol A$ if and only if $\pi y_{0}\in W_{1}\setminus W_{2}$.
	Also, observe that $\pi y_{0}$ is distributed uniformly on $\left[n\right]\setminus f\left(L\right)$
	when $\pi\sim \U\left(D\right)$. Hence, 
	\begin{align}
	\left|\Pr_{\pi\sim \U\left(D\right)}\left(\pi\in\overline{A}\right)-\frac{|W_{1}\setminus W_{2}|}{n}\right| & =\left|\Pr_{x\sim \U\left(\left[n\right]\setminus f\left(L\right)\right)}\left(x\in W_{1}\setminus W_{2}\right)-\frac{|W_{1}\setminus W_{2}|}{n}\right|\nonumber \\
	& =\left|\Pr_{x\sim \U\left(\left[n\right]\setminus f\left(L\right)\right)}\left(x\in W_{1}\setminus W_{2}\right)-\Pr_{x\sim \U\left([n]\right)}\left(x\in W_{1}\setminus W_{2}\right)\right|\nonumber \\
	& \le d_{\TV}\left(\U\left(\left[n\right]\setminus f\left(L\right)\right),\U\left(\left[n\right]\right)\right) & \text{by Fact \ref{fact:dTVMaxEvent}}\nonumber \\
	& =\frac{\left|L\right|}{n}\,\,\text{.}\label{eq:connected-subgraph-proof-3}
	\end{align}
	Now, 
	\[
	\left|W_{1}\right|=N_{G,\bp_{H}\left(y_{0}\right)}\left(\pstab_H\left(y_{0}\right)\right)\cdot n
	\]
	and 
	\[
	0\le\left|W_{2}\right|\le\left(|V(H)|-1\right)\cdot|L|\,\,\text{.}
	\]
	Thus 
	\[
	\left|\frac{\left|W_{1}\setminus W_{2}\right|}{n}-N_{G,\bp_{H}\left(y_{0}\right)}\left(\pstab_H\left(y_{0}\right)\right)\right|\le\frac{\left(|V(H)|-1\right)\cdot\left|L\right|}{n}\,\,\text{,}
	\]
	which implies (\ref{eq:connected-subgraph-proof-1}) by virtue of
	(\ref{eq:connected-subgraph-proof-3}) and the triangle inequality.
	
	We turn to proving (\ref{eq:connected-subgraph-proof-2}). First note
	that if $\pi\in\ol A$ then $w_{G\pi}^{y_{0}}y_{0},w_{G\pi}^{y_{1}}y_{0},\dotsc,w_{G\pi}^{y_{\left|V\left(H\right)\right|-1}}y_{0}$
	are distinct. Indeed, if $w_{G\pi}^{y_{i}}y_{0}=w_{G\pi}^{y_{j}}y_{0}$
	then $\left(w^{y_{j}}\right)^{-1}w^{y_{i}}\in\stab_{G}\left(y_{0}\right)\cap\bp_{H}\left(y_{0}\right)$,
	and thus $\left(w^{y_{j}}\right)^{-1}w^{y_{i}}\in\pstab_H(y_{0})$
	because $\pi\in\ol A$. Consequently,
	\[
	y_{j}=w_{H}^{y_{j}}y_{0}=w_{H}^{y_{i}}y_{0}=y_{i}\,\,\text{,}
	\]
	and so $i=j$ since $y_{0},\dotsc,y_{\left|V\left(H\right)\right|-1}$
	are distinct. 
	
	In particular, for $1\le i\le\left|V(H)\right|-1$, the event $\pi\in\ol A\cap\bigcap_{j=0}^{i-1}B_{j}$
	implies that 
	\[
	\pi^{-1}\left(w_{G}^{y_{i}}(\pi y_{0})\right)=w_{G\pi}^{y_{i}}y_{0}\notin\left\{ y_{0},\dotsc,y_{i-1}\right\} \,\,\text{,}
	\]
	and so
	\begin{equation}
	w_{G}^{y_{i}}(\pi y_{0})\notin\left\{ \pi y_{0},\dotsc,\pi y_{i-1}\right\} \,\,\text{.}\label{eq:SGraphVerticesDistinctDueToRecentEvents}
	\end{equation}
	Observe that, since 
	\[
	\pi\in B_{j}\iff w_{G}^{y_{j}}\left(\pi y_{0}\right)=\pi y_{j},
	\]
	the event $\pi\in\ol A\cap\bigcap_{j=0}^{i-1}B_{j}$ is determined
	solely by the restriction $\pi_{\mid\left\{ y_{0},\dotsc,y_{i-1}\right\} }$.
	In other words, for a permutation $\pi\in D$ we have 
	\begin{equation}
	\pi\in\ol A\cap\bigcap_{j=0}^{i-1}B_{j}\iff\pi_{\mid\left\{ y_{0},\dotsc,y_{i-1}\right\} }\in K\label{eq:connected-subgraph-proof-piAndV}
	\end{equation}
	where 
	\[
	K=\left\{ g\colon\left\{ y_{0},\dotsc,y_{i-1}\right\} \to\left[n\right]\setminus f(L)\mid g\text{ is injective and }\text{\ensuremath{\pi_{\mid y_{0},\dotsc,y_{i-1}}=g} implies }\text{\ensuremath{\pi\in}\ensuremath{\ol A\cap\bigcap_{j=0}^{i-1}B_{j}}}\right\} \,\,\text{.}
	\]
	Fix some $g\in K$, and consider a random permutation $\pi\sim \U\left(D\right)$,
	conditioned on $\pi_{\mid\left\{ y_{0},\dotsc,y_{i-1}\right\} }=g$.
	Note that $\pi y_{i}$ is distributed uniformly on the set $[n]\setminus\left(L\cup\left\{ \pi y_{0},\dotsc,\pi y_{i-1}\right\} \right)$.
	Since $\pi\in\overline{A}$, we have $w_{G}^{y_{i}}(\pi y_{0})\notin L$.
	Together with (\ref{eq:SGraphVerticesDistinctDueToRecentEvents}),
	it follows that
	\begin{align*}
	 \Pr_{\pi\sim \U\left(D\right)}\left(\pi\in B_{i}\mid\pi_{\mid\left\{ y_{0},\dotsc,y_{i-1}\right\} }=g\right)
	 &=\Pr_{\pi\sim \U\left(D\right)}\left(\pi y_{i}=w_{G}^{y_{i}}(\pi y_{0})\mid\pi_{\mid\left\{ y_{0},\dotsc,y_{i-1}\right\} }=g\right)\\
	& =\frac{1}{\left|[n]\setminus\left(L\cup\left\{ \pi y_{0},\dotsc,\pi y_{i-1}\right\} \right)\right|}=\frac{1}{n-L-i}\,\,\text{.}
	\end{align*}
	Finally, (\ref{eq:connected-subgraph-proof-2}) follows by virtue
	of (\ref{eq:connected-subgraph-proof-piAndV}).
\end{proof}
We next generalize Lemma \ref{lem:ProbOfConnectedSubgraph} to the
case in which $H$ is not necessarily connected, provided that it
does not have too many edges.
\begin{lem}
	\label{lem:ProbOfSubgraph}Fix $G\in\calG_{S}\left(n\right)$ and
	$H\in\PG_{S}\left(n\right)$. Denote the connected components of $H$
	by $C_{1},\ldots,C_{k}$, and suppose that every connected component has at least one edge. Fix $x_{i}\in V\left(C_{i}\right)$
	for each $1\le i\le k$. Suppose that
	\begin{equation}
	\left|E\left(H\right)\right|\le n^{\frac{1}{4}}\,\,\text{.}\label{eq:HSmallerThanRootn}
	\end{equation}
	Then,
	\[
	\Pr_{\pi\sim \U\left(\Sym\left(n\right)\right)}\left(H\subset G\pi\right)=n^{-r\left(H\right)}\prod_{i=1}^{k}N_{G,\bp_{H}\left(x_{i}\right)}\left(\pstab_H\left(x_{i}\right)\right)+O\left(n^{-r\left(H\right)-\frac{1}{4}}\right)\,\,\text{,}
	\]
	where the implied constant is universal.
\end{lem}

\begin{proof}
	Clearly, 
	\[
	\Pr_{\pi\sim \U\left(\Sym\left(n\right)\right)}\left(H\subset G\pi\right)=\prod_{i=1}^{k}\Pr_{\pi\sim \U\left(\Sym\left(n\right)\right)}\left(C_{i}\subset G\pi\mid\bigcup_{j=1}^{i-1}C_{j}\subset G\pi\right)\,\,\text{,}
	\]
	where by $C_{i}\subset G\pi$ we mean that $G\pi$ contains every
	edge in the connected component $C_{i}$ of $H$. To bound the $i$-th
	term of this product, let $L_{i}=\bigcup_{j=1}^{i-1}C_{j}$, and note
	that the event $\bigcup_{j=1}^{i-1}C_{j}\subset G\pi$ is determined
	solely by the restriction $\pi_{\mid L_{i}}$. Indeed, $G\pi$ contains
	the edge $a\overset{s}{\rightedge b}$ of $C_{i}$ if and only if
	$s_{G}\left(\pi a\right)=\pi b$. Hence there exists a set $D_{i}$
	of injective functions $f\colon L_{i}\to[n]$, such that 
	\[
	\bigcup_{j=1}^{i-1}C_{j}\subset G\pi\iff\pi_{\mid L_{i}}\in D_{i}\,\,\text{.}
	\]
	Thus, we can write 
	\begin{equation}
	\Pr_{\pi\sim \U\left(\Sym\left(n\right)\right)}\left(C_{i}\subset G\pi\mid\bigcup_{j=1}^{i-1}C_{j}\subset G\pi\right)=\sum_{f\in D_{i}}\Pr_{\pi\sim \U\left(\Sym(n)\right)}\left(\pi_{\mid L_{i}}=f\mid\pi_{\mid L_{i}}\in D_{i}\right)\cdot\Pr_{\pi\sim \U\left(\Sym\left(n\right)\right)}\left(C_{i}\subset G\pi\mid\pi_{\mid L_{i}}=f\right)\,\,\text{.}\label{eq:ProbOfSubgraph1}
	\end{equation}
	By Remark \ref{rem:ProbOfConnectedSubgraphSpecial},
	\begin{equation}
	\Pr_{\pi\sim \U\left(\Sym(n)\right)}\left(C_{i}\subset G\pi\mid\pi_{\mid L_{i}}=f\right)=n^{-\left(|C_{i}|-1\right)}\cdot\left(N_{G,\bp_{H}\left(x_{i}\right)}\left(\pstab_H\left(x_{i}\right)\right)\pm O\left(n^{-\frac{1}{2}}\right)\right)\,\,\text{,}\label{eq:ProbOfSubgraph2}
	\end{equation}
	where we used \eqref{eq:HSmallerThanRootn}, and the fact that both
	$|C_{i}|$ and $|L_{i}|$ are bounded from above by $2\left|E\left(H\right)\right|$.
	Now, (\ref{eq:ProbOfSubgraph1}) and (\ref{eq:ProbOfSubgraph2}) yield
	\[
	\Pr_{\pi\sim \U\left(\Sym\left(n\right)\right)}\left(C_{i}\subset G\pi\mid\bigcup_{j=1}^{i-1}C_{j}\subset G\pi\right)=n^{-\left(\left|C_{i}\right|-1\right)}\cdot\left(N_{G,\bp_{H}\left(x_{i}\right)}\left(\pstab_H\left(x_{i}\right)\right)\pm O\left(n^{-\frac{1}{2}}\right)\right)\,\,\text{.}
	\]
	Since $r\left(H\right)=\sum_{i=1}^{k}\left(\left|C_{i}\right|-1\right)$,
	and $N_{G,P}(Q)\le1$ for every $P$ and $Q$, it follows that 
	\begin{align*}
	\Pr_{\pi\sim \U\left(\Sym\left(n\right)\right)}\left(H\subset G\pi\right) & =\prod_{i=1}^{k}\Pr_{\pi\sim \U\left(\Sym\left(n\right)\right)}\left(C_{i}\subset G\pi\mid\bigcup_{j=1}^{i-1}C_{j}\subset G\pi\right)\\
	& =n^{-r(H)}\left(\prod_{i=1}^{k}N_{G,\bp_{H}\left(x_{i}\right)}\left(\pstab_H\left(x_{i}\right)\right)+O\left(\sum_{i=1}^{k}\binom{k}{i}\cdot n^{-\frac{i}{2}}\right)\right)\,\,\text{,}
	\end{align*}
	which yields the lemma since $k\le\left|E\left(H\right)\right|\le n^{\frac{1}{4}}$. 
\end{proof}
\begin{cor}
	\label{cor:ProbOfSubgraphDiff}Let $G_{0},G_{1}\in\calG_{S}\left(n\right)$,
	and take $H\in\PG_{S}\left(n\right)$. Suppose that $\left|E(H)\right|\le n^{\frac{1}{4}}$, and that every connected component of $H$ has at least one edge.
	Then,
	\[
	\left|\Pr_{\pi\sim \U\left(\Sym\left(n\right)\right)}\left(H\subset G_{0}\pi\right)-\Pr_{\pi\sim \U\left(\Sym\left(n\right)\right)}\left(H\subset G_{1}\pi\right)\right|\le n^{-r\left(H\right)}\cdot\left|E(H)\right|\cdot d_{\TV}\left(N_{G_{0},B_{2\left|E(H)\right|}},N_{G_{1},B_{2\left|E(H)\right|}}\right)+cn^{-r\left(H\right)-\frac{1}{4}}\,\,\text{,}
	\]
	for a universal constant $c>0$.
\end{cor}

\begin{proof}
	Suppose that $H$ has $k$ connected components, and let $x_{1},\ldots,x_{k}$
	be representative vertices of these components. For $j\in\left\{ 0,1\right\} $,
	Lemma \ref{lem:ProbOfSubgraph} yields 
	\begin{equation}
	\Pr_{\pi\sim \U\left(\Sym\left(n\right)\right)}\left(H\subset G_{j}\pi\right)=n^{-r\left(H\right)}\prod_{i=1}^{k}N_{G_{j},\bp_{H}\left(x_{i}\right)}\left(\pstab_H\left(x_{i}\right)\right)+O\left(n^{-r\left(H\right)-\frac{1}{4}}\right)\,\,\text{.}\label{eq:ProbHSubgraphNoTrivials}
	\end{equation}
	Now,
	
	\begin{align}
	& \left|\prod_{i=1}^{k}N_{G_{0},\bp_{H}\left(x_{i}\right)}\left(\pstab_H\left(x_{i}\right)\right)-\prod_{i=1}^{k}N_{G_{1},\bp_{H}\left(x_{i}\right)}\left(\pstab_H\left(x_{i}\right)\right)\right|\nonumber \\
	& \le d_{\TV}\left(\prod_{i=1}^{k}N_{G_{0},\bp_{H}\left(x_{i}\right)},\prod_{i=1}^{k}N_{G_{1},\bp_{H}\left(x_{i}\right)}\right) & \text{by Fact \ref{fact:dTVMaxEvent}}\nonumber \\
	& \le\sum_{i=1}^{k}d_{\TV}\left(N_{G_{0},\bp_{H}\left(x_{i}\right)},N_{G_{1},\bp_{H}\left(x_{i}\right)}\right) & \text{by Fact \ref{fact:dTVProduct}}\nonumber \\
	& \le k\cdot d_{\TV}\left(N_{G_{0},B_{2\left|E(H)\right|}},N_{G_{1},B_{2\left|E(H)\right|}}\right) & \text{by Corollary \ref{cor:d_TVIncreasingInP}}\nonumber \\
	& \le\left|E(H)\right|\cdot d_{\TV}\left(N_{G_{0},B_{2\left|E(H)\right|}},N_{G_{1},B_{2\left|E(H)\right|}}\right) & \text{since }k\le\left|E(H)\right|.\label{eq:ProbOfSubgraphProductInequality}
	\end{align}
	The claim follows from (\ref{eq:ProbHSubgraphNoTrivials}) and (\ref{eq:ProbOfSubgraphProductInequality}).
\end{proof}
Lemma \ref{lem:LSMNegativeMainTechnical} now follows from Lemma \ref{lem:BoundedNumberOfRuns}
and Corollary \ref{cor:ProbOfSubgraphDiff}.
\begin{proof}[Proof of Lemma \ref{lem:LSMNegativeMainTechnical}]
	Denote $\delta=d_{\TV}\left(N_{G_{0},B_{2q}},N_{G_{1},B_{2q}}\right)$.
	Then, 
	
	\begin{align*}
	d_{\TV}\left(\theta_{0},\theta_{1}\right) & =\frac{1}{2}\sum_{H\in\calR_{\cM}\left(n\right)}\left|\Pr_{\pi\sim \U\left(\Sym(n)\right)}\left(H_{\cM,G_{0}\pi}=H\right)-\Pr_{\pi\sim \U\left(\Sym(n)\right)}\left(H_{\cM,G_{1}\pi}=H\right)\right|\\
	& =\frac{1}{2}\sum_{H\in\calR_{\cM}\left(n\right)}\left|\Pr_{\pi\sim \U\left(\Sym(n)\right)}\left(H\subset G_{0}\pi\right)-\Pr_{\pi\sim \U\left(\Sym(n)\right)}\left(H\subset G_{1}\pi\right)\right|\\
	&
	\le\frac{1}{2}\sum_{H\in\calR_{\cM}(n)}\left(n^{-r\left(H\right)}\cdot\left|E(H)\right|\cdot d_{\TV}\left(N_{G_0,B_{2|E(H)|}},N_{G_1,B_{2|E(H)|}}\right)+cn^{-r\left(H\right)-\frac{1}{4}}\right) & \text{by Corollary \ref{cor:ProbOfSubgraphDiff}}\\
	& \le\frac{1}{2}\sum_{H\in\calR_{\cM}(n)}\left(n^{-r\left(H\right)}\cdot\left|E(H)\right|\cdot\delta+cn^{-r\left(H\right)-\frac{1}{4}}\right) & \text{by Corollary \ref{cor:d_TVIncreasingInP}}\\
	& \le\frac{1}{2}\sum_{H\in\calR_{\cM}(n)}\left(n^{-r\left(H\right)}\cdot q\cdot\delta+cn^{-r\left(H\right)-\frac{1}{4}}\right)\\
	& \le\frac{1}{2}\sum_{r=0}^{q}\left|\calR_{\cM,r}\left(n\right)\right|\cdot\left(n^{-r}\cdot q\cdot\delta+cn^{-r-\frac{1}{4}}\right)\\
	& \le\frac{1}{2}\sum_{r=0}^{q}\binom{q}{r}\left(2q\right)^{q-r}\cdot\left(q\delta+cn^{-\frac{1}{4}}\right) & \text{by Lemma \ref{lem:BoundedNumberOfRuns}}\\
	& \le\frac{1}{2}\sum_{r=0}^{q}\left(2q\right)^{q}\left(q\delta+cn^{-\frac{1}{4}}\right)\\
	& \le\left(2q\right)^{q+2}\left(\delta+cn^{-\frac{1}{4}}\right)\,\,\text{.}
	\end{align*}
\end{proof}

\section*{Acknowledgements}
O.B.\ has received funding from the European Research Council (ERC) under the European Union's Horizon 2020 research and innovation programme (grant agreement No. 803711).
A.L.\ is supported by a grant of the Institute for Advanced Study and by the European Research Council (ERC) under the European Union Horizon 2020 research and innovation program (Grant No. 692854).
J.M.\ is partially supported by NSF grant CCF-1814603.

J.M.\ wishes to thank Venkatesan Guruswami for a helpful conversation, and Nicolas Resch and Jo\~{a}o Ribeiro for useful comments.

\bibliographystyle{plain}
\bibliography{testability}

\appendix

\section{An explicit description of certain systems of relations \label{app:Equations}}

\subsection{\label{appendix:abels}A system of relations which is testable but not stable}

Here we describe instable testable systems of relations using Theorem \ref{thm:MainPositivePrime}, as discussed in Section \ref{subsec:testableInstable}.
Fix a prime number $p$.
Let

\[
S_p=\left\{ d_{2},d_{3},s_{12},s_{13},s_{23},s_{24},s_{34}\right\} \,\,\text{,}
\]
and consider the following system of relations over the alphabet $S_p$:

\begin{align*}
	E_{p}= & \left\{ d_{2}d_{3}=d_{3}d_{2},s_{12}s_{34}=s_{34}s_{12}\right\} \cup\\
	& \left\{ s_{23}s_{12}=s_{12}s_{23}s_{13},s_{34}s_{23}=s_{23}s_{34}s_{24}\right\} \cup\\
	& \left\{ s_{13}s_{12}=s_{12}s_{13},s_{13}s_{23}=s_{23}s_{13},s_{24}s_{23}=s_{23}s_{24}\right\} \cup\\
	& \left\{ s_{24}s_{34}=s_{34}s_{24},s_{13}s_{24}=s_{24}s_{13}\right\} \cup\\
	& \left\{ s_{12}d_{2}=d_{2}s_{12}^{p},s_{12}d_{3}=d_{3}s_{12}\right\} \cup\\
	& \left\{ d_{2}s_{23}=s_{23}^{p}d_{2},s_{23}d_{3}=d_{3}s_{23}^{p}\right\} \cup\\
	& \left\{ s_{34}d_{2}=d_{2}s_{34},d_{3}s_{34}=s_{34}^{p}d_{3}\right\}.
\end{align*}

As explained below, this system is testable by Theorem \ref{thm:MainPositivePrime}, yet instable due to \cite[Theorem 1.3(ii)]{BLT}. 

It was shown in \cite{Abels} that $\langle S_{p}\mid E_{p}\rangle$ is a presentation of \emph{Abels' group} $A_p$, which is defined as follows:
\[
A_{p}=\left\{ \left(\begin{array}{cccc}
1 & * & * & *\\
0 & p^{m} & * & *\\
0 & 0 & p^{n} & *\\
0 & 0 & 0 & 1
\end{array}\right)\subset\GL_{4}\ZZ\left[1/p\right]\mid m,n\in\ZZ\right\} \,\,\text{.}
\]
Here $\ZZ\left[1/p\right]$ is the ring of rational number whose denominator
is a power of $p$, and $\GL_{4}\ZZ\left[1/p\right]$ is the group
of $4\times4$ matrices with entries in $\ZZ\left[1/p\right]$ and
determinant $\pm p^{l}$, $l\in\ZZ$.

An isomorphism $\langle S_{p}\mid E_{p}\rangle\overset{\sim}{\longrightarrow}A_{p}$
is given by
\[
d_{2}\mapsto\left(\begin{array}{cccc}
	1\\
	& p\\
	&  & 1\\
	&  &  & 1
\end{array}\right),\quad d_{3}\mapsto\left(\begin{array}{cccc}
	1\\
	& 1\\
	&  & p\\
	&  &  & 1
\end{array}\right)
\]
\[
s_{ij}\mapsto I+e_{ij}\,\,\text{,}
\]
where $e_{ij}$ is the $4\times4$ matrix with $1$ on the $\left(i,j\right)$-entry
and $0$ elsewhere. 

The group $A_{p}$ is solvable because it is contained in the group
of upper triangular matrices. Hence $A_p$ is  amenable. In other words, $E_p$ satisfies the hypothesis of Theorem \ref{thm:MainPositivePrime}, and so it is testable.

On the other hand, \cite[Theorem 1.3(ii)]{BLT} characterizes the systems of relations that are stable, among those that satisfy the hypothesis of Theorem \ref{thm:MainPositivePrime}.
By this characterization (which is given in terms of \emph{invariant random subgroups} and \emph{cosoficity}), $E_p$ is instable \cite[Corollary 8.7]{BLT}.

\subsection{A non-testable system of relations}
\label{appendix:SL}

Let $m\geq3$.
Here, as discussed in Section \ref{subsec:non-testable}, we use Theorem \ref{thm:MainNegativePrime} to provide
a non-testable system of relations $E_{m}$, over an
alphabet $S_{m}^\pm$, such that $\Gamma\left(E_{m}\right)\coloneqq\langle S_{m}\mid E_{m}\rangle\cong\SL_{m}\ZZ$:

\[
S_{m}=\left\{ s_{ij}\mid i,j\in\left[m\right],i\neq j\right\} \,\,\text{,}
\]
\begin{align*}
	E_{m}= & \left\{ s_{ij}s_{kl}=s_{kl}s_{ij}\mid i,j,k,l\in\left[m\right],j\neq k,i\neq l\right\} \cup\\
	& \left\{ s_{ij}s_{jk}=s_{ik}s_{jk}s_{ij}\mid i,j,k\in\left[m\right],i\neq j,j\neq k,k\neq i\right\} \cup\\
	& \left\{ \left(s_{12}s_{21}^{-1}s_{12}\right)^{4}=1\right\} \,\,\text{.}
\end{align*}

By \cite[Corollary 10.3]{MilnorKTheory}, $\langle S_{m}\mid E_{m}\rangle$ is a presentation of the group $\SL_{m}\ZZ$. The isomorphism given by
\[
s_{ij}\mapsto I+e_{ij}\,\,\text{,}
\]
where $e_{ij}$ is the $m\times m$ matrix with $1$ on the $\left(i,j\right)$-entry
and $0$ elsewhere. The group $\SL_m\ZZ$ is well known to satisfy Property $\T$, and has infinitely many finite quotients. This means that $E_m$ satisfies the hypothesis of Theorem \ref{thm:MainNegativePrime}, and therefore, it is not testable.

\section{A note on sampling inverses of permutations\label{app:AvoidingInverses}}

For a given permutation $\sigma\in\Sym\left(n\right)$ and $x\in\left[n\right]$,
our model assumes that an algorithm $\calM$ can read $\sigma x$ by
making a single query. We also allow $\calM$ to read $\sigma^{-1}x$
with a single query. It is also natural to study testability in a
model that allows to read $\sigma x$, but not $\sigma^{-1}x$, with
a single query. Here we describe how every testable (resp. stable)
system gives rise to a closely related system which is testable (resp.
stable) even under the model where sampling inverses is not allowed (see Definitions \ref{def:testable-equations} and \ref{def:StableRelations}).

Let $E$ be a system of equations over $S^{\pm}$. We say the $E$
is \emph{inverseless} if all equations in $E$ are over the alphabet
$S$, that is, letters from $S^{-1}$ are not used.
For example, $E_2^{\comm}=\{\eq{XY=YX}\}$ is inverseless.
Thus, $\SAS_{k}^{E_2^{\comm}}$, $k\in\NN$, can be implemented without sampling
inverses (simply by checking whether $\eq{XY}x_j=\eq{YX}x_j$ rather than
$\eq{X^{-1}Y^{-1}XY}x_j=x_j$ in Line 2 of Algorithm \ref{alg:SAS}).
Similarly, $\SAS_{k}^{E}$ can be implemented without sampling inverses whenever $E$ is inverseless.

In general, $E$ gives rise to an inverseless system of equations
$E'$ as follows. Extend the alphabet by defining $\overline{S}=\left\{ \overline{s}_{1},\dotsc,\overline{s}_{d}\right\} $
and $S'=S\cup\overline{S}$. Define a system of equations $\tilde{E}$
over $S'$ by starting from $E$ and replacing each occurrence of
each $s_{i}^{-1}$, $1\leq i\leq d$, by $\overline{s}_{i}$. Finally,
set $E'=\tilde{E}\cup\underbrace{\left\{ s_{i}\overline{s}_{i}=1\mid1\leq i\leq d\right\} }_{\eqqcolon E_{\inv}}$.
Then $E'$ is an inverseless system of equations, i.e., it is a system
of equations over $S'$ rather than $\left(S'\right)^{\pm}$. Furthermore,
$E$ is stable if and only if $E'$ is stable because the groups $\Gamma\left(E\right)$
and $\Gamma\left(E'\right)$ are isomorphic and by Proposition \ref{prop:stability-group-property}.

Moreover, by Proposition \ref{prop:testability-group-property}, $E$
is testable if and only if $E'$ is. Additionally, the algorithm $\LSM_{k,P,\delta}^{E}$ (see Section \ref{subsec:intro-LSM})
can be altered to avoid sampling inverses as follows:
First run $\SAS_{k'}^{E_{\inv}}$ for
a large enough $k'$, and reject if $\SAS_{k'}^{E_{\inv}}$ rejects.
Otherwise, run $\LSM_{k,P,\delta}^{E}$ without sampling inverses
by replacing each query of $s_{i}^{-1}x$ by a query of $\overline{s}_{i}x$,
and accept if $\LSM_{k,P,\delta}^{E}$ accepts.


\section{A review of free groups and group presentations\label{app:FreeGroupAndPresentations}}

Here we give a brief introduction to free groups and their universal
property, and to group presentations (see \cite{Johnson}).

Let $S=\left\{ s_{1},\dotsc,s_{d}\right\} $ be a set of letters.
Write $S^{-1}=\left\{ s_{1}^{-1},\dotsc,s_{d}^{-1}\right\} $ for
the set of formal inverses of the letters in $S$, and let $S^{\pm}\coloneqq S\cup S^{-1}$.
We define $\left(s_{i}^{-1}\right)^{-1}\coloneqq s_{i}$, and so $s\mapsto s^{-1}$
becomes an involution on $S^{\pm}$. Write $w_{1}\sim w_{2}$ for
words $w_{1}$ and $w_{2}$ over $S^{\pm}$ if there is a sequence
of words $w_{1}=u_{0},\dotsc,u_{n}=w_{2}$ such that $u_{i+1}$ is
obtained from $u_{i}$ by adding or removing a \emph{null subword}, i.e.,
a subword of the form
$ss^{-1}$, $s\in S^{\pm}$. For example, $s_{1}s_{2}s_{2}^{-1}s_{1}\sim s_{1}s_{1}s_{3}^{-1}s_{3}$.
A word $w$ over the alphabet $S^{\pm}$ is \emph{reduced} if it does
not contain a null subword as above. Every
word $w$ over $S^{\pm}$ is equivalent to a unique reduced word called
the \emph{reduced form} of $w$. For words $w_{1},w_{2},w'_{1},w'_{2}$
over $S^{\pm}$ such that $w_{1}\sim w'_{1}$ and $w_{2}\sim w'_{2}$,
we have $w_{1}w_{2}\sim w'_{1}w'_{2}$. In particular, the reduced
forms of $w_{1}w_{2}$ and $w'_{1}w'_{2}$ are equal.

The \emph{free group} $F_{S}$ over $S$ is the set of reduced words
over $S^{\pm}$, endowed with the following multiplication operation:
for $w_{1},w_{2}\in F_{S}$, $w_{1}\cdot w_{2}$ is the unique reduced
word equivalent to the concatenation $w_{1}w_{2}$. We sometimes abuse
notation and view a word $w$ over $S^{\pm}$, not necessarily reduced,
as an element of $F_{S}$. We do so only when making statements where
only the equivalence class of $w$ matters.

It is worth noting the following alternative way to define $F_S$, although we are not using it in this paper.
One can realize the free group $F_S$ as the group of
$\sim$-equivalence classes of
(not necessarily reduced) of words over $S^{\pm}$, with the multiplication defined by $[w_1]\cdot [w_2] \coloneqq [w_1w_2]$, where $[w]$ denotes the $\sim$-equivalence class of a word $w$.
This point of view also makes it easy to define \emph{group presentations}, i.e., to define the group $\langle S \mid R\rangle$ for a set of words $R$ over $S^\pm$: We define an equivalence class $\sim_R$ on the set of words over $S^\pm$ just as we defined $\sim$, only that we consider all subwords of the form $vr^{\pm 1} v^{-1}$ to be null subwords, for every word $v$ over $S^{\pm}$ and $r\in R$, in a addition to the null subwords $ss^{-1}$, $s\in S$. We then let $\langle S \mid R\rangle$ be the group of $\sim_R$-equivalence classes, with the multiplication of classes defined as in $F_S$ above.
We now go back to thinking of $F_S$ as the set of reduced words and define $\langle S \mid R\rangle$ in a different, equivalent, way.

The free group $F_{S}$ has the following important property, known
as \emph{the universal property of $F_{S}$}: for every group $\Gamma$
and function $f\colon S\to\Gamma$, there is exactly one group homomorphism
$\tilde{f}\colon F_{S}\to\Gamma$ such that $\tilde{f}\left(s_{i}\right)=f\left(s_{i}\right)$
for all $1\le i\le d$. The map $f\mapsto\tilde{f}$ is a bijection
from the set of functions $S\to\Gamma$ to the set of homomorphisms
$F_{S}\to\Gamma$.

In particular, for a group $\Gamma$ generated by $\gamma_{1},\dotsc,\gamma_{d}\in\Gamma$,
there is a unique homomorphism $\pi\colon F_{S}\to\Gamma$ such that
$\pi\left(s_{i}\right)=\gamma_{i}$ for each $1\leq i\leq d$.
The map $\pi$ is surjective
because its image $\pi\left(F_{S}\right)$ contains a generating set
for $\Gamma$. That is, every group $\Gamma$ generated by $d$ elements
is a quotient of $F_{S}$. For example, $\Gamma_{0}\coloneqq\ZZ^{2}$
is a quotient of $F_{\left\{ s_{1},s_{2}\right\} }$ as exhibited
by the unique homomorphism $\pi_{0}\colon F_{\left\{ s_{1},s_{2}\right\} }\to\ZZ^{2}$
sending $s_{1}\mapsto\left(1,0\right)$ and $s_{2}\mapsto\left(0,1\right)$.

The kernel $\ker\pi$ of $\pi\colon F_{S}\to\Gamma$ is a normal subgroup
of $F_{S}$. It is often useful to have a subset $R$ of $\ker\pi$
such that $\ker\pi=\lla R\rra$. Here $\lla R\rra$ denotes the normal
closure of $R$ in $F_{S}$, that is, the smallest normal subgroup
of $F_{S}$ that contains $R$. Concretely, $\lla R\rra$ consists
of all elements of the form $\prod_{i=1}^{m}v_{i}r_{i}^{\eps_{i}}v_{i}^{-1}$,
where $m\geq0$, $v_{i}\in F_{S}$, $r_{i}\in R$ and $\eps_{i}\in\left\{ \pm1\right\} $.
In the cases of interest of this paper, $\ker\pi=\lla R\rra$ where
$R$ is a finite set. For $\pi_{0}\colon F_{\left\{ s_{1},s_{2}\right\} }\to\ZZ^{2}$
as in the example, it can be shown that $\ker\pi_{0}=\lla R_{0}\rra$
for $R_{0}=\left\{s_{1}^{-1}s_{2}^{-1}s_{1}s_{2}\right\} $.

The surjective homomorphism $\pi\colon F_{S}\to\Gamma$ gives rise
to an isomorphism $F_{S}/\ker\pi\overset{\sim}{\longrightarrow}\Gamma$.
If $\ker\pi=\lla R\rra$ for $R\subset\ker\pi$, we write
$\langle S\mid R\rangle$ for the quotient group $F_{S}/\ker\pi$
and say that the group $\langle S\mid R\rangle$, which is isomorphic
to $\Gamma$, is a \emph{presentation} of $\Gamma$. For a system
of relations $E$, we write $\langle S\mid E\rangle$ for $\langle S\mid R_{E}\rangle$
(where $R_E$ is as in the introduction).

For example,
\[
\langle s_{1},s_{2}\mid s_{1}s_{2}=s_{2}s_{1}\rangle=\langle s_{1},s_{2}\mid s_{1}^{-1}s_{2}^{-1}s_{1}s_{2}\rangle\cong\ZZ^{2}\,\,\text{.}
\]

The group $\langle S\mid R\rangle$ is a \emph{finite presentation}
if $S$ and $E$ are finite sets. In this paper we also write $\Gamma\left(E\right)$
for $\langle S\mid E\rangle$ when the set $S$ is understood from
the context.

\end{document}